%% file: paper.tex
\documentclass[preprint,12pt]{elsarticle}

\usepackage{amssymb}
 \usepackage{amsthm}

\usepackage{amsmath,amsfonts}
\usepackage{tikz}
\usepackage{bm}
\usepackage{scrextend}
\usepackage{url,hyperref,multirow,complexity}
\usepackage[capitalize]{cleveref}
\usetikzlibrary{calc,backgrounds,fit,shapes,arrows,decorations.markings,patterns,positioning}

\definecolor{mylila}{RGB}{190,186,218}
\definecolor{myyellow}{RGB}{255,255,179}
\definecolor{mygreen}{RGB}{141,211,199}

\tikzstyle{vertex}=[draw, circle, inner sep = 2pt]
\tikzstyle{terminal}=[draw, inner sep = 2pt]
\tikzstyle{Bvertex}=[draw, circle, fill, inner sep = 2pt]
\tikzstyle{Avertex}=[draw, fill, inner sep = 2pt]
\tikzstyle{white-vertex}=[trapezium, draw, trapezium left angle=60,
  trapezium right angle=-60]

\input{fancyProblem}

\usepackage{caption}
 \newtheorem{theorem}{Theorem}
 \newtheorem{numberedclaim}{Claim}[]
 \newtheorem{observation}[theorem]{Observation}
 \newtheorem{definition}[theorem]{Definition}
 \newtheorem{lemma}[theorem]{Lemma}
 \newtheorem{remark}[theorem]{Remark}

\crefname{paragraph}{paragraph}{paragraphs}
\Crefname{paragraph}{Paragraph}{Paragraphs}
\crefname{numberedclaim}{claim}{claims}
\Crefname{numberedclaim}{Claim}{Claims}

\newenvironment{claimproof}[1]{\begin{proof}[Proof of \Cref{#1}.]}{\end{proof}}

\newcommand{\largeWeight}{\ensuremath{4}}

\newcommand{\even}{even}
\newcommand{\odd}{odd}

\newcommand{\witness}{\bm{y}}
\newcommand{\wof}[1]{\witness_{#1}}
\newcommand{\wprime}{\witness'}
\newcommand{\wone}{\witness^1}
\newcommand{\wtwo}{\witness^2}
\newcommand{\wthree}{\witness^3}
\newcommand{\woneof}[1]{\wone_{#1}}
\newcommand{\wtwoof}[1]{\wtwo_{#1}}
\newcommand{\wthreeof}[1]{\wthree_{#1}}

\newcommand{\wpone}{\witness^{P, \operatorname{\odd}}}

\newcommand{\wptwo}{\witness^{P, \operatorname{\even}}}

\newcommand{\wceven}{\witness^{C, \operatorname{\even}}}
\newcommand{\wcevenof}[1]{\wceven_{#1}}
\newcommand{\wcone}[1][]{\witness^{C, 1}_{#1}}
\newcommand{\wctwo}[1][]{\wceven_{#1}}
\newcommand{\wcthree}[1][]{\witness^{C, 2}_{#1}}
\newcommand{\wconeof}[1]{\witness^{C, 1}_{#1}}
\newcommand{\wctwoof}[1]{\witness^{C, \operatorname{\even}}_{#1}}
\newcommand{\wcthreeof}[1]{\witness^{C, 2}_{#1}}
\newcommand{\wcmtwo}{\witness^{C, 0}}
\newcommand{\wcmtwoof}[1]{\wcmtwo_{#1}}
\newcommand{\wcoeven}{\witness^{C_1, \operatorname{\even}}}

\newcommand{\wcoone}[1][]{\witness^{C_1, 1}_{#1}}

\newcommand{\wcothree}[1][]{\witness^{C_1, 2}_{#1}}

\newcommand{\wcomtwo}{\witness^{C_1, 0}}

\newcommand{\wctone}[1][]{\witness^{C_2, 1}_{#1}}
\newcommand{\wctthree}[1][]{\witness^{C_2, 2}_{#1}}

\newcommand{\weone}[1][]{\witness^{e, 1}_{#1}}
\newcommand{\wetwo}[1][]{\witness^{e, 2}_{#1}}
\newcommand{\weeven}[1][]{\witness^{e, \operatorname{\even}}_{#1}}

\newcommand{\vote}{\mathsf{vote}}
\newcommand{\weightedvote}{\widetilde{\mathsf{vote}}}
\newcommand{\wDhat}{\neu{\hat{\witness}^{D}}}
\newcommand{\wDhatof}[1]{\neu{\hat{\witness}^{D}_{#1}}}

\usepackage{algorithm}
\usepackage[noend]{algpseudocode}
\algnewcommand\algorithmicinput{\textbf{Input:}}
\algnewcommand\algorithmicoutput{\textbf{Output:}}
\algnewcommand\Input{\item[\algorithmicinput]}
\algnewcommand\Output{\item[\algorithmicoutput]}
\algnewcommand\algorithmicgoto{\textbf{GoTo }}
\algnewcommand\GoTo{\item[\algorithmicgoto]}
\algnewcommand\algorithmicforeach{\textbf{for each}}
\algdef{S}[FOR]{ForEach}[1]{\algorithmicforeach\ #1\ \algorithmicdo}
\newcommand{\wpm}{\textsc{Popular Matching with Weighted Voters}}
\newcommand{\vect}[1]{\bm{{#1}}}
\newcommand{\vectof}[2]{\bm{#1}_{#2}}
\newcommand{\witnessc}{\witness^C}
\newcommand{\witnesscof}[1]{\vectof{\witnessc}{#1}}

\newcommand{\Ghat}{\neu{\hat{F}}}

\newcommand{\eaf}{a_2}
\newcommand{\eafh}{a_1}
\newcommand{\ebfour}{b_2}
\newcommand{\ebfh}{b_1}
\newcommand{\eato}{a_4}
\newcommand{\eatoh}{a_3}
\newcommand{\ebto}{b_4}
\newcommand{\ebtoh}{b_3}
\newcommand{\eatt}{a_6}
\newcommand{\eatth}{a_5}
\newcommand{\ebtt}{b_6}
\newcommand{\ebtth}{b_5}
\newcommand{\eap}{a_8}
\newcommand{\eaph}{a_7}
\newcommand{\ebp}{b_8}
\newcommand{\ebph}{b_7}
\newcommand{\eas}{a_{10}}
\newcommand{\ebs}{b_{10}}
\newcommand{\eaedge}{a_9}
\newcommand{\ebe}{b_9}

\newcommand{\LogicTRUE}{\mathsf{true}}
\newcommand{\LogicFALSE}{\mathsf{false}}

\newcommand{\hcp}{H^{\mathcal{C} + \mathcal{P}}}
\newcommand{\hletwo}{H^{\le 2}}
\newcommand{\hct}{H^{\mathcal{C} + \mathcal{T}}}

\usepackage{xspace}
\usepackage[textwidth=35mm,obeyFinal,colorinlistoftodos,textsize=tiny]{todonotes}

\newcommand{\unnamedGadget}{6-path gadget}

\crefname{observation}{Observation}{Observations}
\crefname{lemma}{Lemma}{Lemmas}
\crefname{claim}{Claim}{Claims}

\usepackage{comment}
\newcommand{\new}[1]{{\color{black}#1}}
\newcommand{\neu}[1]{{\color{black}#1}}
\newcommand{\Jmatching}{\new{M_\mathcal{J}}}
\newcommand{\Jmatchingprime}{\new{M_\mathcal{J}'}}
\newcommand{\oldJmatching}{{M_\mathcal{J}}}
\newcommand{\oldJmatchingprime}{{M_\mathcal{J}'}}
\newcommand{\Jmatchingpp}{\new{M_\mathcal{J}''}}
\newcommand{\oldJmatchingpp}{{M_\mathcal{J}''}}

\newcommand{\Imatching}{\new{M_{\mathcal{I}}}}
\newcommand{\oldImatching}{{M_{\mathcal{I}}}}
\newcommand{\Imatchingprime}{\new{M_\mathcal{I}'}}
\newcommand{\oldImatchingprime}{{M_\mathcal{I}'}}

\newcommand{\no}{\texttt{no}}

\begin{document}

\begin{frontmatter}


\tnotetext[label1]{\'Agnes Cseh was supported by the J\'anos Bolyai Research Fellowship.
Klaus Heeger was supported by DFG Research Training Group 2434 ``Facets of Complexity'' and DFG project FPTinP (NI 369/16).}
\author[inst1]{Klaus Heeger}
\address[inst1]{Technische Universit\"at Berlin, Algorithmics and Computational Complexity, Germany}
\author[inst2,inst3]{\'Agnes Cseh}
\address[inst2]{Institute of Economics, HUN-REN Centre for Economic and Regional Studies, Budapest, Hungary}
\address[inst3]{Department of Mathematics, University of Bayreuth, Germany}

\title{Popular matchings with weighted voters\tnoteref{label1}}


 %\author{Klaus Heeger}
 %\author[label2]{\'Agnes Cseh}
% \affiliation[label2]{organization={Institute of Economics, Centre for Economic and Regional Studies},
%             addressline={T\'oth K\'alm\'an utca 4.},
%             city={Budapest},
%             postcode={1097},
%             country={Hungary}}

%\author{}

%\affiliation{organization={},%Department and Organization

\begin{abstract}
In the \textsc{Popular Matching} problem, we are given a bipartite graph~$G = (A \cup B, E)$ and for each vertex $v\in A\cup B$, strict preferences over the neighbors of~$v$.
 Given two matchings~$M$ and $M'$, matching $M$ is more popular than $M'$ if the number of vertices preferring $M$ to $M'$ is larger than the number of vertices preferring $M'$ to~$M$.
 A matching~$M$ is called \emph{popular} if there is no matching~$M'$ that is more popular than~$M$.
 
 We consider a natural generalization of \textsc{Popular Matching} where every vertex has a weight. Then, we call a matching~$M$ more popular than matching~$M'$ if the weight of vertices preferring $M$ to $M'$ is larger than the weight of vertices preferring $M'$ to~$M$.
 For this case, we show that it is \NP-hard to find a popular matching. Our main result is a polynomial-time algorithm that delivers a popular matching or a proof for it\new{s} non-existence in instances where all vertices on one side have weight $c$ for some $c > 3$ and all vertices on the other side have weight~1.

\end{abstract}




\end{frontmatter}


%\linenumbers

\section{Introduction}

The simple majority voting rule offers a natural way of aggregating voters' preferences. Already Condorcet~\cite{Con85} used pairwise comparisons to calculate the winning candidate, establishing his famous paradox on the smallest voting instance not admitting a majority winner.

The concept of majority voting translates to other scenarios where voters submit preference lists. One such field is the area of two-sided matchings under preferences, where popular matchings~\cite{Gardenfors75,AbrahamIKM07,SM10,Man13,Cse17} serve as a voting-based alternative to the well-known notions of stable matchings~\cite{GS62,BB02a} and Pareto optimal matchings~\cite{ACM+04}. In short, a popular matching~$M$ is a simple majority winner among all matchings, as it guarantees that no matter what alternative matching is offered on the market, a weak majority of the non-abstaining agents will opt for~$M$. Restricted to matching instances on bipartite graphs, the popular matching problem has been studied in two models.
\begin{itemize}
\item {\em House allocation (i.e.\ one-sided preferences).} One side of the graph consists of agents who have strictly ordered preferences and cast votes, while the other side is formed by houses with no preferences or votes.
\item {\em Two-sided preferences.} Vertices on both sides are agents, who all have strictly ordered preferences and cast votes. This setting is analogous to the classic stable marriage model.
\end{itemize}

In this paper, we focus on the latter model and extend it to a direction motivated by voting. More specifically, we supply the agents in the instance with weights. The vote of each agent then counts with multiplicity: the larger the weight is, the more influence this agent has on the outcome of the voting between two matchings.
Analogous weighted voting~\cite{Ban64} scenarios arise naturally in various real-life problems.
%committee voting~\cite{RS66} and liquid democracy~\cite{KMP21}. 
In committees~\cite{RS66}, the vote of persons on different posts might be weighted differently when tallying them. In liquid democracy~\cite{KMP21}, the delegates' votes also count with multiplicity, depending on the number of agents they represent. Another interpretation of weighted voters might be that prioritized voters are assigned a larger weight in order to ensure their beneficial treatment~\cite{CT77}. Also, coalitions often decide to vote together, in which case the decision made within the coalition will also be counted with multiplicity~\cite{ASS+05}.

Weighted voting readily translates to vertex-weighted graphs in the \new{context of matchings}. It follows already from the Condorcet paradox that the existence of a popular matching is not guaranteed if the vertices are weighted, as \Cref{fig:condorcet} demonstrates. However, the complexity of deciding whether a \new{popular matching} exists in a given instance was open until now.

    \begin{figure}[htb]
      \begin{center}
      \begin{minipage}{0.4\textwidth}
		\[
		\begin{array}{rlll}
		a_{1,2,3}:  & b_1 \succ & b_2 & \\
        b_{1,2}:  & a_1 \succ & a_2 & \succ a_3
		\end{array}
		\]
	\end{minipage}\hspace{17mm}\begin{minipage}{0.4\textwidth}	
        \begin{tikzpicture}[scale=0.75,every node/.style={scale=0.9}]
            \node[vertex, label=180:$a_1$] (a1) at (0,0) {};
            \node[vertex, label=180:$a_2$] (a2) at (0,2) {};
            \node[vertex, label=180:$a_3$] (a3) at (0,4) {};
            \node[vertex, label=0:$b_1$] (b1) at (3,1) {};
            \node[vertex, label=0:$b_2$] (b2) at (3,3) {};
            \draw (a1) edge node[pos=0.2, fill=white, inner sep=2pt] {1} node[pos=0.8, fill=white, inner sep=2pt] {1} (b1);
            \draw (a2) edge node[pos=0.2, fill=white, inner sep=2pt] {1} node[pos=0.8, fill=white, inner sep=2pt] {2}  (b1);
            \draw (a3) edge node[pos=0.2, fill=white, inner sep=2pt] {1} node[pos=0.8, fill=white, inner sep=2pt] {3}  (b1);
            \draw (a1) edge node[pos=0.2, fill=white, inner sep=2pt] {2} node[pos=0.8, fill=white, inner sep=2pt] {1}  (b2);
            \draw (a2) edge node[pos=0.2, fill=white, inner sep=2pt] {2} node[pos=0.8, fill=white, inner sep=2pt] {2}  (b2);
            \draw (a3) edge node[pos=0.2, fill=white, inner sep=2pt] {2} node[pos=0.8, fill=white, inner sep=2pt] {3}  (b2);
        \end{tikzpicture}
        \end{minipage}
      \end{center}
      \caption{With certain vertex weights, no popular matching exists in this instance. The lists on the left side as well as the numbers on the edges indicate the preferences of the vertices: For every~$i\in \{1,2,3\}$, $a_i$'s first choice is $b_1$, while its second choice is~$b_2$. The Condorcet paradox with no majority winner corresponds to assigning weight~1 to $a_1, a_2, a_3$ and weight~0 to $b_1, b_2$.
      For example, matching~$M_1 := \{\{a_1, b_1\}, \{a_2, b_2\}\}$ is less popuplar than $M_2 := \{\{a_2, b_1\}, \{a_3, b_2\}\}$.
      Matching~$M_2$ is less popular than~$M_3 := \{\{a_3, b_1\}, \{a_1, b_2\}\}$.
      Finally, $M_3$ is less popular than~$M_1$.  
      The presence of 0-weight vertices is not necessary: the same example works also with weight~3 for $a_1, a_2, a_3$ and weight~1 for $b_1, b_2$.}
      \label{fig:condorcet}
    \end{figure}
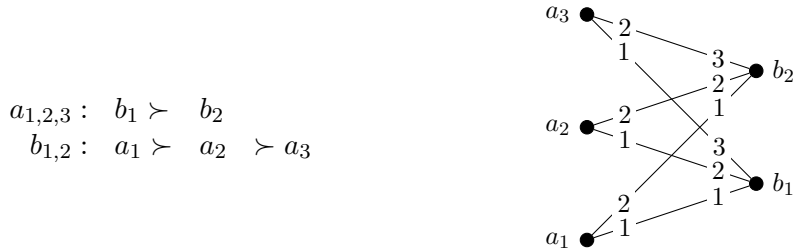

Another motivation for popular matchings arises from the context of stable matchings.
While stable matchings provide a local stability criterion (the absence of so-called blocking pairs), popularity takes a more global perspective and does not forbid local instabilities but only global instabilities, i.e.\ a complete matching preferred to the current matching.
Therefore, popular matchings have been proposed as an alternative to stable matchings~\cite{HK13} as they may be larger than stable matchings.
In the context of bipartite stable matching problems, the different sides often represent fundamentally different agents (i.e.\ when considering the admission of students to colleges, then one side represents the students, while the other side represents the free slots at colleges);
consequently, it is natural to \new{assign larger weights to preferences of one side} than \new{to} the preferences of the other side.
Recently, Kavitha~\cite{DBLP:conf/icalp/Kavitha20} suggested to search for popular matchings when the vote of each vertex from one side is counted with multiplicity~$c > 1$ and each vertex from the other side with multiplicity~1.

\subsection{Related work}

Matchings under preferences constitute a field actively researched by both Computer Scientists and Economists~\cite{Rot82,HK09,Man13}. Besides two-sided matchings, majority voting has also been defined for the roommates problem~\cite{FaenzaKPZ19,GMS+21}, spanning trees~\cite{Dar13}, permutations~\cite{VSW14,KCM21}, the ordinal group activity selection problem~\cite{Dar18}, and branchings~\cite{KKM+20}. Matchings in bipartite graphs are nevertheless the most actively researched area~\cite{BK19,FaenzaKPZ19,FaenzaK20,KM20,GMS+21,HK21,AB22} of the majority voting rule outside of the usual voting scenarios. In the context of matchings it was first introduced by Gärdenfors~\cite{Gardenfors75} for matching markets with two-sided preferences, and then studied by Abraham et al.~\cite{AbrahamIKM07} in the house allocation model. Polynomial\new{-}time algorithms to find a popular matching were given in both settings. In the two-sided preferences model, it was already noticed by Gärdenfors~\cite{Gardenfors75} that all stable matchings are popular, which implies that in standard bipartite stable matching instances, popular matchings always exist. In fact\new{,} stable matchings are the smallest size popular matchings, as shown by Biró et al.~\cite{DBLP:conf/ciac/BiroIM10}, while maximum size popular matchings can be found in polynomial time as well~\cite{HK13,Kav14}. 
Notice that as soon as there are vertices with different weights, for example in the house allocation model, where one side has weight~0, the existence of a popular matching is not guaranteed any more.

\new{A} natural extension of various matching problems is to consider graphs with edge or vertex weights. Notice however that the two weighted extensions define inherently different problems for popularity. If vertices are weighted, as in our setting, then this weight influences the voting power of an agent and redefines the more popular relation between matchings. If edge weights are given, the vertices vote exactly as in the non-weighted version, but the matchings carry a different total weight, and the goal is to find a maximum weight matching, subject to popularity.

For house allocation instances with edge weights and strict preferences, a maximum weight one among popular matchings can be found in $O(n + m)$ time~\cite{MI11}, where $n$ is the number of vertices and $m$ is the number of edges in the graph. For instances with two-sided preferences, Faenza et al.~\cite{FaenzaKPZ19} showed that it is \NP-complete to decide if there exists a popular matching that contains two given edges, which is a very restricted case of popular matchings with edge weights. The same authors provided a 2-approximation algorithm for non-negative edge weights. \NP-hardness was established for edge-weighted non-bipartite instances a couple of years earlier already~\cite{HK21}.

The vertex-weighted case has been studied extensively in the house allocation setting. Mestre~\cite{DBLP:journals/talg/Mestre14} presented algorithms to find a popular matching, or a proof for its non-existence. For instances with strict preference lists, he gave an $O(n + m)$ time algorithm, while for preference lists with ties, he solved the problem in  $O(\min\left\{k\sqrt{n},n\right\} m)$ time, where $k$ is the number of distinct weights the agents are given. The algorithm for strict lists was later extended to capacitated house allocation instances by Sng and Manlove~\cite{SM10}, who posed the complexity of the case with weakly ordered preferences as an open problem. Itoh and Watanabe~\cite{IW10} studied the existence probability of a popular matching in random house allocation instances with vertex weights. Ruangwises and Itoh~\cite{DBLP:journals/mst/RuangwisesI21} designed an algorithm to compute the approximability measure introduced by McCutchen~\cite{DBLP:conf/latin/McCutchen08} and called unpopularity factor for a given matching. As a byproduct, they also developed a polynomial algorithm to verify the popularity of a given matching, even in non-bipartite instances.

Kavitha~\cite{DBLP:conf/icalp/Kavitha20} gave a polynomial-time algorithm for finding a matching which is popular when each vertex has weight~1 and simultaneously is popular in the house allocation setting, i.e.\ when each vertex from~$A$ has weight~1 while each vertex from~$B$ has weight~0.

\subsection{Structure of the paper and our contribution}

We define our notation and main problem in Section~\ref{sec:prel}. There we also sketch three relevant known techniques.
We extend one of these, the concept of a witness, to the vertex-weighted case in Section~\ref{sec:ver}. Then in Section~\ref{sec:npc} we give two different hardness reductions to show that finding a popular matching is \NP-hard even if the vertex weight function is very restricted. A strong inapproximability result is also provided there for the problem variant with edge costs. We complement these findings in Section~\ref{sec:alg} with our main result, an algorithm to find a popular matching in polynomial time (if any exists) for instances where the weight of every vertex in $A$ is a constant~${c>3}$, while the weight of every vertex in $B$ is~1. 
%Since $A$ and $B$ model two different kinds of vertices, it is natural to assign each vertex from~$A$ a different weight to the weight of vertices from~$B$ (while each vertex from $A$ (respectively $B$) has the same weight).
%Notably, Kavitha~\cite{DBLP:conf/icalp/Kavitha20} recently observed that almost nothing is known when all agents from on side have weight~$c > 1$ while all vertices from the other side have weight~1.
We pose open questions in Section~\ref{sec:open} and provide detailed examples in the Appendix. %Results marked with (\appmark) are proven in the appendix.

\section{Preliminaries}
\label{sec:prel}
In this section we first introduce the notation we use, and then describe three known techniques we rely on in our proofs.

\subsection{Notation, input, and the popularity criterion}

In our input, a simple bipartite graph $G = (A \cup B, E)$ on $n$ vertices and $m$ edges is given. We denote the set of vertices by $V = A \cup B$, the set of vertices adjacent to $v \in V$ in~$G$ by $N_G (v)$, and the set of edges incident to $v \in V$ in~$G$ by~$\delta_G (v)$. The vertices of $G$ are equipped with weights determined by the weight function~$w: V \rightarrow \mathbb{Q}_{\geq 0}$. Furthermore, each vertex $v \in V$ has a strictly ordered preference list $\succ_v$ over the vertices in~$N_G (v)$, where $u \succ_v z$ means that $v$ prefers $u$ to~$z$. We assume that being matched to any vertex in $N_G (v)$ is preferred to staying unmatched. 
We may drop the subscript~$G$ if it is clear from context.

A \emph{matching}~$M \subseteq E$ is a set of pairwise disjoint edges.
A matching~$M$ is \emph{perfect} if every vertex is covered by an 
%contained in at least one 
edge of~$M$.
For each matching $M$ and vertex $v \in V$, let $M(v)$ denote the vertex $M$ matches $v$ to.
\new{If $v$ remains unmatched in $M$, then we write~$M(v) := \bot$.}
The preference list $\succ_v$ naturally defines a preference relation between any two matchings $M$ and~$M'$ for each $v \in V$: vertex~$v$ prefers $M$ to $M'$ if $M(v) \succ_v M'(v)$. If $M(v) = M'(v)$, then $v$ is indifferent between~$M$ and~$M'$.

Given two matchings $M$ and $M'$, let $V^+ (M,M')$ be the set of vertices preferring $M$ to~$M'$. We define 
$$\Delta_w (M, M') := \sum_{v\in V^+ (M, M')} w(v) - \sum_{v\in V^+ (M', M)} w (v).$$
A matching $M$ is \emph{popular} if $\Delta_w (M, M') \ge 0$ for every matching~$M'$. In words, $M$ is popular if it is never beaten in a pairwise comparison with another matching, where each vertex casts a vote for its preferred matching with the multiplicity of its weight. 
We call an edge~$e\in E(G)$ \emph{popular} if it is contained in some popular matching. We now formally define the main problem studied in this paper.
\defProblemQuestion{\textsc{Popular Matching with Weighted Voters}}
{
A bipartite graph $G = (A \cup B, E)$, $(\succ_v)_{v \in V}$, and vertex weights $w: V \rightarrow \mathbb{Q}_{\geq 0}$.
}
{
Is there a popular matching $M$?
}

\subsection{Popular matching characterizations}
\label{sec:techn}

We now sketch useful known techniques for characterizing popular matchings in house allocation instances, instances with two-sided lists, and finally, instances with weighted voters. 

\subsubsection{House allocation instances}
We first focus on house allocation instances, which correspond to the subcase of our problem definition with $w(a) = 1$ for each $a \in A$ and  $w(b) = 0$ for each $b \in B$. The characterization of Abraham et al.~\cite{AbrahamIKM07} uses the notion of an \emph{$f$-post} and \emph{$s$-post} of a vertex~$a\in A$.
  \begin{definition}
  The \emph{$f$-post}~$f(a)$ of 
  %a vertex~
  $a\in A$ is the first vertex in the preferences of~$a$.
  The \emph{$s$-post}~$s(a)$ of %a vertex~
  $a\in A$ is the vertex that is best-ranked among all vertices that are not the $f$-post of any vertex, i.e.\ the best-ranked vertex among $B \setminus \{ f(a') : a' \in A\}$.
  If no such vertex exists, then we set~$s(a) := \bot$.
  \end{definition}
  
  %With this notation, we are ready to state the characterization of popular matchings in a house allocation instance. We remark that this 
  The following characterization of popular matchings %also
  leads to an $O(n+m)$ algorithm that outputs either a largest cardinality popular matching or a proof for its nonexistence.%\accom{Do we use this remark anywhere? Does it help the reader?}
  
  \begin{theorem}[Abraham et al.~\cite{AbrahamIKM07}]\label{th:abraham}
   In a house allocation instance, a matching~$M$ is popular if and only if the following conditions are both fulfilled.
   \begin{itemize}
    \item Each vertex $b$ such that $b = f(a)$ for some $a\in A$ is matched in $M$ to some~$a'\in A$ with $b = f(a')$.
    \item For each vertex $a\in A$, it holds that $M(a) \in \{f(a), s(a)\}$.
   \end{itemize}
  \end{theorem}
  
\subsubsection{Instances with two-sided lists}
\label{sec:twosided}
Now we turn to instances with two-sided preferences and unit weights, i.e.\ $w(v) = 1$ for each~${v \in V}$. Let $\tilde{G}$ be the graph $G$ augmented with a loop at each vertex, such that each vertex is its own last choice. This allows us to regard any matching $M$ in $G$ as a perfect matching $\tilde{M}$ in $\tilde{G}$ by including loops at all vertices left unmatched in~$M$. First we define the vote of $u \in V$ for $v \in N_G(u)$ as
%\begin{linenomath}
\[\vote^M_u (v) := \begin{cases}
  0 & \text{ if } \{u, v\} \in M,\\
  1 & \text{ if $u$ is unmatched in~$M$ or $v \succ_u M(u)$, and}\\
  -1 & \text{ otherwise (i.e.\ $M (u) \succ_u v)$.}
 \end{cases}\]
%\end{linenomath}

This function allows us to express the total votes of the two end vertices \new{of an edge} 
%if edge $e$ is added to~$M$ 
as follows. For an edge $e = \{u, v\}$, we define $\vote^M (e) := \vote^M_u (v) + \vote^M_{v} (u)$. %Notice that an edge $\{u,v\}$ is of weight -2 if $M(u) \succ_u v$ and $M(v) \succ_v u$, it is of weight 2 if $v \succ_u M(u)$ and $u \succ_v M(v)$, while in all other cases, it is of weight~0. 
For loops, $\vote^M(\{u,u\}) = 0$ if $u$ is matched to itself in $\tilde{M}$, else $\vote^M(\{u,u\}) = -1$.
For any matching~$M'$ in $ G$, it is easy to see that $\vote^M(\tilde{M'}) := \sum_{e \in \tilde{M}'} \vote^M(e) = \Delta_{w = 1}(M',M)$. This delivers the first characterization of popular matchings in instances with two-sided preferences.

\begin{theorem}[Kavitha et al.~\cite{KavithaMN11}]
  \label{thm:weight_zero} $M$ is popular in $G$ if and only if for every perfect matching $\tilde{M'}$ in~$\tilde{G}$ it holds that $\vote^M(\tilde{M'}) \leq 0$.
\end{theorem}
  The second characterization follows from LP-duality and the fact that $G$ is a bipartite graph.
\begin{theorem}[Kavitha et al.~\cite{KavithaMN11,Kav16}]
  \label{thm:witness}
  A matching $M$ in $(G = (V = A \cup B, E),\allowbreak {(\succ_v)_{v \in V}})$ is popular if and only if there exists a vector $\witness \in \{0, \pm 1\}^n$ such that
  \begin{itemize}
      \item $\sum_{v \in V} \witness_v = 0$,
      \item $\witness_a + \witness_b \ge \vote^M(\{a,b\}) \qquad \forall\, \{a,b\}\in E$, and
      \item $\witness_v  \ge\vote^M(\{v,v\}) \qquad \forall\, v \in V$.
  \end{itemize}
%  $\sum_{v \in V} \witness_v = 0$,
%  \[ \witness_a + \witness_b \ \ \ge \ \ \vote^M(\{a,b\})\ \ \ \forall\, \{a,b\}\in E\ \ \ \ \ \ \text{and}\ \ \ \ \ \ \witness_v \ \ \ge\ \ \vote^M(\{v,v\})\ \ \ \forall\, v \in V.\] 
\end{theorem}

%For a popular matching $M$, a vector $\witness$ as given in Theorem~\ref{thm:witness} is called the {\em witness} of~$M$. A popular matching may have several witnesses.

\subsubsection{Instances with weighted voters}

Ruangwises and Itoh~\cite{DBLP:journals/mst/RuangwisesI21} extended \Cref{thm:weight_zero} to instances with weighted voters. They redefined the vote of $u \in V$ for $v \in N_G(u)$ as
%\begin{linenomath}
\[\weightedvote^M_u (v) := \begin{cases}
  0 & \text{ if } \{u, v\} \in M,\\
  w(u) & \text{ if $M(u) = \bot$ or $v \succ_u M(u)$, and}\\
  -w(u) & \text{ otherwise.}
 \end{cases}\]
% \end{linenomath}
%Whenever we will use~$\weightedvote_u^M (v)$ in the rest of the paper, this will refer to this vote function (and not the special unweighted case introduced in \Cref{sec:twosided}).
For an edge $e = \{u, v\}$, Ruangwises and Itoh~\cite{DBLP:journals/mst/RuangwisesI21} defined $\weightedvote^M (e) := \weightedvote^M_u (v) + \weightedvote^M_{v} (u)$.  For a loop~$e = \{v,v\}$, they set $\weightedvote^M (e) = -w( v)$ if $v$ is matched in~$M$ and $\weightedvote^M ( e) = 0$ otherwise.
Note that for any matching~$M'$ in $G$, we have $\weightedvote^M(\tilde{M'}) := \sum_{e \in \tilde{M}'} \weightedvote^M(e) = \Delta_{w}(M',M)$.

\begin{theorem}[Ruangwises and Itoh~\cite{DBLP:journals/mst/RuangwisesI21}]\label{cl:zero_weight}
    $M$ is popular if and only if for every matching~$\tilde{M}'$ in $\tilde{G}$ it holds that $\weightedvote^M (\tilde{M}') \le 0$,  where~$\tilde{G}$ is the input graph~$G$ augmented with a loop~$\{v, v\}$ for every vertex~${v\in V}$, as defined in \Cref{sec:twosided}.
\end{theorem}

\section{Witness of popularity}
\label{sec:ver}

We now extend \Cref{thm:witness} to two-sided instances with weighted voters.
The proof is essentially analogous to the proof of \Cref{thm:witness}:
Start with the characterization of popular matchings from \Cref{cl:zero_weight}, formulate this as an LP using the standard bipartite matching LP and then dualize the LP.

\begin{theorem}\label{lem:witness}
 For \textsc{Popular Matching with Weighted Voters}, matching $M$ is popular if and only if there exists a vector $\witness\in \mathbb{Q}^{V}$ with $\sum_{v\in V} \wof{v} = 0$ such that $\wof{v} + \wof{u} \ge \weightedvote^M (e)$ for every edge~$e = \{u, v\} \in E(G)$, $\wof{v} \ge 0$ for every vertex unmatched by~$M$, and $\wof{v} \ge -w (v)$ for every~$v \in V$.
\end{theorem}

\begin{proof}
 The perfect matching polytope for bipartite graphs (that is, the convex hull of all perfect matchings) can be described by the inequalities $x_e \ge 0$ for every edge~$e$ and $\sum_{e \in \delta (v)} x_e = 1$ for every vertex~$v$~\cite{Birkhoff46} (see also~\cite[Theorem 18.1]{Schrijver}).
 Applying this to the instance~$\tilde{G}$ constructed in \Cref{sec:twosided}, we can express the problem of finding a perfect matching in~$\tilde{G}$ by the following linear program~(\ref{lp:primal}).
 %\begin{linenomath}
 \begin{equation}
 \max \left\{\sum_{e\in E (\tilde{G})} \weightedvote^M (e) x_e \big\vert \sum_{e\in \delta_{\tilde{G}} (v)} x_e = 1 \ \forall v\in V, x_e \ge 0 \ \forall e\in E(\tilde{G})\right\}
 \label{lp:primal}
 \end{equation}
 %\end{linenomath}
 Dualizing LP~(\ref{lp:primal}) results in LP~(\ref{lp:dual}).
 %\begin{linenomath}
 \begin{equation}
 \min \left\{\sum_{v\in V} y_v \big\vert y_v + y_u \ge \weightedvote^M (\{u,v\}) \ \forall \{u, v\} \in E(G), y_v \ge \weightedvote^M (\{v, v\})\right\}
 \label{lp:dual}
 \end{equation}
 %\end{linenomath}
 Note that the matching~$M$ always induces a solution of cost~0 to LP~(\ref{lp:primal}).
 By the strong duality for linear programs (see e.g.\ \cite[Theorem 5.4]{Schrijver}), it follows that LP~(\ref{lp:dual}) admits a solution of cost 0 if and only if the optimal solution of LP~(\ref{lp:primal}) has cost~0.
 This is equivalent to every perfect matching in $\tilde{G}$ having cost at most~0, which is by \Cref{cl:zero_weight} equivalent to $M$ being popular.
\end{proof}

\Cref{lem:witness} motivates the following definition of a witness:

\begin{definition}[Witness]\label{def:witness}
  Let $G = (V = A\cup B, E)$ together with preferences for each~$v\in V$ and vertex weights $w: V \rightarrow \mathbb{Q}_{\geq 0}$ be an instance of \textsc{Popular Matching with Weighted Voters}.
 Given a popular matching $M$, %a \emph{witness of the popularity of $M$}, or, shortly, 
 a \emph{witness of $M$} is a vector~$\witness\in \mathbb{Q}^{V}$ such that \neu{ % for every $v\in V$
 \begin{itemize}
     \item  $\sum_{v\in V} \witness_v = 0$, $\witness_a + \witness_b \ge \weightedvote^M (\{a, b\})$ for every edge~$\{a, b\} \in E(G)$, 
     \item $\witness_v \ge 0$ for each vertex~$v \in V$ which is unmatched in~$M$, and
     \item $\witness_v \ge - w(v)$ for each vertex~$v \in V $ which is matched in~$M$.
 \end{itemize}
}
 %$\sum_{v\in V} \witness_v = 0$, $\witness_a + \witness_b \ge \weightedvote^M (\{a, b\})$ for every edge~$\{a, b\} \in E(G)$, $\witness_v \ge 0$ if $v$ is unmatched in~$M$ and $\witness_v \ge - w(v)$ otherwise for every $v\in V$.

 For matching $M$ and vector $\witness \in \mathbb{Q}^V$, we call an edge $\{a, b\}$ \emph{conflicting} if $\wof{a} + \wof{b} < \weightedvote^M (\{a,b\})$.
\end{definition}
An example of popular matchings together with a witness can be found in \Cref{fig:hcp} together with \Cref{fig:initial-witnesses} in the Appendix.
By \Cref{lem:witness}, a matching~$M$ is popular if and only if there exists a witness of~$M$.

\section{\NP-hardness}
\label{sec:npc}

In this section, we show that two highly restricted variants of \textsc{Popular Matching with Weighted Voters} are \NP-complete.
The first variant (see \Cref{sec:np-const}) assumes a set of non-unit weight vertices of constant size and only 2 kinds of weights for them.
The second variant assumes that all vertices on each side of the bipartition have identical weights (see \Cref{sec:np:identical-weights}). 
Finally, we show a strong inapproximability result for the case that there are utilities on the edges and one aims to find a popular matching of maximum utility (see \Cref{sec:np:weighted}).

\subsection{Constant number of non-unit weight agents}
\label{sec:np-const}
We now show that determining the existence of a popular matching is \NP-complete even if all but 14 agents have weight~1.
In order to do so, we reduce from the \NP-complete problem of deciding, given an instance of \textsc{Popular Matching} (with two-sided preferences and unit weights) and two edges~$e_1$ and $e_2$, whether there is a popular matching containing both~$e_1$ and~$e_2$~\cite{FaenzaKPZ19}.
The reduction consists of replacing each of~$e_1$ and~$e_2 $ by the gadget from \Cref{fig:eg}.
Intuitively, the gadget from \Cref{fig:eg} ensures that every popular matching in the constructed instance ``contains"~$e_1$ and~$e_2$.

    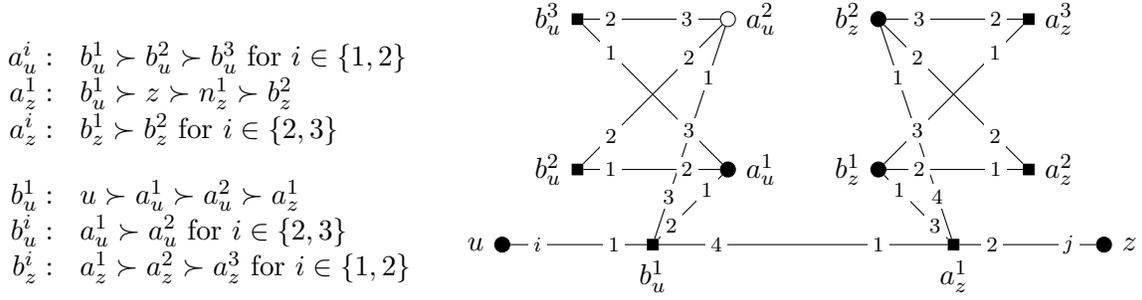
\begin{figure}[htb]
    \centering
       \begin{minipage}{0.7\textwidth}
		\[
		\begin{array}{rl}
    a_u^i&: b_u^1 \succ  b_u^2 \succ b_u^3 $ for $i\in\{1, 2\}\\
    a_z^1 &: b_u^1  \succ z \succ b_z^1 \succ  b_z^2 \\ \vspace{4mm}
    a_z^i &: b_z^1 \succ b_z^2$ for $i \in \{2, 3\}\\
    b_u^1 &: u \succ a_u^1 \succ  a_u^2 \succ a_z^1\\
	b_u^i &: a_u^1 \succ a_u^2$ for $i \in \{2, 3\}\\
    b_z^i&: a_z^1 \succ  a_z^2 \succ a_z^3 $ for $i\in\{1, 2\}
		\end{array}
		\]
	\end{minipage}\vspace{8mm}
 \begin{minipage}{1\textwidth}	
 \centering
        \begin{tikzpicture}
        \node at (-2.75, 0) {};
          \node[vertex, label=90:$u$] (v) at (-2,0) {\scriptsize 1};
          \node[terminal, label=270:$b_u^1$] (xv1) at (-0,0) {\scriptsize 4};
          \node[terminal, label=180:$b_u^2$] (xv2) at (-1,1) {\scriptsize 4};
          \node[terminal, label=180:$b_u^3$] (xv3) at (-1,3) {\scriptsize 4};
          \node[vertex, label=0:$a_u^1$] (yv1) at (1.,1) {\scriptsize 1};
          \node[white-vertex, label=0:$a_u^2$] (yv2) at (1.,3) {\scriptsize 2};
          \node[vertex, label=90:$z$] (w) at (6,0) {\scriptsize 1};
          \node[terminal, label=270:$a_z^1$] (xw1) at (4,0) {\scriptsize 4};
          \node[terminal, label=0:$a_z^2$] (xw2) at (5,1) {\scriptsize 4};
          \node[terminal, label=0:$a_z^3$] (xw3) at (5,3) {\scriptsize 4};
          \node[vertex, label=180:$b_z^1$] (yw1) at (3,1) {\scriptsize 1};
          \node[vertex, label=180:$b_z^2$] (yw2) at (3,3) {\scriptsize 1};

        \draw (v) edge node[pos=0.2, fill=white, inner sep=2pt] {\scriptsize \new{$k$}}  node[pos=0.76, fill=white, inner sep=2pt] {\scriptsize $1$} (xv1);
        \draw (xv1) edge node[pos=0.2, fill=white, inner sep=2pt] {\scriptsize $2$}  node[pos=0.76, fill=white, inner sep=2pt] {\scriptsize $1$} (yv1);
        \draw (xv1) edge node[pos=0.2, fill=white, inner sep=2pt] {\scriptsize $3$}  node[pos=0.76, fill=white, inner sep=2pt] {\scriptsize $1$} (yv2);
        \draw (xv1) edge node[pos=0.2, fill=white, inner sep=2pt] {\scriptsize $4$}  node[pos=0.76, fill=white, inner sep=2pt] {\scriptsize $1$} (xw1);

        \draw (xv2) edge node[pos=0.2, fill=white, inner sep=2pt] {\scriptsize $1$}  node[pos=0.76, fill=white, inner sep=2pt] {\scriptsize $2$} (yv1);
        \draw (xv2) edge node[pos=0.2, fill=white, inner sep=2pt] {\scriptsize $2$}  node[pos=0.76, fill=white, inner sep=2pt] {\scriptsize $2$} (yv2);

        \draw (xv3) edge node[pos=0.2, fill=white, inner sep=2pt] {\scriptsize $1$}  node[pos=0.76, fill=white, inner sep=2pt] {\scriptsize $3$} (yv1);
        \draw (xv3) edge node[pos=0.2, fill=white, inner sep=2pt] {\scriptsize $2$}  node[pos=0.76, fill=white, inner sep=2pt] {\scriptsize $3$} (yv2);

        \draw (w) edge node[pos=0.2, fill=white, inner sep=2pt] {\scriptsize \new{$\ell$}}  node[pos=0.76, fill=white, inner sep=2pt] {\scriptsize $2$} (xw1);
        \draw (xw1) edge node[pos=0.2, fill=white, inner sep=2pt] {\scriptsize $3$}  node[pos=0.76, fill=white, inner sep=2pt] {\scriptsize $1$} (yw1);
        \draw (xw1) edge node[pos=0.2, fill=white, inner sep=2pt] {\scriptsize $4$}  node[pos=0.76, fill=white, inner sep=2pt] {\scriptsize $1$} (yw2);

        \draw (xw2) edge node[pos=0.2, fill=white, inner sep=2pt] {\scriptsize $1$}  node[pos=0.76, fill=white, inner sep=2pt] {\scriptsize $2$} (yw1);
        \draw (xw2) edge node[pos=0.2, fill=white, inner sep=2pt] {\scriptsize $2$}  node[pos=0.76, fill=white, inner sep=2pt] {\scriptsize $2$} (yw2);

        \draw (xw3) edge node[pos=0.2, fill=white, inner sep=2pt] {\scriptsize $1$}  node[pos=0.76, fill=white, inner sep=2pt] {\scriptsize $3$} (yw1);
        \draw (xw3) edge node[pos=0.2, fill=white, inner sep=2pt] {\scriptsize $2$}  node[pos=0.76, fill=white, inner sep=2pt] {\scriptsize $3$} (yw2);
        \end{tikzpicture}
    \end{minipage}
      \caption{An example of the edge gadget for a forced edge $\{u, z\}$, where $u$ ranks~$z $ at the \new{$k$}-th position, and $z$ ranks $u$ at the \new{$\ell$}-th position.
      Squared vertices have weight four, trapezes have weight two, and round vertices have weight one.
      }
      \label{fig:eg}
    \end{figure}%\accom{The prefs were next to the graph, but it looked too crammed.}
    
\begin{theorem}\label{thm:np-harda}
    \textsc{Popular Matching with Weighted Voters} is \NP-com\-plete, even if all but 14 vertices have weight one, and for each vertex, its weight is one, two, or four.
\end{theorem}

We first briefly sketch the proof before giving the formal reduction and correctness proof.
%\begin{proof}[Proof Sketch]
 Given an instance~\neu{$\mathcal{I}$}
 %$\Imatching$
 of \textsc{Popular Matching with Two Forced Edges}, we replace each forced edge~$ \{u, z\}$ by the gadget depicted in \Cref{fig:eg}.
 All agents not contained in such a gadget have weight one.
 We call the resulting instance~$\mathcal{J}$.
 \new{Note that \neu{th}e matchings in instance~$\mathcal{J}$ will always be called $\Jmatching$ (plus possibly a superscript), while matchings in instance~$\mathcal{I}$ will be denoted by~$\Imatching$ (plus possibly a superscript).}
 
 Given a popular matching~$\Imatching$ in~$\mathcal{I}$ containing the two forced edges, we construct a popular matching~$\Jmatching$ in~$\mathcal{J}$ by replacing each forced edge~$\{u, z\}$ by the edges~$\{a_u^2, b_u^2\}$, $\{a_u^1, b_u^3\}$, $\{a_z^2, b_z^2\}$, $\{a_z^3, b_z^1\}$, $\{u, b_u^1\}$, and $\{ a_z^1, z\}$.
 To show that $\Jmatching$ is popular, one first shows that when comparing~$\Jmatching$ to any other matching~$\Jmatchingprime$ in $\mathcal{J}$, then \new{the} summed weighted vote of the agents added in the two edge gadgets will never be in favor of~$\Jmatchingprime$.
 As a second step, we show that if there is a matching~$\Jmatchingprime$ which is more popular than~$\Jmatching$, then we may assume that $\Jmatchingprime$ contains either edges~$\{u, b_u^1\}$ and~$\{z, a_z^1\}$ or neither edge~$\{u, b_u^1\}$ nor~$\{z, a_z^1\}$ for every forced edge~$\{u, z\}$.
 Consequently, if matching~$\Jmatchingprime$ is more popular than $\Jmatching$, then also the matching~$M'$ arising from~$\Jmatchingprime $ by ``inverting" the edge gadgets (i.e.\ removing the edges in the edge gadgets and adding edge~$\{u, z\}$ whenever~$\{u, b_u^1\}$ and $\{a_z^1, z\}$ are contained in~$\Jmatchingprime$) is more popular than~$M$.
 Thus, the popularity of $M$ implies the popularity of~$\Jmatching$.
 
 For the reverse direction, given a popular matching~$\Jmatching$ in~$\mathcal{J}$, the first step is to show that~$\Jmatching$ contains~$\{u, b_u^1\}$ and $\{ a_z^1, z\}$ for every forced edge~$\{u, z\}$.
 The matching~$\Imatching$ arising from $\Jmatching$ by inverting the edge gadgets is then a popular matching in $\mathcal{I}$ which contains every forced edge.
 This finishes the proof sketch.
%\end{proof}

%%%%%%%%%%%%%%%%%%%%%

%\subsection{Proof of Theorem~\ref{thm:np-harda}}

We now give the formal reduction and proof of \cref{thm:np-harda}.
%\npharda*
  \begin{proof}[Proof of \Cref{thm:np-harda}]
  For a set of vertices $X\subseteq V$ and two matchings $\Jmatching$ and~$\Jmatchingprime$, we say that $X$ \emph{prefers~$\Jmatching$ to~$\Jmatchingprime$} if 
  $$\sum_{v\in V^+ (\Jmatching, \Jmatchingprime)\cap X} w(v) - \sum_{v\in V^+ (\Jmatchingprime, \Jmatching\neu{)} \cap X} w(v) > 0.$$
  neu{)}
    By \Cref{cl:zero_weight}, \textsc{Popular Matching with Weighted Voters} is in \NP, as a popular matching can be used as a certificate. To show \NP-hardness, we reduce from \textsc{Popular Matching with Two Forced Edges}, which was shown to be \NP-complete by Faenza et al.~\cite{FaenzaKPZ19}.

\defProblemQuestion{\textsc{Popular Matching with Two Forced Edges}}
{
A popular matching instance and a set $F$ of two edges.
}
{
Does there exist a popular matching $M$ with $F\subseteq M$?
}

We can assume without loss of generality that $F$ is a matching, as otherwise no matching can contain~$F$. Let $\mathcal{I}= (G = ( A \cup B, E), (\succ_v)_{v \in A\cup B}, F)$ be an instance of this problem.

    \textbf{Construction.}  We replace each forced edge $e = \{u, z\}\in F$, where $u\in A$ and $z \in B$, by an edge gadget.
    The edge gadget contains ten vertices $b_u^1$, $b_u^2$, $b_u^3$, $a_u^1$, $a_u^2$ and $a_z^1, a_z^2, a_z^3$, $b_z^1$, and~$b_z^2$, as depicted in \Cref{fig:eg}, and it is asymmetric in $u$ and~$z$. 
    %We denote the vertices of a forced edge by $u$ and $z$, and assume that $u\in A$ and $z \in B$.
    %The gadget contains ten vertices $b_u^1$, $b_u^2$, $b_u^3$, $a_u^1$, $a_u^2$ and $a_z^1, a_z^2, a_z^3$, $b_z^1$, and~$b_z^2$. 
    We set $w( v) = 1$ for all $v\in A \cup B$, $w (b_u^i) = \largeWeight = w (a_z^i)$, $w(a_u^1) = 1$, $w(a_u^2) = 2$, and $w(b_z^i ) = 1$. We call the resulting instance~$\mathcal{J}$.
%    \new{Recall that we will refer to matchings in~$\mathcal{J}$ by $\Jmatching$ (plus possible super- or subscripts).}

    Note that all agents outside the two edge gadgets have weight one, while each edge gadget contains one agent of weight \new{two} and six agents of weight \new{four}.
    Thus, all but 14 vertices have weight one, and the remaining vertices have weight two or four.

    We define a ``projection'' $\pi$ from \new{any} matching \new{$\Imatching$} in~$\mathcal{I}$ \new{(which may contain 0, 1, or 2 forced edges)} to \new{some} matching in~$\mathcal{J}$ via
   % \begin{linenomath}
    \begin{align*}
    \pi(\Imatching) := &(\Imatching\setminus F) \cup \{\{a_u^2, b_u^2\}, \{a_u^1, b_u^3\}, \{a_z^2, b_z^2\}, \{a_z^3, b_z^1\} : \{ u, z\} \in F\} \cup\\
    &\{\{u, b_u^1\}, \{ a_z^1, z\} : \{ u, z\} \in F \cap \neu{\oldImatching}\} \cup \{\{a_z^1, b_u^1\}: \{u, z\} \in F \setminus \Imatching\}.
    \end{align*}
  %  \end{linenomath}
    Intuitively, given a matching~$\Imatching$ in~$\mathcal{I}$, the function~$\pi$ computes an ``equivalent'' matching~$\pi (\Imatching)$ in the instance~$\mathcal{J}$. Equivalent here means that \new{a matching~$\Imatching$ containing both forced edges} is popular if and only if $\pi (\Imatching)$ is popular, which we will prove later, in \Cref{cl:Jpop,cl:Ipop}. 
    Furthermore, we define a ``projection'' $\rho$ from \new{each} matching \new{$\Jmatching$} in $\mathcal{J}$ to \new{some} matching in $\mathcal{I}$ via $\rho (\Jmatching) = (\Jmatching \cap E(\mathcal{I})) \cup \{\{u, z\} \in F : \{u, b_u^1\} \in \Jmatching \land \{a_z^1, z\} \in \Jmatching\}$.
    Intuitively, $\rho$ is the inverse of $\pi$.
    Given a matching~$\Jmatching$ in $\mathcal{J}$, matching~$\rho (\Jmatching)$ is popular \new{in} %$\mathcal{J}$}
    \neu{$\mathcal{I}$} if and only if $\Jmatching$ is popular in~$\mathcal{J}$, as we will see later.
    Note that $\rho(\pi(\Imatching)) = \Imatching$ for any matching $\Imatching$ in~$\mathcal{I}$.

  Before proving the correctness of the reduction, we present two helpful claims.
  First, we show that for every matching~$\Imatching$ in $\mathcal{I}$ containing the forced edges, \new{in order to show popularity of $\pi (\Imatching)$ it is sufficient to compare $\pi (\Imatching)$ with~$\pi (\Imatchingprime)$ for  every matching~$\Imatchingprime \in \mathcal{I}$}.
%  we only have to compare $M$ with a subset of other matchings in $\mathcal{I}$ in order to check the popularity of $\pi (M)$ in~$\mathcal{J}$}.

  \begin{numberedclaim}\label{lem:g}
    Let $\Imatching$ be a matching in $\mathcal{I}$ containing both forced edges, and $\Jmatchingprime$ be any matching in~$\mathcal{J}$.
    Let $X:= \{b_u^i, a_z^i: i\in [3]\}$ and $Y:= \{a_u^j, b_z^j: j\in [2]\}$ for a forced edge $\{u, z\}$.
    Then $X\cup Y$ does not prefer $\Jmatchingprime$ to $\Jmatching := \pi (\Imatching)$.

    If $\Jmatchingprime$ is more popular than $\Jmatching$, then in $\mathcal{J}$ there exists a matching $\Jmatchingpp$ that is more popular than $\Jmatching$ and for each forced edge $\{u, z\}$, the matching~$\Jmatchingpp$ contains either the edge~$\{a_z^1, b_u^1\}$ or the edges $\{u, b_u^1\}$ and $\{a_z^1, z\}$.
  \end{numberedclaim}

  \begin{claimproof}{lem:g}
    First, we show that~$X \cup Y$ does not prefer~$\Jmatchingprime$ to $\Jmatching$.
    In order to do so, we make a case distinction \new{based on which edges are contained in $\Jmatchingprime$.}%on whether~$\{a_z^1, b_u^1\}$ is in~$M'$.

    \new{
    \textbf{Case 1:} $\{a_z^1, b_u^1\} \in \Jmatchingprime$.
    
    T}hen $a_z^1$ prefers~$\Jmatchingprime$ while $b_u^1$ prefers~$\Jmatching$, and thus, their votes cancel out.

    \new{
    \textbf{Case 1.1:} $\{a_u^1, b_u^2\} \in \Jmatchingprime$.
    
    T}hen~$a_u^1$ and~$b_u^2$ prefer~$\Jmatchingprime$, while $b_u^3$ and $a_u^2$ prefer~$\Jmatching$.
    
    \new{\textbf{Case 1.2:} $\{a_u^1, b_u^2\} \notin \Jmatchingprime$.
    
    Then} none of the vertices $a_u^1$, $a_u^2$, $b_u^2$, and $b_u^3$ prefers~$\Jmatchingprime$.

    \new{In each of Case~1.1 and Case~1.2,} $\{a_u^1, a_u^2, b_u^2, b_u^3\}$ does not prefer~$\Jmatchingprime$ to~$\Jmatching$.
    Symmetric arguments show that $\{a_z^\new{2}, a_z^3, b_z^1, b_z^2\}$ does not prefer~$\Jmatchingprime$, and therefore, $X\cup Y$ does not prefer~$\Jmatchingprime$ to~$\Jmatching$.

    \new{
    \textbf{Case 2:} $\{a_z^1, b_u^1\} \notin \Jmatchingprime$.
    
    T}hen neither~$a_z^1$ nor $b_u^1$ prefers~$\Jmatchingprime$ to $\Jmatching$.
    We first look at vertices~$a_u^1$, $a_u^2$, $b_u^1$, $b_u^2$, and $b_u^3$.

    \new{
    \textbf{Case 2.1:} $\{u, b_u^1\} \in \Jmatchingprime$.
    
    T}hen we already saw in \new{Case~1.1} that $\{a_u^1, a_u^2, b_u^2, b_u^3\}$ does not prefer $\Jmatchingprime$ to $\Jmatching$.

    \new{
    \textbf{Case 2.2:} $\{u , b_u^1\} \notin \Jmatchingprime$.
    
    Then}~$b_u^1$ prefers~$\Jmatching$ to~$\Jmatchingprime$.
    Note that~$b_u^2$ can prefer~$\Jmatchingprime$ to $\Jmatching$ only if $\{a_u^1, b_u^2\} \in \Jmatchingprime$, which implies that $\Jmatching \succ_{b_u^3} \Jmatchingprime$.
    Consequently, vertices~$b_u^1$, $b_u^2$, and~$b_u^3$ \new{contribute to the vote} for~$\Jmatching$ by weight at least 4, while $a_u^1$ and $a_u^2$ only cast votes of summed weight 3.
    Therefore, $\{a_u^1, a_u^2, b_u^1, b_u^2, b_u^3\}$ does not prefer $\Jmatchingprime$ to $\Jmatching$.
    
    Symmetric arguments show that also $\{a_z^1, a_z^2, a_z^3, b_z^1, b_z^2\}$ does not prefer~$\Jmatchingprime$ to~$\Jmatching$.

    We now prove the second part of the claim.
    Let~$\{u, z\} $ be a forced edge.
    If $\Jmatchingprime $ contains either~$\{a_z^1, b_u^1\}$ or both of $\{u, b_u^1\}$ and $\{a_z^1, z\}$, then there is nothing to show.
    So assume that $\Jmatchingprime$ contains neither~$\{a_z^1, b_u^1\}$ nor both of $\{u, b_u^1\}$ and $\{a_z^1, z\}$.
    \new{We make a case distinction depending on %how
    \neu{what} $\Jmatchingprime$ looks like.
    In each case, we will construct another matching~$\Jmatchingpp$ which is also more popular than $\Jmatching$.
    In order to show that $\Jmatchingpp$ is more popular than $\Jmatching$, we will show that $\Delta^*:= \Delta_w (\Jmatching, \Jmatchingpp) - \Delta_w (\Jmatching, \Jmatchingprime) \le 0$.
    This implies that $\Delta_w (\Jmatching, \Jmatchingpp)  = \Delta_w (\neu{\oldJmatching, \oldJmatchingprime}) + \Delta^* \le  \Delta_w (\Jmatching, \Jmatchingprime) < 0$, i.e.\ that $\Jmatchingpp$ is more popular than $\Jmatching$.}

    \new{\textbf{Case 1:} $b_u^1$ is unmatched in $\Jmatchingprime$.

    Then the matching~$\Jmatchingpp$ arising from~$\Jmatchingprime $ by adding $\{a_z^1, b_u^1\}$ to~$\Jmatchingprime$ (and possibly deleting an edge incident to~$a_z^1$) is also more popular than $\Jmatching$:
    \neu{Agent~$b_u^1$ prefers $\oldJmatching$ to both $\oldJmatchingprime$ and $\oldJmatchingpp$, so $b_u^1$ contributes 0 to $\Delta^*$.
    Agent~$a_z^1$ prefers $\oldJmatchingpp (a_z^1) = b_u^1$ to $\oldJmatching (a_z^1) = z$ but does not prefer $\oldJmatchingprime (a_z^1)$ to~$\oldJmatching (a_z^1)$.
    Thus, $a_z^1$ contributes at most $- w(a_z^1 ) = - 4$ to~$\Delta^*$.}
%    Agent~$b_u^1$ prefers~$\Jmatchingpp$ to $\Jmatching$ but prefers $\Jmatching$ to~$\Jmatchingprime$, implying that $b_u^1$ contributes $-2 w(b_u^1) = -8$ to $\Delta^*$. \ronecom{P.14,l.362: $b^1_u$ is matched to its best choice in $M_{J}$. How can $M''_J$ be preferred?}
 %   Agent~$a^1_z$ prefers~$\Jmatchingpp$ to~$\Jmatchingprime$, implying that $a^1_z$ contributes at most 0 to~$\Delta^*$.
    Finally, $\Jmatchingprime (a_z^1)$ can contribute at most $2 w(\neu{\oldJmatchingprime (a_z^1)}) = 2$ to~$\Delta^*$.
    Thus, we have $\Delta^* \le \neu{-w (a_z^1) + 2w(\oldJmatchingprime (a_z^1))  = -2}$.
\iffalse
    \begin{linenomath}
    \[
        \Delta_w (\Jmatching, \Jmatchingpp ) \le \Delta_w (\Jmatching, \Jmatchingprime)  - 2w (b_u^1) + 2w(\Jmatchingprime (a_z^1)) = \Delta_w (\new{\Jmatching, \Jmatchingprime}) -6 <0\,.
    \]
    \end{linenomath}
\fi
    }

    \new{\textbf{Case 2:}
    $\Jmatchingprime (b_u^1) = a_u^i$ for some~$i \in \{1,2\}$.

    Then there exists some~$j \in \{2,3\}$ such that $b_u^j$ is unmatched in~$\Jmatchingprime$.}
    The matching~$\Jmatchingpp$ arising from~$\Jmatchingprime$ by adding~$\{a_z^1, b_u^1\}$ \new{and $\{a_u^i, b_u^j\}$} to~$\Jmatchingprime$ (and possibly deleting an edge incident to~$a_z^1$) results in a matching more popular than~$\Jmatching$:
    Only $b_u^1$, $b_u^j$, $a_u^i$, $a_z^1$, and $\Jmatchingprime (a_z^1)$ may vote differently between~$\Jmatchingprime$ and~$\Jmatching$ compared to their vote between~$\Jmatchingpp$ and $\Jmatching$.
%    have $\vote^{M'}_a (M_{\mathcal{J}} (a)) \neq \vote^{M''}_a (M_{\mathcal{J}} (a))$.
    Agent~$b_u^1$ prefers~$\Jmatching$ to both~$\Jmatchingprime$ and $\Jmatchingpp$.
    \new{Thus, $b_u^1 \in V^+ (\Jmatching, \Jmatchingprime)$ if and only if $b_u^1 \in V^+ (\Jmatching, \Jmatchingpp)$ implying that $b_u^1$ contributes 0 to $\Delta^*$.}
    \new{
    Agent~$b_u^j$ is unmatched in~$\Jmatchingprime$ and therefore prefers both~$\Jmatchingpp$ and $\Jmatching$ to $\Jmatchingprime$.
    Thus, $b_u^j$ contributes at most~$0$ to $\Delta^*$.
    Agent~$a_u^j$ is matched to its last choice in~$\Jmatching$ and thus does not prefer~$\Jmatching$ to $\Jmatchingprime$ or~$\Jmatchingpp$.
    Thus, $a_u^j$ contributes at most $w(a_u^j)\le 2$ to~$\Delta^*$.
    }
%    Agent~$M' (b_u^1)$ \new{can have arbitrary preferences between $M'$, $M''$, and $M_\mathcal{J}$.
%    Therefore, $M' (b_u^1)$ contributes at most $2 w( M' (b_u^1))$ to $\Delta_w (M_{\mathcal{J}}, M'') - \Delta_w (M_{\mathcal{J}}, M')$.}
    Agent~$a_z^1$ prefers~$\Jmatchingpp$ to~$\Jmatching$.
    \new{Therefore, $a_z^1$ contributes $0$ to $\Delta^*$ if $a_z^1$  %is
    prefers~$\Jmatchingprime$ to $\Jmatching$, $-w(a_z^1)$ if $a_z^1$ is indifferent between $\Jmatchingprime $ and $\Jmatching$, and $-2 w(a_z^1)$ otherwise.} \neu{From here on, we consider two subcases.}

    \new{\textbf{Case 2.1:} $\{a_z^1, z\} \in \Jmatchingprime$.
    
    T}hen $a_z^1$ and $z $ are indifferent between~$\Jmatchingprime$ and $\Jmatching$.
    \new{Thus, $a_z^1 $ contributes $w(a_z^1)=-4$ to $\Delta^*$ and $z$ contributes $w(z) = 1$ to $\Delta^*$. Therefore, we have $\Delta^* \le 2 -4 + 1  = - 1$.}
\iffalse
%    , implying that
    \begin{linenomath}
    \[
        \Delta_w (\Jmatching, \Jmatchingpp ) \le \Delta_w (\Jmatching, \Jmatchingprime) \new{+w(a^j_u)} - w (a_z^1) \new{+w(z)} \le \Delta_w (\new{\Jmatching, \Jmatchingprime}) \new{-1}< 0
    \]
    \end{linenomath}
    implying that also $\Jmatchingpp$ is more popular than~$\Jmatching$.
\fi
    
    \new{\textbf{Case 2.2:} $\{a_z^1, z\} \notin \Jmatchingprime$.
    
    T}hen $a_z^1$ prefers~$\Jmatchingpp$ to~$\Jmatching$ \new{but prefers $\Jmatching$} to~$\Jmatchingprime$.
    \new{Thus, $a_z^1$ contributes $-2 w(a_z^1) = 8$ to~$\Delta^*$.}
 \new{Agent~$\Jmatchingprime (a_z^1)$ can have arbitrary preferences between~$\Jmatchingprime$, $\Jmatchingpp$, and~$\Jmatching$ and \neu{hence $M'_\mathcal{J} (a_1^z)$} contributes at most $2 w(\Jmatchingprime (a_z^1)) = 2$ to~$\Delta^*$.}
    Thus, we have \new{$\Delta^* \le w(a_u^j) - 2 w(a_z^1) + 2 w(\Jmatchingprime (a_z^1)) \le 2-8+2 = -4 $}.
%    Thus, only $M (b_u^1)$ may prefer~$M_{\mathcal{J}}$ to
%    but does not prefer~$M'$ to~$M_\mathcal{J}$ to~$M'$. Further, at most~$M' (a_z^1)$ and $M' (b_u^1)$ prefer~$M'$ to~$M_\mathcal{J}$ and $M_\mathcal{J}$ to $M''$.
%    Consequently, we have (note that the weight of~$v \in \{a_z^1, M' (a_z^1), M'(b_u^1)\}$ is multiplied by~2 as its votes may change from $+w(v)$ to $-w(v)$) 
\iffalse
\begin{linenomath}
    \begin{align*}
        \Delta_w (\Jmatching, \Jmatchingpp) & \le \Delta_w (\Jmatching, \Jmatchingprime) \new{+w(a_u^j)} - 2 w(a_z^1) + 2 w(\Jmatchingprime (a_z^1))\\
        & < \Delta_w (\new{\Jmatching, \Jmatchingprime}) < 0
    \end{align*}
    \end{linenomath}
    implying that also $\Jmatchingpp$ is more popular than~$\Jmatching$.
\fi
    
%    If $M'$ contains~$\{a_z^1, z\}$ but not $\{u, b_u^1\}$, then the matching~$M''$ arising from $M'$ by adding~$\{a_z^1, b_u^1\}$ to $M'$ (and deleting~$\{a_z^1, z\}$ as well as possibly an edge incident to~$b_u^1$) is also more popular than~$M_\mathcal{J}$ as~$a_z^1$ now votes for~$M''$ but prefers $M_\mathcal{J}$ to~$M'$ while at most~$M' (a_z^1) = z$ and $M' (b_u^1)$ prefer~$M_\mathcal{J}$ to $M''$ while preferring~$M'$ to $M_\mathcal{J}$.
    
%    Finally, we consider the case that
    \new{
    \textbf{Case 3:}
    $\Jmatchingprime$ contains~$\{u, b_u^1\}$.
    }

    Then the matching~$\Jmatchingpp$ arising from $\Jmatchingprime$ by adding~$\{a_z^1, b_u^1\}$ to $\Jmatchingprime$ (and deleting~$\{u, b_u^1\}$ as well as possibly an edge incident to~$a_z^1$) is also more popular than~$\Jmatching$:
    Agent~$b_u^1$ is indifferent between $\Jmatching$ and $\Jmatchingprime$ but prefers $\Jmatching$ to~$\Jmatchingpp$.
    \new{Thus, $b_u^1$ contributes $w(b_u^1)$ to $\Delta^*$.}
    Agent~$a_z^1$ prefers~$\Jmatchingpp$ to~$\Jmatching$ but prefers $\Jmatching$ to~$\Jmatchingprime$.
    \new{Thus, $a_z^1$ contributes $-2 w(a_z^1)$ to $\Delta^*$.}
    \new{Agent~$\Jmatchingprime(a_z^1)$ \neu{(respectively $\oldJmatchingprime (b_u^1) = u$)} contributes at most $2 w(\neu{\oldJmatchingprime (a_z^1)})$ \neu{(respectively $2 w(\neu{\oldJmatchingprime (b_u^1)}) = 2$)} to $\Delta^*$.}
    Consequently, we have \new{$\Delta^* \le - 2 w(a_z^1) + w(b_u^1) + 2 w(\Jmatchingprime (a_z^1)) + 2 w(\Jmatchingprime (b_u^1)) = 0$ (using $w(\Jmatchingprime(b_u^1)\neu{)} =1 $ as $\Jmatchingprime (b_u^1)  = u$ for the last equality).}
\iffalse
    \begin{linenomath}
    \begin{align*}
        \Delta_w (\Jmatching, \Jmatchingpp) & \le \Delta_w (\Jmatching, \Jmatchingprime) - 2 w(a_z^1) + w(b_u^1) + 2 w(\Jmatchingprime (a_z^1)) + 2 w(\Jmatchingprime (b_u^1)) \\
        & = \Delta_w (\new{\Jmatching, \Jmatchingprime}) < 0
    \end{align*}
    \end{linenomath}
    implying that also $\Jmatchingpp$ is more popular than~$\Jmatching$.
\fi
  \end{claimproof}

  We now show that ``projecting'' two matchings~$\Imatching$ and $\Imatchingprime$ from $\mathcal{I}$ to $\mathcal{J}$ via $\pi$ does not influence whether $\Imatching$ is more popular than $\Imatchingprime$ or not.
  This will be used to prove that for a popular matching~$\Jmatching$ in $\mathcal{J}$, matching~$\rho (\Jmatching)$ is popular in~$\mathcal{I}$.
  
  \begin{numberedclaim}\label{lem:f}
    For any two matchings $\Imatching$ and $\Imatchingprime$ in $\mathcal{I}$, we have $\Delta (\Imatching, \Imatchingprime) = \Delta_w (\pi(\Imatching), \pi(\Imatchingprime))$.
  \end{numberedclaim}

  \begin{claimproof}{lem:f}
    Consider a forced edge $\{u, z\}$ that appears in $\Imatching$, but not in $\Imatchingprime$.
    Then $\pi (\Imatching) \succ_{b_u^1} \pi (\Imatchingprime)$ and $\pi (\Imatchingprime) \succ_{a_z^1} \pi (\Imatching)$ hold.
    Similarly, if the forced edge $\{u, z\}$ is contained in $\Imatchingprime$, but not in~$\Imatching$,
    then $\pi (\Imatchingprime) \succ_{b_u^1} \pi (\Imatching)$ and $\pi (\Imatching) \succ_{a_z^1} \pi (\Imatchingprime)$ hold.

    Since $b_u^1 $ and $a_z^1$ have the same weight, and all other vertices are matched identically in~$\Imatching$ and~$\Imatchingprime$, the claim follows.
  \end{claimproof}

\textbf{Correctness.} We are now ready to prove the correctness of the reduction.
\begin{numberedclaim}\label{cl:Jpop}
  For each $\Imatching$ in $\mathcal{I}$ that is a popular matching containing the set~$F$ of forced edges, there is a popular matching in~$\mathcal{J}$.
\end{numberedclaim}
    \begin{claimproof}{cl:Jpop}
    Let $\Imatching$ be a popular matching in $\mathcal{I}$ containing the set~$F$ of forced edges.  We claim that $\Jmatching := \pi(\Imatching)$ is popular in $\mathcal{J}$.
    Assume for a contradiction that $\Jmatching$ \new{is less popular than} a matching~$\Jmatchingprime$.
    By \Cref{lem:g}, we may assume that for each forced edge $\{u, z\}$, matching~$\Jmatchingprime$ contains either edges~$\{u, b_u^1\} $ and $\{a_z^1, z\}$ or edge~$\{a_z^1, b_z^1\}$.
    We define $\Imatchingprime := \rho(\Jmatchingprime)$, and claim that $\Imatchingprime$ \new{is more popular than}~$\Imatching$, a contradiction to the popularity of~$\Imatching$.

    Since $\Jmatchingprime$ contains either edges~$\{u, b_u^1\}$ and $\{a_z^1, z\}$ or edge~$\{a_z^1, b_z^1\}$ for every forced edge~$\{u, z\}$, every vertex in $A \cup B$ votes between $\Imatching$ and~$\Imatchingprime$ the same as between $\Jmatching$ and $\Jmatchingprime$.
    Furthermore, the set~$S$ of vertices from the edge gadget do not prefer $\Jmatchingprime$ to $\Jmatching$ by \Cref{lem:g}.
    It follows that the vertices in~$V$ prefer $\Jmatchingprime$ to $\Jmatching$, and thus, $\Imatchingprime$ is more popular than $\Imatching$.
    \end{claimproof}

\begin{numberedclaim}\label{cl:Ipop}
  For each popular matching $\Jmatching$ in~$\mathcal{J}$, there is a popular matching in $\mathcal{I}$ that contains the set~$F$ of forced edges.
\end{numberedclaim}
\begin{claimproof}{cl:Ipop}
    Let $\Jmatching$ be a popular matching in $\mathcal{J}$. We first show by case distinction that for each forced edge $\{u, z\}\in F$, matching~$\Jmatching$ contains edges $\{u, b_u^1\}$ and~$\{a_z^1, z\}$.
    We assume for a contradiction that $\Jmatching$ does not contain these edges, and distinguish three cases.
    For each of the three cases, we construct a more popular matching, contradicting the popularity of~$\Jmatching$.
    \begin{labeling}{Case 3:}
     \item[Case 1:] Vertex $b_u^1$ is matched neither to $u$ nor to~$a_z^1$.\\
     In this case, we construct a more popular matching $\Jmatchingprime$ by replacing the edges $\{a_u^1, \Jmatching (a_u^1)\} $ and $\{a_u^2, \Jmatching(a_u^2)\}$ by the edges $\{a_u^1, \Jmatching(a_u^2)\}$ and $\{a_u^2, b_u^j\}$, where $b_u^j$ is unmatched in $\Jmatching$.
     The matching $\Jmatchingprime$ is more popular than~$\Jmatching$ as $b_u^j$ and $\Jmatching(a_u^2)$ prefer $\Jmatchingprime$ while at most $\Jmatching(a_u^1)$, %vertex~$a_u^1$, and $a_u^2$ 
     \neu{vertices $a_u^1$ and $a_u^2$} prefer $\Jmatching$ to $\Jmatchingprime$, and it holds that $w(b_u^j) + w(\Jmatching(a_u^2)) = 8 > w(\Jmatching(a_u^1)) + w(a_u^1) + w(a_u^2) = 7$.

     \item[Case 2:] Vertex $a_z^1$ is matched neither to $z$ nor to $b_u^1$.\\
     This case is symmetric to Case 1.

     \item[Case 3:] Vertex $b_u^1$ is matched to $a_z^1$.\\
     \new{In this case, we construct a matching $\Jmatchingprime$ \new{from $\Jmatching$} by matching $b_u^1$ to $a_u^2$, and $a_u^1 $ to~$\Jmatching(a_u^2)$. Matching $\Jmatchingprime$ is more popular than $\Jmatching$, because \begin{itemize}
     \item vertices $\Jmatching(a_u^2)$, $b_u^1$, and $a_u^2$ prefer $\Jmatchingprime$, while at most $a_z^1$, vertex~$b_u^{j}$ for at most one $j\in \{2, 3\}$, and~$a_u^1$ prefer $\Jmatching$,
     \item and $w(b_u^1 ) = w (a_z^1) = w(\Jmatching (a_u^2)) = w( b_u^{j})$, and $w(a_u^2) > w(a_u^1)$.
     \end{itemize}}
     %In this case, we construct a more popular matching $M'$ by matching $b_u^1$ to $a_u^2$, and~$a_u^1 $ to $M(a_u^2)$, as $M(a_u^2)$, vertex~$b_u^1$, and $a_u^2$ prefer $M'$, while at most $a_z^1$, vertex~$b_u^{j}$ for at most one $j\in \{2, 3\}$, and~$a_u^1$ prefer $M$, and $w(b_u^1 ) = w (a_z^1) = w(M (a_u^2)) = w( b_u^{j})$, and $w(a_u^2) > w(a_u^1)$.
    \end{labeling}
    
    Thus,~$\Imatching := \rho(\Jmatching)$ contains every forced edge.
    It remains to show that $\Imatching$ is popular.
    For a contradiction, assume that $\Imatchingprime$ is more popular than $\Imatching$.
    We claim that $\Jmatchingprime := \pi (\Imatchingprime)$ wins against~$\Jmatching$, contradicting the popularity of~$\Jmatching$.

    For every forced edge~$\{u, z\}$, matching~$\Jmatching$ has to contain the edges $\{a_u^2, b_u^2\}$ and $\{a_u^1, b_u^3\}$, as \new{otherwise} matching~$b_u^1$ and $b_u^2$ along these edges yields a more popular matching.
    \neu{Similarly, $\oldJmatching$ contains edges~$\{a_z^2, b_z^2\}$ and~$\{a_z^3, b_z^1\}$ or edges~$\{a_z^2, b_z^1\}$ and $\{a_z^3, b_z^2\}$ as otherwise $\oldJmatching$ would not be popular.
    In both cases, the votes of $a_z^2$, $a_z^3$, $b_z^1$, and $b_z^2$ sum up to zero.
    Thus, by the same arguments as in the proof of \Cref{lem:f}, we have $ 0 < \Delta_w (\oldImatchingprime, \oldImatching) = \Delta_w (\oldJmatchingprime, \oldJmatching)$ (using that $\oldImatchingprime$ is more popular than $\oldImatching$ for the inequality).
    This contradicts the popularity of $\oldJmatching$.}
%   Symmetrically, $\Jmatching$ contains edges~$\{a_z^2, b_z^2\}$ and~$\{a_z^3, b_z^1\}$. \ronecom{P.17,l.456: Why could $M_J$ not also contain the edges ${a3z, b2z}$ and ${a1z, b2z}$?}
%   Thus, we have $\Jmatching = \pi(\Imatching)$.
%   Since $\Imatchingprime$ is more popular than~$\Imatching$ and by \Cref{lem:f}, we have $0 < \Delta_w (\new{\Imatchingprime, \oldImatching) = \Delta_w (\oldJmatchingprime, \oldJmatching})$, \rcom{Line 457: $\Delta_w(M_\mathcal{I}, M'_\mathcal{I}) \rightarrow \Delta_w(M'_\mathcal{I}, M_\mathcal{I})$} contradicting the popularity of~$\Jmatching$.
  \end{claimproof}
  
  This finishes the proof of \Cref{thm:np-harda}.
  \end{proof}
  
%%%%%%%%%%%%%%%%

\subsection{Identical weights on each side}
\label{sec:np:identical-weights}

We now consider the variant that all agents on one side have weight one while all agents on the other side have weight $c$ for some~$1 < c \le 2$.
To show that also this variant of \textsc{Popular Matching with Weighted Voters} is \NP-complete, we reduce from \textsc{3-SAT}.
As usual, the reduction is based on variable and clause gadgets.
However, the reduction also uses another gadget which we call "6-path gadget" (one might view this gadget also as part of the variable gadgets).
The idea behind the 6-path gadget is to ensure that the two solutions from the variable gadgets do not weakly dominate each other (where weak dominance is defined as in \Cref{sec:witness}).

\begin{theorem}\label{thm:np-hardb}
    For any constant ${1 < c \le 2}$, \textsc{Popular Matching with\linebreak Weighted Voters} is \NP-complete even if $w (a) = c$ for all $a\in A$ and $w (b) = 1 $ for all $b\in B$.
\end{theorem}

%%%%%%%%%%%%%

%\nphardb*

  \begin{proof}
    Membership in \NP\ was shown in \Cref{cl:zero_weight}.
    Fix a constant $1 < c \le 2$.
    To show \NP-hardness, we reduce from \textsc{3-SAT}, the restriction of \textsc{Satisfiability} to instances where every clause contains exactly three literals. Let~$X$ be the set of variables.
    
  \textbf{Construction.} We first describe the 3 types of gadgets in our reduction.
  
  \textbf{\unnamedGadget.}
  The \unnamedGadget\ consists of a path on 6 vertices $a_1$, $a_2$, $a_3$, $b_1$, $b_2$, and~$b_3$.
  The preferences of the vertices are described in \Cref{fig:ug} (where vertices $a_x$ for every~$x \in X$ will be defined when describing variable gadgets).
  
 \begin{figure}
     \centering
     \begin{minipage}{0.4\textwidth}
		\[
		\begin{array}{rl}
    a_1 & : b_1  \\
    a_2 & : b_1 \succ b_2\\\vspace{4mm}
    a_3 &: b_2 \succ b_3\\
    b_1 &: a_1 \succ a_2\\
    b_2 & : a_2 \succ a_3\\
    b_3 &: \{a_x : x\in X\} \succ a_3\\
    &\\
        a_x & : b_x^t \succ b_x^f \succ b_3\\ \vspace{4mm}
        \bar a_x & : b_x^t \succ b_x^f \\ 
        b_x^t & : a_x \succ \bar a_x\\
        b_x^f & : a_x \succ \bar a_x
		\end{array}
		\]
	\end{minipage}
	\begin{minipage}{0.55\textwidth}	
     \begin{tikzpicture}
          \node[Avertex, label=180:$a_1$] (a1) at (-1,1) {};
          \node[Avertex, label=180:${a}_2$] (a2) at (-1,-1) {};
          \node[Avertex, label=180:${a}_3$] (a3) at (-1,-3) {};
          \node[Bvertex, label=0:$b_1$] (b1) at (1.5,1) {};
          \node[Bvertex, label=0:$b_2$] (b2) at (1.5,-1) {};
          \node[Bvertex, label=270:$b_3$] (b3) at (1.5,-3) {};

        \draw (a1) edge node[pos=0.2, fill=white, inner sep=2pt] {\scriptsize $1$}  node[pos=0.76, fill=white, inner sep=2pt] {\scriptsize $1$} (b1);
        
        \draw (a2) edge node[pos=0.2, fill=white, inner sep=2pt] {\scriptsize $1$}  node[pos=0.76, fill=white, inner sep=2pt] {\scriptsize $2$} (b1);
        \draw (a2) edge node[pos=0.2, fill=white, inner sep=2pt] {\scriptsize $2$}  node[pos=0.76, fill=white, inner sep=2pt] {\scriptsize $1$} (b2);

        \draw (a3) edge node[pos=0.2, fill=white, inner sep=2pt] {\scriptsize $1$}  node[pos=0.76, fill=white, inner sep=2pt] {\scriptsize $2$} (b2);
        \draw (a3) edge node[pos=0.2, fill=white, inner sep=2pt] {\scriptsize $2$}  node[pos=0.76, fill=white, inner sep=2pt] {\scriptsize $|X|+1$} (b3);
        
        \begin{scope}[xshift=4.5cm, yshift = -3 cm]
          \node[terminal, label=270:$a_x$] (xv2) at (-1,1) {};
          \node[terminal, label=90:$\bar{a}_x$] (xv3) at (-1,3) {};
          \node[vertex, label=0:$b_x^t$] (yv1) at (1.,1) {};
          \node[vertex, label=0:$b_x^f$] (yv2) at (1.,3) {};
        \end{scope}

        \draw (xv2) edge node[pos=0.2, fill=white, inner sep=2pt] {\scriptsize $1$}  node[pos=0.76, fill=white, inner sep=2pt] {\scriptsize $1$} (yv1);
        \draw (xv2) edge node[pos=0.2, fill=white, inner sep=2pt] {\scriptsize $2$}  node[pos=0.76, fill=white, inner sep=2pt] {\scriptsize $1$} (yv2);

        \draw (xv3) edge node[pos=0.2, fill=white, inner sep=2pt] {\scriptsize $1$}  node[pos=0.76, fill=white, inner sep=2pt] {\scriptsize $2$} (yv1);
        \draw (xv3) edge node[pos=0.2, fill=white, inner sep=2pt] {\scriptsize $2$}  node[pos=0.76, fill=white, inner sep=2pt] {\scriptsize $2$} (yv2);
        
        \draw (b3) edge node[pos=0.2, fill=white, inner sep=2pt] {\scriptsize $5$}  node[pos=0.76, fill=white, inner sep=2pt] {\scriptsize $3$} (xv2);
     \end{tikzpicture}
     \end{minipage}
     \caption{The \unnamedGadget\ (left) and the variable gadget for a variable~$x$ (right). Vertices in $A$ are squared.
     \new{Agent~$b_3$ ranks the agents~$a_x$ ($x \in X$) from the variable gadgets at rank $1$ to~$ |X|$. For concreteness, we assign rank 5 to the agent~$a_x$ in the above example.}}
     \label{fig:ug}\label{fig:vg}
 \end{figure}
  
  \textbf{Variable gadgets.} An example of a variable gadget is depicted in \Cref{fig:vg}.
    The variable gadget for a variable $x$ contains four vertices $a_x$, $\bar{a}_x$, $b_x^t$, and $b_x^f$.
    Vertex~$a_x$ as well as $\bar{a}_x$ prefers $b_x^t$ to $b_x^f$, and vertices~$b_x^t$ and $b_x^f$ prefer $a_x $ to $\bar{a}_x$. Intuitively, matching $a_x$ to $b_x^t$ (and therefore also $\bar{a}_x $ to $b_x^f$) corresponds to setting variable $x$ to $\LogicTRUE$, while matching $\bar{a}_x$ to $b_x^t$ (and therefore also $a_x$ to $b_x^f$) corresponds to setting variable $x$ to $\LogicFALSE$.
\begin{comment}
        \begin{figure}
      \centering
       \begin{minipage}{0.4\textwidth}
		\[
		\begin{array}{rl}
        a_x & : b_x^t \succ b_x^f \succ b_3\\ \vspace{4mm}
        \bar a_x & : b_x^t \succ b_x^f \\ 
        b_x^t & : a_x \succ \bar a_x\\
        b_x^f & : a_x \succ \bar a_x
		\end{array}
		\]
	\end{minipage}\begin{minipage}{0.4\textwidth}	
        \begin{tikzpicture}
          \node[terminal, label=180:$a_x$] (xv2) at (-1,1) {};
          \node[terminal, label=180:$\bar{a}_x$] (xv3) at (-1,3) {};
          \node[vertex, label=0:$b_x^t$] (yv1) at (1.,1) {};
          \node[vertex, label=0:$b_x^f$] (yv2) at (1.,3) {};

        \draw (xv2) edge node[pos=0.2, fill=white, inner sep=2pt] {\scriptsize $1$}  node[pos=0.76, fill=white, inner sep=2pt] {\scriptsize $1$} (yv1);
        \draw (xv2) edge node[pos=0.2, fill=white, inner sep=2pt] {\scriptsize $2$}  node[pos=0.76, fill=white, inner sep=2pt] {\scriptsize $1$} (yv2);

        \draw (xv3) edge node[pos=0.2, fill=white, inner sep=2pt] {\scriptsize $1$}  node[pos=0.76, fill=white, inner sep=2pt] {\scriptsize $2$} (yv1);
        \draw (xv3) edge node[pos=0.2, fill=white, inner sep=2pt] {\scriptsize $2$}  node[pos=0.76, fill=white, inner sep=2pt] {\scriptsize $2$} (yv2);
        \end{tikzpicture}
    \end{minipage}
      \caption{The variable gadget. Vertices in $A$ are squared.}
      \label{fig:vg}
    \end{figure}
\end{comment}
    
Every variable gadget is connected to the \unnamedGadget\ via edge~$\{a_x, b_3\}$, where $a_x$ prefers both~$b^f_x$ and $b^t_x$ to~$b_3$, and $b_3$ prefers $a_x$ to~$a_3$ (the preferences of $b_3$ between two vertices $a_x$ and $a_{x'}$ are arbitrary).
    
\textbf{Clause gadgets.} An example of a clause gadget is depicted in \Cref{fig:clause-gadget}.
A clause gadget for clause~$C$ contains twelve vertices $a^C_{k}$, $\hat a^C_{k}$, $b^C_{ k}$, and $\hat b^C_{ k}$ for all $k\in [3]$.
The preferences of these vertices over other vertices of the clause gadget are described in \Cref{fig:clause-gadget}.
There are several popular matchings inside the depicted gadget.
Our proof will use the following three popular matchings, indicating that the first, second, or third literal of the clause is satisfied.
%\begin{linenomath}
  \begin{align*}      
  M_1^C &=\left\{ \{\hat{a}^C_{1}, b^C_1\}, \{a^C_{1}, \hat{b}^C_1\}, \{\hat{a}^C_{2}, \hat{b}^C_2\}, \{a^C_{2}, b^C_2\}, \{\hat{a}^C_{3}, \hat{b}^C_3\}, \{a^C_{3}, b^C_3\}\right\}\\
  M_2^C &=\left\{ \{\hat{a}^C_{1}, \hat{b}^C_1\}, \{a^C_{1}, b^C_1\}, \{\hat{a}^C_{2}, b^C_2\}, \{a^C_{2}, \hat{b}^C_2\}, \{\hat{a}^C_{3}, \hat{b}^C_3\}, \{a^C_{3}, b^C_3\}\right\}\\
  M_3^C &=\left\{ \{\hat{a}^C_{1}, \hat{b}^C_1\}, \{a^C_{1}, b^C_1\}, \{\hat{a}^C_{2}, \hat{b}^C_2\}, \{a^C_{2}, b^C_2\}, \{\hat{a}^C_{3}, b^C_3\}, \{a^C_{3}, \hat{b}^C_3\}\right\}
  \end{align*}
%  \end{linenomath}
  In \Cref{fig:clause-gadget}, these are marked by dotted, thick, and gray edges, in this order. Their three witnesses are shown next to the respective vertices, in green, red, and gray, in this order for $M_1^C, M_2^C$, and~$M_3^C$. Intuitively, the $k$-th literal of clause~$C$ is selected to be satisfied if the popular matching contains edges $\{a^C_k, \hat b^C_k\}$ and $\{\hat a^C_k, b^C_k\}$, while a matching containing edges~$\{a^C_i, b^C_i\}$ and $\{\hat a^C_i, \hat b^C_i\}$ poses no condition on whether the $i$-th literal is satisfied. The clause gadget is designed in such a way that it selects an arbitrary literal of the clause to be $\LogicTRUE$.

\begin{figure}
\centering
 %      \begin{minipage}{0.3\textwidth}
%	\end{minipage}
%	\begin{minipage}{0.6\textwidth}
  \begin{tikzpicture}
          \node[terminal, label=270:$a^C_{1}$, label={[xshift=-20pt]180:\color{green}1, \color{red}$2-c$, \color{gray}$2-c$}] (a1) at (-1,1) {};
          \node[terminal, label=180:$\hat{a}^C_{1}$, label={[xshift=-20pt]180:\color{green}$-c$, \color{red}$1$, \color{gray}$1$}] (a1p) at (-1,3) {};
          \node[vertex, label={[xshift=-.2cm]275:$b^C_{ 1}$}, label={[xshift=20pt]0:\color{green}$c$, \color{red}$c-2$, \color{gray}$c-2$}] (b1) at (1.,1) {};
          \node[vertex, label=0:$\hat{b}^C_{1}$, label={[xshift=20pt]0:\color{green}$-1$, \color{red}$-1$, \color{gray}$-1$}] (b1p) at (1.,3) {};

        \node[terminal, label=180:$a^C_{2}$, label={[xshift=-25pt]180:\color{green}$-c$, \color{red}$1$, \color{gray}$-c$}] (a2) at ($(a1) + (0, -4)$) {};
        \node[terminal, label={[xshift=-5pt]180:$\hat{a}^C_{2}$}, label={[xshift=-20pt]180:\color{green}$-1$, \color{red}$-c$, \color{gray}$-1$}] (a2p) at ($(a2) + (0, 2)$) {};
        \node[vertex, label=0:$b^C_{ 2}$, label={[xshift=20pt]0:\color{green}$c$, \color{red}$c$, \color{gray}$c$}] (b2) at ($(a2) + (2, 0)$) {};
        \node[vertex, label={[xshift=-.1cm]0:$\hat{b}^C_{ 2}$}, label={[xshift=20pt]0:\color{green}$1$, \color{red}$-1$, \color{gray}$1$}] (b2p) at ($(b2) + (0, 2)$) {};

        \node[terminal, label=270:$a^C_{3}$, label={[xshift=-20pt]180:\color{green}$-c$, \color{red}$2-c$, \color{gray}$1$}] (a3) at ($(a2) + (0, -4)$) {};
        \node[terminal, label=270:$\hat{a}^C_{ 3}$, label={[xshift=-20pt]180:\color{green}$-1$, \color{red}$1$, \color{gray}$-c$}] (a3p) at ($(a3) + (0, 2)$) {};
        \node[vertex, label=0:$b^C_{3}$, label={[xshift=20pt]0:\color{green}$c$, \color{red}$c-2$, \color{gray}$c$}] (b3) at ($(a3) + (2, 0)$) {};
        \node[vertex, label=315:$\hat{b}^C_{3}$, label={[xshift=20pt]0:\color{green}$1$, \color{red}$-1$, \color{gray}$-1$}] (b3p) at ($(b3) + (0, 2)$) {};

        \begin{scope}[xshift=-8cm, yshift = -0 cm]
          \node[terminal, label=270:$a_x$] (xv2) at (-1,1) {};
          \node[terminal, label=90:$\bar{a}_x$] (xv3) at (-1,3) {};
          \node[vertex, label=0:$b_x^t$] (yv1) at (1.,1) {};
          \node[vertex, label=90:$b_x^f$] (yv2) at (1.,3) {};
        \end{scope}

        \draw (xv2) edge node[pos=0.2, fill=white, inner sep=2pt] {\scriptsize $1$}  node[pos=0.76, fill=white, inner sep=2pt] {\scriptsize $1$} (yv1);
        \draw (xv2) edge node[pos=0.2, fill=white, inner sep=2pt] {\scriptsize $2$}  node[pos=0.76, fill=white, inner sep=2pt] {\scriptsize $1$} (yv2);

        \draw (xv3) edge node[pos=0.2, fill=white, inner sep=2pt] {\scriptsize $1$}  node[pos=0.76, fill=white, inner sep=2pt] {\scriptsize $2$} (yv1);
        \draw (xv3) edge node[pos=0.2, fill=white, inner sep=2pt] {\scriptsize $2$}  node[pos=0.76, fill=white, inner sep=2pt] {\scriptsize $2$} (yv2);

        \begin{scope}[xshift=-8cm, yshift = -4 cm]
          \node[terminal, label=270:$a_\new{y}$] (va2) at (-1,1) {};
          \node[terminal, label=90:$\bar{a}_\new{y}$] (vba2) at (-1,3) {};
          \node[vertex, label=0:$b_\new{y}^t$] (bt2) at (1.,1) {};
          \node[vertex, label=90:$b_\new{y}^f$] (bf2) at (1.,3) {};
        \end{scope}

        \draw (va2) edge node[pos=0.2, fill=white, inner sep=2pt] {\scriptsize $1$}  node[pos=0.76, fill=white, inner sep=2pt] {\scriptsize $1$} (bt2);
        \draw (va2) edge node[pos=0.2, fill=white, inner sep=2pt] {\scriptsize $2$}  node[pos=0.76, fill=white, inner sep=2pt] {\scriptsize $1$} (bf2);

        \draw (vba2) edge node[pos=0.2, fill=white, inner sep=2pt] {\scriptsize $1$}  node[pos=0.76, fill=white, inner sep=2pt] {\scriptsize $2$} (bt2);
        \draw (vba2) edge node[pos=0.2, fill=white, inner sep=2pt] {\scriptsize $2$}  node[pos=0.76, fill=white, inner sep=2pt] {\scriptsize $3$} (bf2);

        \begin{scope}[xshift=-8cm, yshift = -8 cm]
          \node[terminal, label=270:$a_\new{z}$] (va3) at (-1,1) {};
          \node[terminal, label=90:$\bar{a}_\new{z}$] (vba3) at (-1,3) {};
          \node[vertex, label=0:$b_\new{z}^t$] (bt3) at (1.,1) {};
          \node[vertex, label=90:$b_\new{z}^f$] (bf3) at (1.,3) {};
        \end{scope}

        \draw (va3) edge node[pos=0.2, fill=white, inner sep=2pt] {\scriptsize $1$}  node[pos=0.76, fill=white, inner sep=2pt] {\scriptsize $1$} (bt3);
        \draw (va3) edge node[pos=0.2, fill=white, inner sep=2pt] {\scriptsize $2$}  node[pos=0.76, fill=white, inner sep=2pt] {\scriptsize $1$} (bf3);

        \draw (vba3) edge node[pos=0.2, fill=white, inner sep=2pt] {\scriptsize $1$}  node[pos=0.76, fill=white, inner sep=2pt] {\scriptsize $2$} (bt3);
        \draw (vba3) edge node[pos=0.2, fill=white, inner sep=2pt] {\scriptsize $2$}  node[pos=0.76, fill=white, inner sep=2pt] {\scriptsize $2$} (bf3);

        \draw (a1p) edge[bend left= 60] node[pos=0.13, fill=white, inner sep=2pt] {\scriptsize $2$}  node[pos=0.9, fill=white, inner sep=2pt] {\scriptsize $1$} (b2);
        
        \draw (a2p) edge[bend right = 70] node[pos=0.15, fill=white, inner sep=2pt] {\scriptsize $2$}  node[pos=0.85, fill=white, inner sep=2pt] {\scriptsize $3$} (b3);
        
        \draw (a3p) edge node[pos=0.2, fill=white, inner sep=2pt] {\scriptsize $2$}  node[pos=0.9, fill=white, inner sep=2pt] {\scriptsize $3$} (b1);

        \draw (a1) edge[ultra thick, gray] node[pos=0.2, fill=white, inner sep=2pt] {\scriptsize $1$}  node[pos=0.76, fill=white, inner sep=2pt] {\scriptsize $1$} (b1);
        \draw (a1) edge[densely dotted] node[pos=0.2, fill=white, inner sep=2pt] {\scriptsize $3$}  node[pos=0.76, fill=white, inner sep=2pt] {\scriptsize $1$} (b1p);

        \draw (a1p) edge[densely dotted] node[pos=0.2, fill=white, inner sep=2pt] {\scriptsize $1$}  node[pos=0.76, fill=white, inner sep=2pt] {\scriptsize $2$} (b1);
        \draw (a1p) edge[ultra thick, gray] node[pos=0.2, fill=white, inner sep=2pt] {\scriptsize $3$}  node[pos=0.76, fill=white, inner sep=2pt] {\scriptsize $2$} (b1p);

        \draw (a2) edge[densely dotted, gray] node[pos=0.2, fill=white, inner sep=2pt] {\scriptsize $1$}  node[pos=0.76, fill=white, inner sep=2pt] {\scriptsize $2$} (b2);
        \draw (a2) edge[ultra thick] node[pos=0.2, fill=white, inner sep=2pt] {\scriptsize $2$}  node[pos=0.76, fill=white, inner sep=2pt] {\scriptsize $1$} (b2p);

        \draw (a2p) edge[ultra thick] node[pos=0.2, fill=white, inner sep=2pt] {\scriptsize $1$}  node[pos=0.85, fill=white, inner sep=2pt] {\scriptsize $3$} (b2);
        \draw (a2p) edge[densely dotted, gray] node[pos=0.2, fill=white, inner sep=2pt] {\scriptsize $3$}  node[pos=0.76, fill=white, inner sep=2pt] {\scriptsize $2$} (b2p);

        \draw (a3) edge[densely dotted, ultra thick] node[pos=0.2, fill=white, inner sep=2pt] {\scriptsize $1$}  node[pos=0.76, fill=white, inner sep=2pt] {\scriptsize $1$} (b3);
        \draw (a3) edge[gray] node[pos=0.2, fill=white, inner sep=2pt] {\scriptsize $3$}  node[pos=0.76, fill=white, inner sep=2pt] {\scriptsize $1$} (b3p);

        \draw (a3p) edge[gray] node[pos=0.2, fill=white, inner sep=2pt] {\scriptsize $1$}  node[pos=0.76, fill=white, inner sep=2pt] {\scriptsize $2$} (b3);
        \draw (a3p) edge[densely dotted, ultra thick] node[pos=0.2, fill=white, inner sep=2pt] {\scriptsize $3$}  node[pos=0.76, fill=white, inner sep=2pt] {\scriptsize $2$} (b3p);

        \draw (a1) edge[] node[pos=0.2, fill=white, inner sep=2pt] {\scriptsize $2$}  node[pos=0.76, fill=white, inner sep=2pt] {\scriptsize $3$} (yv2);
        \draw (a2p) edge[bend right] node[pos=0.2, fill=white, inner sep=2pt] {\scriptsize $4$}  node[pos=0.76, fill=white, inner sep=2pt] {\scriptsize $2$} (bf2);
        \draw (a3) edge[] node[pos=0.2, fill=white, inner sep=2pt] {\scriptsize $2$}  node[pos=0.76, fill=white, inner sep=2pt] {\scriptsize $3$} (bf3);
  \end{tikzpicture}
%  \end{minipage}
		\[
		\begin{array}{rlrl}
\hat{a}^C_{1} & : b^C_{1} \succ b^C_{ 2}  \succ \hat{b}^C_{ 1} \qquad \qquad \qquad & \hat{b}^C_{ 1} &: a^C_{ 1} \succ \hat{a}^C_{ 1}\\
 a^C_{1} & : b^C_{1} \succ \hat{b}^C_{1} \qquad & b^C_{ 1} & : a^C_{ 1} \succ \hat{a}^C_1 \succ \hat{a}^C_{ 3}\\
 \hat{a}^C_2 &: b^C_{ 2} \succ b^C_{ 3} \succ \hat{b}^C_{ 2} \qquad &\hat{b}^C_{ 2} & : a^C_{ 2} \succ \hat{a}^C_{ 2}\\
 a^C_{ 2} & : b^C_{ 2} \succ \hat{b}^C_{ 2} \qquad & b^C_{ 2}  & : \hat{a}^C_{ 1} \succ a^C_{ 2} \succ \hat{a}^C_{ 2} \\
 \hat{a}^C_{ 3} & : b^C_{ 3} \succ b^C_{ 1} \succ \hat{b}^C_{ 3} \qquad & \hat{b}^C_{ 3} & : a^C_{ 3} \succ \hat{a}^C_{ 3}\\ 
 %\vspace{4mm}
 a^C_{3} & : b^C_{ 3} \succ \hat{b}^C_{ 3} \qquad & b^C_{ 3} & : a^C_{ 3} \succ \hat{a}^C_{ 3} \succ \hat{a}^C_{ 2}
 \end{array}
		\]
  \caption{An example of a clause gadget for a clause~$C = x \lor \neg y \lor z$.
  Vertices in~$A$ are squared.
  Matchings~$M_1$, $M_2$, and~$M_3$ are depicted by the dotted, bold, and gray edges, respectively.
  A witness for~$M_1$, $M_2$, and $M_3$ are given by the green, red, and gray numbers, respectively.}
  \label{fig:clause-gadget}
\end{figure}

We have that~$A = \{a_1, a_2, a_3\} \cup \{a_x, \bar a_x : x \in X\} \cup \{a_k^C, \hat a_k^C : k \in [3], C \text{ clause} \}$ and~$B = \{b_1, b_2, b_3\} \cup \{b_x^t, b_x^f : x \in X\} \cup \{b_k^C, \hat b_k^C : k \in [3], C \text{ clause} \}$ and set~$w (a) := c$ for every $a\in A$ and $w(b) = 1$ for every $b\in B$.

\textbf{Connection between variable and clause gadgets.} We now describe the connection between variable and clause gadgets.
\new{We process the clauses in arbitrary order.}
Consider a clause~$C$, and let $y$ be the $i$-th literal in~$C$.
\new{If $y$ is a non-negated variable~$x$,}
%If $y = x$ for some variable $x$, 
then we add edge~$\{a^C_{ i}, b^f_x\}$, where $a^C_{ i}$ \new{inserts} $b^f_x$ \new{directly after $b^C_i$ in} its preferences, and $b^f_x$ \new{appends} $a^C_{ i}$ \emph{at the end of its preferences}.
Otherwise \new{$y$ is the negation~$\neg x$ of} some variable~$x$.
Then we add edge~$\{\hat{a}^C_{ i}, b^f_x\}$, where $\hat{a}^C_{i}$ \new{appends} $b^f_x$ \new{at} the end of its preferences and $b^f_x$ \new{inserts} $\hat{a}^C_{i}$ \new{directly after~$a_x$ in its preferences}.
\new{
As an example, consider a variable~$x$ which appears in clauses~$C_1 = x \lor \dots$, $C_2 = \neg x \lor \dots$, and $C_3 = x' \lor \neg x \lor x''$.
Then the preferences of $b^f_x$ are $a_x \succ \hat{a}^{C_3}_2 \succ \hat a^{C_2}_1 \succ \bar a_x \succ a^{C_1}_1$.
}

Having described our construction, we now turn to proving the correctness of the reduction.

\textbf{Satisfying assignment implies popular matching.} Let $f: X \rightarrow \{\LogicTRUE, \LogicFALSE\}$ be a satisfying assignment.
We construct a popular matching~$M$ as follows.
First, $M$ contains edges $\{a_i, b_i\}$ for $i\in [3]$.
For every variable~$x \in X$ with $f (x) = \LogicTRUE$, we add edges $\{a_x, b_x^t\}$ and $\{\bar a_x , b_x^f\}$ to $M$, while for all other variables, we add edges~$\{a_x, b_x^f\}$ and $\{\bar a_x, b_x^t\}$.
For every clause~$C$, there is at least one literal that is satisfied by $f$.
We fix one such literal~$x_{\new{C}}$ and define $k_C \in [3]$ to be the number such that $x_{\new{C}}$ is the $k_C$-th literal in $C$.
We add~$M_{k_C}^C$ to~$M$, i.e.\ we add edges $\{a^C_{ k_C}, \hat{b}^C_{ k_C}\}$ and $\{\hat{a}^C_{k_C}, b^C_{k_C}\}$ as well as edges $\{a^C_{j}, b^C_{ j}\}$ and $\{\hat  a^C_{ j} ,\hat{b}^C_{ j}\}$ for $j \in [3]\setminus \{k_C\}$ to $M$.

\begin{numberedclaim}\label{lem:backward}
  $M$ is popular.
\end{numberedclaim}

\begin{claimproof}{lem:backward}
  \begin{table}
      \centering
      \begin{tabular}{c | c || c | c | c}
           Vertex & Value of $\witness$ & Vertex & Value of $\witness$ & Condition\\
           \hline
           $a_1$ & $-c$ & $b_1$ & $c$ &\\
           $a_2$ & $-1$ & $b_2$ & $1$\\
           $a_3$ & $c-2$ & $b_3$ & $2-c$\\
           \hline
           \hline
           $a_x$ & $-1$ & $b_x^t$ & $1$ & \multirow{2}*{$f(x) = \LogicTRUE$}\\
           $\bar a_x$ & $c-2$ & $b_x^f$ & $2-c$ & \\
           \hline
           $a_x$ & $1$ & $b_x^t$ & $c$ & \multirow{2}*{$f(x) = \LogicFALSE$}\\
           $\bar a_x$ & $-c$ & $b_x^f$ & $-1$ & \\
           \hline
           \hline
           $\hat a^C_1$ & $-c$ & $\hat b^C_1$ & $-1$&  \multirow{6}*{$k_C = 1$}\\
           $a^C_1$ & $1$ & $b^C_1$ & $c$ &\\
           $\hat a^C_2$ & $-1$ & $\hat b^C_2$ & $1$& \\
           $a^C_2$ & $-c$ & $b^C_2$ & $c$ & \\
           $\hat a^C_3$ & $-1$ & $\hat b^C_3$ & $1$& \\
           $a^C_3$ & $-c$ & $b^C_3$ & $c $ & \\           
           \hline
           $\hat a^C_1$ & $1$ & $\hat b^C_1$ & $-1$& \multirow{6}*{$k_C = 2$}\\
           $a^C_1$ & $2-c$ & $b^C_1$ & $c-2$ & \\
           $\hat a^C_2$ & $-c$ & $\hat b^C_2$ & $-1$& \\
           $a^C_2$ & $1$ & $b^C_2$ & $c$ & \\
           $\hat a^C_3$ & $1$ & $\hat b^C_3$ & $-1$& \\
           $a^C_3$ & $2-c$ & $b^C_3$ & $c -2$ & \\
           \hline
           $\hat a^C_1$ & $1$ & $\hat b^C_1$ & $-1$& \multirow{6}*{$k_C = 3$}\\
           $a^C_1$ & $2-c$ & $b^C_1$ & $c-2$ & \\
           $\hat a^C_2$ & $-1$ & $\hat b^C_2$ & $1$& \\
           $a^C_2$ & $-c$ & $b^C_2$ & $c$ & \\
           $\hat a^C_3$ & $-c$ & $\hat b^C_3$ & $-1$& \\
           $a^C_3$ & $1$ & $b^C_3$ & $c $ &
      \end{tabular}
      \caption{A witness of $M$.}
      \label{fig:witness}
  \end{table}

  We give a witness~$\witness$ of $M$ in \Cref{fig:witness}. It remains to show that this is indeed a feasible witness.
  It is straightforward to verify that there is no conflicting edge inside any of the gadgets.
  For an edge~$e=\{a_x, b_3\}$ between a variable and the \unnamedGadget, we have $\weightedvote^M (e) = 1-c$.
  Since $\wof{a_x} \ge -1$, we have $\wof{a_x} + \wof{b_3} \ge -1 + 2-c = 1-c = \weightedvote^M (e)$, and thus $e$ is not conflicting.
  For an edge~$e$ between a variable and a clause gadget, we make a case distinction.
  
  \textbf{Case 1:} $ e= \{a^C_{ k}, b_x^f\}$ for some clause $C$, $k \in [3]$, and $x \in X$.
  
  In this case, the $k$-th literal of $C$ is~$x$.
  
  \textbf{Case 1.1:} $k = k_C$ \new{ (and thus $x = x_C$)}.
  
  Then $\weightedvote^M (e) = c-1$ and $\wof{a^C_{ k}} = 1$.
  By the definition of $k_C$, we have $f(x) = \LogicTRUE$ and thus $\wof{b_x^f} = 2-c$.
  Thus, $\wof{a^C_{ k}}+ \wof{b_x^f} = 3- c \ge c-1 = \weightedvote^M (e)$ using $c \le 2 $ for the inequality.
  
  \textbf{Case 1.2:} $k \neq k_C$.
  
  Then we have $M(a^C_{ k}) = b^C_{k}$ and thus $\weightedvote^M (e) = -c -1$.
  Since $\wof{a^C_{k}} \ge -c $ and $\wof{b_x^f} \ge -1$, it follows that $\weightedvote^M (e) \le \wof{a^C_{ k}} + \wof{b_x^f}$.
  
  \textbf{Case 2:} $e = \{ \hat a^C_{k}, b_x^f\}$ for some clause~$C$, $k \in [3]$, and $x \in X$.
  
  In this case, the $k$-th literal of~$C$ is $\neg x$.
  
  \textbf{Case 2.1:} $k = k_C$ \new{ (and thus $x = x_C$)}.
  
  By the definition of $k_C$, we have $f (x) = \LogicFALSE$ and therefore $M(b^f_x) = a_x$.
  Thus, we have $\weightedvote^M (e) = -c -1$.
  Since $\wof{\hat a^C_{k}} \ge -c $ and $\wof{b_x^f} \ge -1$, we have $\wof{\hat a^C_{k}} + \wof{b_x^f} \ge \weightedvote^M (e)$.
  
  \textbf{Case 2.2:} $ k \neq k_C$.
  
  Then we have $\wof{\hat a^C_{ k}} \ge -1$.
  If $f(x) = \LogicTRUE$, then $M (b_x^f) = \hat a_x$ and $\wof{b_x^f} = 2-c$.
  Thus, we have $\wof{ \hat a^C_{ k}} + \wof{b_x^f} \ge -1 + 2-c = 1-c = \weightedvote^M (e)$.
  If $f(x) = \LogicFALSE$, then $\weightedvote^M (e) = -c -1 \le \wof{\hat a^C_{ k}} + \wof{b_x^f}$.
\end{claimproof}

%\subparagraph*{Proof of the forward direction.}
\textbf{Popular matching implies satisfying assignment.}
Let $M$ be a popular matching. We start with a technical observation we will frequently rely on in our proof.

\begin{observation}\label{obs:first-choice}
In instances where $w(a) > 0$ for each $a \in A$ and $w(b_1) = w(b_2)$ for each $b_1, b_2 \in B$, each popular matching $M$ covers every vertex~$b\in B$ that is the first choice of some $a\in A$.
\end{observation}
\begin{claimproof}{obs:first-choice}
  If $M$ leaves $b$ unmatched, then $\bigl ( M \setminus \{\{a,M(a)\}\}\bigr) \cup \{ \{a, b\}\}$ is more popular than $M$, a contradiction to the popularity of~$M$.
  \end{claimproof}
  Notice that in our construction, the vertex weights are identical and positive on each side, and thus, each first-choice vertex must be matched in all popular matchings.
\begin{numberedclaim}\label{lem:unnamedGadget}
  Matching~$M$ contains edges $\{a_k, b_k\}$ for every $k\in [3]$.
\end{numberedclaim}

\begin{claimproof}{lem:unnamedGadget}
  \Cref{obs:first-choice} implies that $b_1$ and $b_2$ are matched in~$M$.
  \new{More specifically, $M$ contains $\{a_1, b_1\}$ as otherwise~$\bigl ( M \setminus \{\{a_2 , b_1\}\}\bigr) \cup \{\{a_1, b_1\}\}$ would be more popular than~$M$.
  Further, $M$ contains $\{a_2, b_2\}$ as otherwise~$\bigl ( M \setminus \{\{a_3 , b_2\}\}\bigr) \cup \{\{a_2, b_2\}\}$ would be more popular than~$M$.}
%  Thus, $M$ contains $\{a_1, b_1\}$ and $\{a_2, b_2\}$.
  Assume for a contradiction that $M$ does not contain $\{a_3, b_3\}$. Then $a_3$ is unmatched.
  If $b_3$ was also unmatched, then $M \cup \{\{a_3, b_3\}\}$ would be more popular than $M$, so we must have $\{a_x, b_3\} \in M$ for some variable $x\in X$.
  Then $\bigl(M \setminus \{ \{a_x, b_3\}, \{\bar a_x, b_x^t\}\}\bigr) \cup \{\{a_3, b_3\}, \{a_x, b_x^t\}\}$ (note that $b_x^t$ is either unmatched or matched to $\bar a_x$) is more popular than $M$, a contradiction.
  Thus, $M$ contains $\{a_3, b_3\}$.
  \end{claimproof}

\begin{numberedclaim}\label{lem:vg}
  $M$ contains edges $\{a_x, b_x^t\}$ and $\{\bar a_x, b_x^f\}$ or edges $\{a_x, b_x^f\}$ and $\{\bar a_x, b_x^t\}$ for every variable~${x \in X}$.
\end{numberedclaim}

\begin{claimproof}{lem:vg}
  Assume for a contradiction that the claim does not hold for a variable~$x\in X$.
  Due to the maximality of $M$ and to \Cref{lem:unnamedGadget}, $M$ contains edge~$\{a^C_{k}, b_x^f\}$ or $\{\hat a^C_{k}, b_x^f\}$ for some clause~$C$ and $k \in [3]$.
  \Cref{obs:first-choice} implies that $a_x$ is matched as it is the first choice of $b_x^t$.
  \new{By \Cref{lem:unnamedGadget}, $a_x$ is not matched to~$b_3$ and as $b_x^f$ is matched to $a^C_k$ or $\hat a^C_k$ for some clause~$C$ and $k \in [3]$, it follows that $a_x$ is matched to~$b_x^t$.
  Consequently, $\bar a_x$ is unmatched as it is matched neither to $b_x^t$ nor to $b_x^f$.}
  We will now show that there exists an alternating path starting in~$\bar a_x$ such that augmenting~$M$ with this path results in a matching more popular than~$M$.
  If $M$ contains $\{a^C_{k}, b_x^f\}$, then $\bigl(M \setminus \{\{a^C_{ k}, b_x^f\}\} \bigr) \cup \{\{\bar a_x, b_x^f\}\}$ is more popular than~$M$, contradicting the popularity of~$M$.
  Thus, $M$ contains $\{\hat a^C_{k}, b_x^f\}$.
  If $b^C_{ k}$ or $\hat b^C_{ k}$ was unmatched (call this unmatched vertex~$b$), then $\bigl (M\setminus \{\{\hat a^C_{k}, b_x^f\}\}\bigr) \cup \{\{a^C_{k}, b\}\}$ would be more popular than $M$, contradicting the popularity of~$M$.
  Thus, $M$ contains edges $\{a^C_{ k}, \hat b^C_{ k}\}$ and $\{\hat a^C_{ k-1}, b^C_{ k}\}$ (interpreting $k-1$ modulo 3).
  By \Cref{obs:first-choice}, $b^C_{ k-1}$ is not unmatched. 
  If $\{a^C_{k-1}, b^C_{ k-1}\}\in M$, then 
  \begin{align*}
  \bigl (M \setminus \{\{\hat a^C_{ k}, b_x^f\}, \{\hat a^C_{ k-1}, b^C_{k}\}, \{a^C_{ k-1}, b^C_{k-1}\}\}\bigr)\cup \\ 
  \{\{\bar a_x, b_x^f\}, \{\hat a^C_{ k}, b^C_{ k}\} , \{\hat a^C_{k-1}, b^C_{ k-1}\}, \{a^C_{k-1} , \hat b^C_{k-1}\}\}
  \end{align*}
  is more popular than $M$, contradicting the popularity of~$M$.
  Thus, we have $\{\hat a^C_{ k-2}, b^C_{k-1}\}\in M$ (also interpreting $k-2 $ modulo 3).
  Then one of $b^C_{k-2} $ and $\hat b^C_{k-2}$ is unmatched and we call the unmatched vertex~$b$.
  Then 
  \begin{align*}
  \bigl (M \setminus \{\{\hat a^C_{ k}, b_x^f\}, \{\hat a^C_{k-1}, b^C_{ k}\}, \{\hat a^C_{ k-2}, b^C_{ k-1}\}\}\bigr)\cup \\
  \{\{\bar a_x, b_x^f\}, \{\hat a^C_{ k}, b^C_{ k}\} , \{\hat a^C_{ k-1}, b^C_{k-1}\}, \{\hat a^C_{k-2} , b\}\}
  \end{align*}
  is more popular than $M$, contradicting the popularity of~$M$.
\end{claimproof}

\begin{numberedclaim}\label{lem:ha1b2}
  $M$ does not contain edge $\{\hat a^C_{k}, b^C_{k+1}\}$ for any clause~$C$ and any $k\in [3]$ (where $k+1 $ is taken modulo $3$).
\end{numberedclaim}

\begin{claimproof}{lem:ha1b2}
  Assume for a contradiction that~$M$ contains $\{\hat a^C_{ k}, b^C_{ k+1}\}$ for some $k \in [3]$.
  If $M$ contains $\{\hat a^C_i, b^C_{i+ 1}\}$ for every $i \in [3]$ (considering $i+1$ modulo 3), then $\bigl ( M\setminus \{\{\hat a^C_{ k}, b^C_{ k+1}\},\{\hat a^C_{ k-1}, b^C_{ k}\}, \{ \hat a^C_{ k+1}, b^C_{ k-1}\}\}\bigr) \cup \{\{\hat a^C_{ i}, b^C_{ i}\} : i \in [3]\}$ is more popular than $M$, a contradiction.
  Thus, we may assume \new{without \neu{loss} of generality} that $M$ does not contain $\{\hat a^C_{k-1}, b^C_k\}$.
  By \Cref{obs:first-choice}, vertex $b^C_{k}$ is not unmatched.
  Thus, we have $\{a^C_{ k}, b^C_{ k}\} \in M$ and $\hat b^C_{k}$ is unmatched.
  If $a^C_{ k+1}$ or $\hat a^C_{ k+1}$ was unmatched (call this unmatched vertex~$a$), then $\bigl ( M \setminus \{\{\hat a^C_{k}, b^C_{k+1}\}, \{a^C_{ k}, b^C_{k}\}\}\bigr) \cup \{\{a, b^C_{ k + 1}\}, \{\hat a^C_{k}, b^C_{ k}\}, \{a^C_{k}, \hat b^C_{ k}\}\}$ would be more popular than $M$, contradicting the popularity of~$M$.
  Thus, using \Cref{lem:vg}, $M$ contains edges $\{a^C_{ k+1}, \hat b^C_{ k+1}\}$ and $\{\hat a^C_{ k+1}, b^C_{ k-1}\}$.
  Using \Cref{lem:vg}, one of $a^C_{ k - 1}$ or $\hat a^C_{ k - 1}$ is unmatched (call this unmatched vertex~$a^*$).
  Then 
  \begin{align*}
  \bigl ( M \setminus \{\{\hat a^C_{ k}, b^C_{ k+1}\}, \{a^C_{ k}, b^C_{ k}\}, \{\hat a^C_{ k+1}, b^C_{ k-1}\}\}\bigr) \cup \\
  \{\{a^*, b^C_{ k-1}\}, \allowbreak\{\hat a^C_{ k+1}, b^C_{ k + 1}\}, \{\hat a^C_{ k}, b^C_{ k}\}, \{a^C_{ k}, \hat b^C_{ k}\}\}
  \end{align*}
  is more popular than $M$, a contradiction to the popularity of~$M$.
  \end{claimproof}

  \begin{numberedclaim}\label{lem:cg}
    For every clause~$C$, there exists some $k\in [3]$ such that matching~$M$ contains edges~$\{a^C_{k }, \hat b^C_{ k}\}$ and $\{\hat a^C_{ k}, b^C_{ k}\}$.
  \end{numberedclaim}

  \begin{claimproof}{lem:cg}
    From the fact that every popular matching is maximal and from \Cref{lem:vg,lem:ha1b2} follows that $M$ contains either edges $\{a^C_{ i}, b^C_{ i}\}$ and $\{\hat a^C_{ i}, \hat b^C_{i}\}$ or edges $\{a^C_{ i}, \hat b^C_{ i}\}$ and $\{\hat a^C_{ i}, b^C_{ i}\}$ for every~${i \in [3]}$.
    So assume for a contradiction that $M$ contains $\{a^C_{ i}, b^C_{ i}\}$ and $\{\hat a^C_{ i}, \hat b^C_{ i}\}$ for every~$i \in [3]$.
    Then 
    \begin{align*}
    \bigl ( M \setminus \{\{a^C_{ i}, b^C_{ i}\}, \{\hat a^C_{ i}, \hat b^C_{ i}\} : i \in [3]\} \bigr) \cup \\
    \{\{\hat{a}^C_{1}, b^C_{ 2}\},\allowbreak \{a^C_{ 1}, \hat{b}^C_{ 1}\}, \{\hat{a}^C_{ 2}, b^C_{ 3}\}, \{a^C_{ 2}, \hat{b}^C_{ 2}\},\allowbreak \{\hat{a}^C_{ 3}, b^C_{ 1}\},\allowbreak \{a^C_{ 3}, \hat{b}^C_{ 3}\}\}
    \end{align*}
    is more popular than $M$, a contradiction to the popularity of~$M$.
  \end{claimproof}

  \begin{numberedclaim}\label{lem:vg-cg}
    Assume that $M$ contains edges $\{a^C_{ k}, \hat b^C_{ k}\}$ and $\{\hat a^C_{k}, b^C_{ k}\}$ for some clause $C $ and some $k\in [3]$.
    If the $k$-th literal of $C$ is $x$ for a variable $x\in X$, then $M$ contains~$\{a_x, b_x^t\}$ and~$\{\bar a_x, b_x^f\}$.
    If the $k$-th literal of $C$ is $\neg x$ for a variable $x\in X$, then $M$ contains~$\{a_x, b_x^f\}$ and~$\{\bar a_x, b_x^t\}$.
  \end{numberedclaim}

  \begin{claimproof}{lem:vg-cg}
    First assume that the $k$-th literal of $C$ is $x$ but $M$ does not contain edges $\{a_x, b_x^t\}$ and $\{\bar a_x, b_x^f\}$.
    By \Cref{lem:vg}, $M$ contains edges $\{a_x, b_x^f\}$ and $\{\bar a_x, b_x^t\}$.
    Then $\bigl ( M \setminus \{\{\bar a_x, b_x^t\}, \{a_x, b_x^f\}, \allowbreak \{a^C_{ k}, \hat b^C_{ k}\}\}\bigr) \cup \{\{a_x, b_x^t\}, \{a^C_{ k}, b_x^f\}\}$ is more popular than $M$, a contradiction to the popularity of~$M$.
    
    Now assume that the $k$-th literal of $C$ is $\neg x$ but $M$ does not contain edges $\{a_x, b_x^f\}$ and $\{\bar a_x, b_x^t\}$.
    By \Cref{lem:vg}, $M$ contains edges $\{a_x, b_x^t\}$ and $\{\bar a_x, b_x^f\}$.
    Then $\bigl ( M \setminus \{\{a_1, b_1\}, \{a_2, b_2\},\allowbreak \{a_3, b_3\}, \{a_x, b_x^t\}, \{\bar a_x, b_x^f\}, \{a^C_{ k}, \hat b^C_{ k}\}, \{\bar a^C_{ k}, b^C_{ k}\}\bigr) \cup\allowbreak \{\{a_2, b_1\},\allowbreak \{a_3, b_2\},\allowbreak \{a_x, b_3\}, \{\bar a_x, b_x^t\}, \{\hat a^C_{ k}, b_x^f\}, \{a^C_{ k}, b^C_{ k}\}\}$ is more popular than $M$, a contradiction to the popularity of~$M$.
    \end{claimproof}

  \begin{numberedclaim}\label{lem:forward}
    Any popular matching implies a satisfying assignment.
  \end{numberedclaim}

  \begin{claimproof}{lem:forward}
  Let $M$ be a popular matching.
   By \Cref{lem:vg}, for each variable~$x$, matching~$M$ contains either $\{a_x, b_x^t\}$ and $\{\bar{a}_x, b_x^f\}$, or the edges $\{\bar{a}_x, b_x^t\}$ and $\{a_x, b_x^f\}$.
   In the former case, we set the variable to $\LogicTRUE$, while we set the variable to $\LogicFALSE$ in the latter case.
   It remains to show that every clause is satisfied by one variable.
   By \Cref{lem:cg}, for every clause~$C$, there is some~$k \in [3]$ such that $\{a^C_{ k}, \new{\hat b^C_{ k}}\}$ and $\{\neu{\hat a^C_{ k}, b^C_{k}}\}$ are in~$M$.
   By \Cref{lem:vg-cg}, it follows that the $k$-th \new{literal} of $C$ is \new{satisfied}.
   %, while all other variables appearing in this clause are set to $\LogicFALSE$.
  \end{claimproof}
With this we have finished the proof of \Cref{thm:np-hardb}.
  \end{proof}
  
Note that except for Case 1.1 in the proof of \Cref{lem:backward}, the whole proof works for all~$\new{1 < } c \le 3$.
The condition $c \le 3$ is used in the \unnamedGadget: if $c > 3$, then the \unnamedGadget\ does not admit a popular matching and thus the whole constructed instance does not have a popular matching (the given witness~$\witness$ would not be feasible because $\wof{b_3} = 2-c < -1$ for $c > 3$).
  
\begin{remark}
    Reducing from a restricted version of \textsc{3-Sat} where every variable appears at most three times and making one copy the \unnamedGadget\ for every variable gadget, the hardness also extends to graphs of maximum degree five.
\end{remark}

%%%%%%%%%%%%%

\subsection{Maximum-Utility popular matchings}
\label{sec:np:weighted}

In \Cref{sec:alg}, we will show that \textsc{Popular Matching with Weighted Voters} can be solved in linear time if all agents from~$A$ have weight~$c$ for some $c > 3$ while all agents from~$B $ have weight 1.
%For such weight functions, we have a strong inapproximability result for finding an optimal popular matching.
In this problem, the input additionally contains a function~$\omega: E \rightarrow \mathbb{Q}_{\geq 0}$ on the edges, and the goal is to find a popular matching maximizing $\sum_{e \in M}\omega(e)$ among all popular matchings.
We show that the resulting problem is hard even to approximate by a reduction from \textsc{Independent Set}. Note that this is in sharp contrast to the 2-approximation for the unit-vertex-weight case~\cite{FaenzaKPZ19} and also with the optimal house allocation case with weighted voters, for which the results of Mestre~\cite{DBLP:journals/talg/Mestre14} and McDermid and Irving~\cite{MI11} imply that an optimal popular matching can be computed in polynomial time.
%Mestre's reduction~\cite{DBLP:journals/talg/Mestre14} to a house allocation instance without vertex weights and the polynomial-time algorithm~\cite{MI11} for finding an optimal popular matching in such instances imply that an optimal popular matching can be computed in polynomial time.

\begin{theorem}\label{thm:np-hardc}
    Even if  $w (a) = c > 3$ for all~$a\in A$, $w (b) = 1 $ for all $b\in B$, and every edge~$e\in E$ has weight~0 or~1, approximating a maximum edge-weight popular matching by a factor of~$O(n^{1 - \varepsilon})$ is \NP-hard for any~$\epsilon > 0$.
\end{theorem}

%%%%%%%%%%%%%%%%

\begin{proof}
    Fix a constant $c > 3$.
    We reduce from \textsc{Independent Set}. The input $(G, \ell)$ consists of a graph $G$ and an integer~$\ell$, while the goal is to decide whether $G$ has an independent vertex set of size~$\ell$. Zuckerman~\cite{Zuc07} showed that this problem cannot be approximated within $O(|V(G)|^{1-\varepsilon})$ for any $\varepsilon  > 0$, unless \P = \NP.
  To simplify notation in our reduction, we fix an arbitrary orientation for each edge in~$G$.
  
  \textbf{Construction.}  For each $v\in V(G)$, we add a 4-cycle containing vertices $a_v, \hat a_v, b_v$, and $\hat b_v$---not unlike in the variable gadget in \Cref{thm:np-hardb}.
  For each edge~$e =\{u, z\}\in E(G)$, we fix an orientation~$(u, z)$ and add the edge~$\{\hat a_z, \hat b_u\}$.
  The preferences are depicted in \Cref{fi:indep}.
  We have~$A = \{a_v, \hat a_v : v\in V( G) \} $ and $B = \{b_v, \hat b_v : v\in V(G)\}$.
  We set~$w (a) = c$ for every~$ a \in A$ and $w(b) = 1$ for every $b \in B$.
  For every~$v \in V(G)$, we set $\omega (\{ a_v, \hat b_v\}) := 1$, while for all other edges~$e$ we set $\omega (e) = 0$.
\begin{figure}
      \centering  
         \begin{minipage}{0.45\textwidth}
         \centering 
		\[
		\begin{array}{rl}
    a_v & : b_v \succ \hat b_v\\ 
    \hat a_v & : b_v  \succ \hat b_v \succ [\hat b_z ]_{(z, v)\in \delta^- (v)} \\
    b_v & : a_v \succ \hat a_v\\
    \hat b_v & : [\hat a_z]_{(v, z)\in \delta^+(v)} \succ a_v \succ \hat a_v\\
		\end{array}
		\]
		\vspace{20mm}
	\begin{tikzpicture}
	\node[vertex, label=below:$u$] (v1) at (0, 0) {};
    \node[vertex, label=below:$v$] (v2) at (2, 0) {};
    \node[vertex, label=above:$z$] (v3) at (1, 1.7) {};

    \draw [->, very thick, color=blue] (v2) -- (v1);
    \draw [->, very thick, color=blue] (v3) -- (v2);
    \draw [->, very thick, color=blue] (v1) -- (v3);

\draw[line width=1mm,-implies,double, double distance=1mm] (3,1) -- (4,1);
  \end{tikzpicture}
  \vspace{20mm}
	\end{minipage}\begin{minipage}{0.45\textwidth}
	\centering
  \begin{tikzpicture}
          \node[terminal, label=180:$\hat{a}_u$] (a1) at (-1,1) {};
          \node[terminal, label=180:$a_u$] (a1p) at (-1,3) {};
          \node[vertex, label={[xshift=-.2cm]275:$\hat{b}_u$}] (b1) at (1.,1) {};
          \node[vertex, label=0:$b_u$] (b1p) at (1.,3) {};

        \node[terminal, label=180:$\hat{a}_v$] (a2) at ($(a1) + (0, -4)$) {};
        \node[terminal, label=180:$a_v$] (a2p) at ($(a2) + (0, 2)$) {};
        \node[vertex, label=0:$\hat{b}_v$] (b2) at ($(a2) + (2, 0)$) {};
        \node[vertex, label={[xshift=-.1cm]0:$b_v$}] (b2p) at ($(b2) + (0, 2)$) {};

        \node[terminal, label=180:$\hat{a}_z$] (a3) at ($(a2) + (0, -4)$) {};
        \node[terminal, label=180:$a_z$] (a3p) at ($(a3) + (0, 2)$) {};
        \node[vertex, label=0:$\hat{b}_z$] (b3) at ($(a3) + (2, 0)$) {};
        \node[vertex, label=90:$b_z$] (b3p) at ($(b3) + (0, 2)$) {};

        \draw (a1) edge[bend left= 0] node[pos=0.13, fill=white, inner sep=2pt] {\scriptsize $3$} 
        node[pos=0.87, fill=white, inner sep=2pt] {\scriptsize $1$} (b2);
        
        \draw (a2) edge[bend right = 0] node[pos=0.13, fill=white, inner sep=2pt] {\scriptsize $3$}
        node[pos=0.87, fill=white, inner sep=2pt] {\scriptsize $1$} (b3);
        
        \draw (a3) edge[out=20,in=0,looseness=1] node[pos=0.08, fill=white, inner sep=2pt]       {\scriptsize $3$} node[pos=0.9, fill=white, inner sep=2pt] {\scriptsize $1$} (b1);

        \draw (a1) edge[ultra thick] node[pos=0.2, fill=white, inner sep=2pt] {\scriptsize $2$}  node[pos=0.76, fill=white, inner sep=2pt] {\scriptsize $3$} (b1);
        \draw (a1) edge[] node[pos=0.2, fill=white, inner sep=2pt] {\scriptsize $1$}  node[pos=0.76, fill=white, inner sep=2pt] {\scriptsize $2$} (b1p);

        \draw (a1p) edge[] node[pos=0.2, fill=white, inner sep=2pt] {\scriptsize $2$}  node[pos=0.76, fill=white, inner sep=2pt] {\scriptsize $2$} (b1);
        \draw (a1p) edge[ultra thick] node[pos=0.2, fill=white, inner sep=2pt] {\scriptsize $1$}  node[pos=0.76, fill=white, inner sep=2pt] {\scriptsize $1$} (b1p);

        \draw (a2) edge[] node[pos=0.2, fill=white, inner sep=2pt] {\scriptsize $2$}  node[pos=0.76, fill=white, inner sep=2pt] {\scriptsize $3$} (b2);
        \draw (a2) edge[ultra thick] node[pos=0.2, fill=white, inner sep=2pt] {\scriptsize $1$}  node[pos=0.76, fill=white, inner sep=2pt] {\scriptsize $2$} (b2p);

        \draw (a2p) edge[ultra thick] node[pos=0.2, fill=white, inner sep=2pt] {\scriptsize $2$}  node[pos=0.85, fill=white, inner sep=2pt] {\scriptsize $2$} (b2);
        \draw (a2p) edge[] node[pos=0.2, fill=white, inner sep=2pt] {\scriptsize $1$}  node[pos=0.76, fill=white, inner sep=2pt] {\scriptsize $1$} (b2p);

        \draw (a3) edge[ultra thick] node[pos=0.2, fill=white, inner sep=2pt] {\scriptsize $2$}  node[pos=0.76, fill=white, inner sep=2pt] {\scriptsize $3$} (b3);
        \draw (a3) edge[] node[pos=0.2, fill=white, inner sep=2pt] {\scriptsize $1$}  node[pos=0.76, fill=white, inner sep=2pt] {\scriptsize $2$} (b3p);

        \draw (a3p) edge[] node[pos=0.2, fill=white, inner sep=2pt] {\scriptsize $2$}  node[pos=0.76, fill=white, inner sep=2pt] {\scriptsize $2$} (b3);
        \draw (a3p) edge[ultra thick] node[pos=0.2, fill=white, inner sep=2pt] {\scriptsize $1$}  node[pos=0.76, fill=white, inner sep=2pt] {\scriptsize $1$} (b3p);
  \end{tikzpicture}
  \end{minipage}
  \caption{An example instance constructed from the triangle on the left in the proof of \Cref{thm:np-hardc}.
  Squared brackets denote an arbitrary ordering of the corresponding vertices.
  The set of edges leaving $v$ is denoted by $\delta^-(v)$, while the set of edges entering $v$ is denoted by $\delta^+(v)$. Thick edges mark the popular matching corresponding to the independent vertex set $\{v\}$.}
  \label{fi:indep}
  \end{figure}
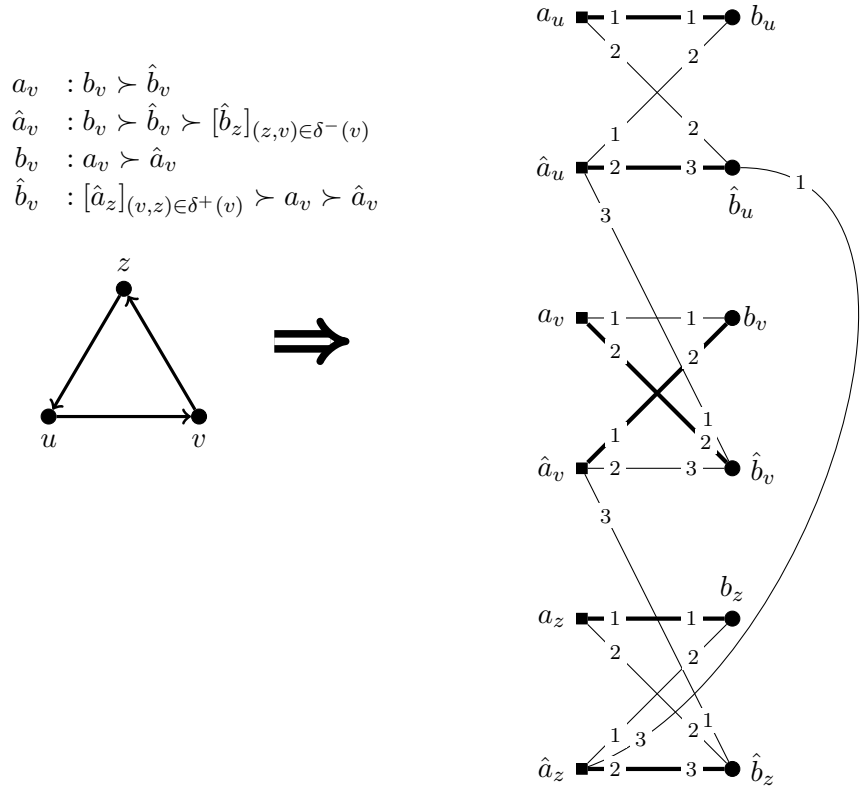
  
  \textbf{Correctness.} We now prove the correctness of the above reduction.
  \begin{numberedclaim}\label{lem:is->pm}
    If $G$ has an independent set of size $\ell$, then there exists a popular matching~$M$ with $\omega (M) = \sum_{e\in M} \omega (e) \ge \ell$.
  \end{numberedclaim}

  \begin{claimproof}{lem:is->pm}
    Let $X$ be an independent set of size $\ell$.
    Consider matching $M:= \{\{a_v, b_v\}, \{\hat a_v, \hat b_v\} : v\in V(G) \setminus X\} \cup \{\{a_v, \hat b_v\}, \{\hat a_v, b_v\}: v\in X\}$.
    Clearly, $\omega (M) = \ell$, so it remains to show that $M$ is popular.
    We define a witness~$\witness$ of $M$ in \Cref{fig:witness-weighted}.
      \begin{table}[htb]
      \centering
      \begin{tabular}{c | c || c | c | c}
           Vertex & Value of $\witness$ & Vertex & Value of $\witness$ & Condition\\
           \hline\
           $a_v$ & $-c$ & $b_v$ & $c$ & \multirow{2}*{$v \in V(G) \setminus X$}\\
           $\hat a_v$ & $-1$ & $\hat b_v$ & $1$ & \\
           \hline
           $a_v$ & $1$ & $b_v$ & $c$ & \multirow{2}*{$v\in X$}\\
           $\hat a_v$ & $-c$ & $\hat b_v$ & $-1$ & \\
      \end{tabular}
      \caption{A witness of $M$.}
      \label{fig:witness-weighted}
  \end{table}
  
  Consider an edge $e = \{\hat a_{v}, \hat b_z\}$ with~$v \neq z$.
  Then $\weightedvote^M (e) = 1-c$.
  If $v \notin X$, then $\wof{\hat a_{v}} + \wof{\hat b_v} \ge -1 -  1 = -2 \ge \weightedvote^M (e)$ as $c > 3$.
  Otherwise we have $z\notin X$ (as $X$ is an independent set) and consequently, $\wof{a_{v}} + \wof{\hat b_z} = -c + 1 = \weightedvote^M(e)$.
  Thus, $e$ is not conflicting.
  It is straightforward to verify that for every $v\in V(G)$, none of the edges $\{a_v, b_v\}$, $\{a_v, \hat b_v\}$, $\{\hat a_v, b_v\}$, and $\{\hat a_v, \hat b_v\}$ is conflicting.
  Thus, $M$ is a popular matching.
  \end{claimproof}

  \begin{numberedclaim}\label{lem:no-ae-hbv}
    No popular matching contains an edge $\{\hat a_v, \hat b_z\}$ for $z \neq v$.
  \end{numberedclaim}

  \begin{claimproof}{lem:no-ae-hbv}
    For every $v \in V(G)$, we have that $f(\hat a_v) = b_v$ and $s(\hat a_v) = \hat b_v$.
    \Cref{lem:popular-with-0,th:abraham} imply that for any vertex~$v \in V$, any popular matching contains~$ \{\hat a_v, b_v\}$ or $\{\hat a_v, \hat b_v\}$.
    Consequently, no popular matching can contain edge~$\{\hat a_v, \hat b_z\}$ for $z\neq v$.
    \end{claimproof}

  \begin{numberedclaim}\label{lem:pm->is}
    If there exists a popular matching~$M$ with $\omega(M) = \sum_{e \in M} \omega(e) \geq \ell$, then $G$ admits an independent set of size $\ell$.
  \end{numberedclaim}

  \begin{claimproof}{lem:pm->is}
    Let $X:= \{v \in V(G): \{a_v, \hat b_v\}\in M\}$.
    Clearly $|X| \ge \ell$, so it is enough to show that $X$ is an independent set.
    Assume for a contradiction that there exists an edge $e = (v, z) \in E(G)$ with $v \in X$ and $z\in X$.
    By \Cref{lem:no-ae-hbv} and the fact that every popular matching is maximal, matching~$M$ contains edges $\{\hat a_v, b_v\}$ and $\{\hat a_z, b_z\}$.
    Then 
    $$M' := \bigl ( M\setminus \{ \{\hat a_v, b_v\}, \{a_v, \hat b_v\}, \{a_z, \hat b_z\}, \{\hat a_z, b_z\}\}\bigr) \cup \{ \{a_v, b_v\}, \{\hat a_{z}, \hat b_v\}, \{a_z, b_{z}\}\}$$
    is more popular than $M$ (because $a_v$, $a_z$, $b_v$, $b_z$, and $\hat b_v$ prefer $M'$ while only $\hat a_v$, $\hat a_z$, and $\hat b_z$ prefer~$M$), a contradiction to the popularity of~$M$.
    \end{claimproof}

    The correctness of the reduction follows from \Cref{lem:pm->is,lem:is->pm} and the corresponding hardness result for \textsc{Independent Set}~\cite{Zuc07}. This finishes the proof of \Cref{thm:np-hardc}.
  \end{proof}

%%%%%%%%%%%%%%%%

  \section{Algorithm for different weights for~$A$ and~$B$}
  \label{sec:alg}
  
  Following up on \Cref{thm:np-hardb,thm:np-hardc}, we now consider the case that all vertices from one side of the bipartition have a weight $c>3$ while all vertices from the other side of the bipartition have weight~1.
  The main result of this section is that this restricted variant of \textsc{Popular Matching with Weighted Voters} is solvable in polynomial time if $c>3$. Before we describe the algorithm, we show some structural properties for the witnesses of popular matchings.
  
  \subsection{Structural properties of witnesses}
  \label{sec:structural-properties}
  
    We now restrict the value domain of possible witnesses. This technique has been used by Kavitha~\cite{Kav19} upon defining truly popular matchings for non-bipartite instances.
%    \accom{Truly popular matchings have a witness consisting of -1,0,1, and all odd sets are 0 in it. I don't think we want to go into these details.}
    
    Given an edge $ e= \{a, b\} \in E$, we call a witness~$\witness$ of a popular matching~$M$ \emph{tight on $e$} if and only if $\wof{a} + \wof{b} = \weightedvote^{M} (e)$.
    It follows from complementary slackness that every witness is tight on every edge contained in some popular matching.
    
    \begin{lemma}\label{lem:tight}
      Let $e$ be an edge which is contained in a popular matching~$M^*$.
      For any popular matching~$M$ with witness~$\witness$, we have that $\witness $ is tight on~$e$.
%     Assume that there exist two popular matchings $M_1$ and $M_2$.
%     Let~$C$ be a connected component of $M_1 \triangle M_2$.
%     Then every witness $\witness^1$ of $M_1$ and $\witness^2$ of~$M_2$ is tight on every edge of $C$.
    \end{lemma}
    
    %%%%%%%%%%%
    
        \begin{proof}
     Let $P$ and $D$ be the primal LP~(\ref{lp:primal}) and dual LP~(\ref{lp:dual}) from \cref{lem:witness} for $M$.
     Because both~$M$ and $M^*$ are popular, it holds that $\Delta_w (M, M^*) = 0$.
     Thus, we have that $M^*$ is an optimal solution for~$P$.
     The statement now follows from complementary slackness (complementary slackness states that for every optimal solutions~$\bm{x}$ to a primal LP, in any optimal solution to the dual LP the inequality corresponding to any non-zero variable in~$\bm{x}$ must be tight, i.e.\ fulfilled with equality; see e.g.~\cite[Equation~$5.13$]{Schrijver}).
    \end{proof}
    
    %%%%%%%%%%%%%

\new{
    We observe the following direct consequences of \Cref{lem:tight}:
    \begin{observation}
        \label{obs:witness-bounds}
        Let $\witness$ be a witness of a popular matching~$M$.
        Then $\wof{a} = - \wof{b}$ for every~$\{a, b\} \in M$.
        Further, $\wof{a} \le 1$ for every $ a\in A$ and $\wof{b} \le c$ for every~$b\in B$.
        For each unmatched agent~$v$, we have $\wof{v} = 0$.
    \end{observation}

    \begin{proof}
        For any edge~$\{a, b\} \in M$, applying \Cref{lem:tight} to the inequality $\wof{a} + \wof{b} \ge \weightedvote^M (\{a, b\}) = 0$ implies $\wof{a} = - \wof{b}$.
        Using this, the inequalities $\wof{a} \ge -c$ and $\wof{b} \ge -1$ now imply $\wof{b} \le c$ and $\wof{a} \le 1$.
        It remains to consider unmatched agents.
        Note that $0 \ge \sum_{v \in V} \wof{v} = \sum_{\{a, b\} \in M} \bigl(\wof{a} + \wof{b}\bigr) + \sum_{v \in V : M(v) = \bot} \wof{v} = \sum_{ v\in V: M(v) = \bot} \wof{v}$.
        As $\wof{v} \ge 0$ for every unmatched~$v \in V$ (as $\weightedvote^M (\{v, v\}) = 0$), it follows that $\wof{v} = 0$ for every unmatched~$ v\in V$.
    \end{proof}
}
  
    Next, we prove that for each popular matching, there exists a witness taking one of only six possible values at each vertex.
    More precisely, we can find a witness~$\witness$ with $\witness_a \in R^a := \{-c, \allowbreak {1-c},\allowbreak {2-c}, -1, 0, 1\}$ for all $a\in A$ and $\witness_b \in R^b := \{-1, 0, 1, c-2, c-1, c\}$ for all $b\in B$. 
    We call such a witness \emph{nice}.
    The idea behind the proof of the existence of such a nice witness is that, given a witness~$\witness$, we can either increase $\wof{a}$ by some $\varepsilon > 0$ for all $a\in A$ with $\wof{a}$ being not nice (i.e.\ $\wof{a} \notin R^a$) and decrease $\wof{b}$ by some $\varepsilon>0$ for all $b\in B$ with $\wof{b}\notin R^b$, or decrease~$\wof{a}$ and increase~$\wof{b}$, which results in another feasible witness that takes ``nice'' values at more vertices.
  
    \begin{lemma}\label{lem:nice}
     Let $M$ be a popular matching.
     Then there exists a nice witness of~$M$.
    \end{lemma}
    
    %%%%%%%%%%%%%%
    
     \begin{proof}
     For a vertex $v$, let $R^v$ be the set of possible values of a nice witness at $v$, i.e.\ $R^v = \{-c, \allowbreak {1- c},\allowbreak 2-c, -1, 0, 1\}$ if $v\in A$ and $R^v = \{-1, 0, 1, c-2, c-1, c\}$ if $v\in B$.
     Let $\witness$ be a witness  of $M$, and let $S(\witness)$ be the set of vertices $v$ such that $\witness_v \notin R^v$.
     We assume that $\witness$ is a witness minimizing the cardinality of $S(\witness)$.
     Assume for a contradiction that $S(\witness) \neq \emptyset$.
     Let $v\in S(\witness)$ and $r\in R^v$ such that $|y_v - r|$ is minimized.
     Set $x:= r - \witness_v$.
     
     First assume that $v\in A$.
     We construct a new witness $\witness'$ as follows:
     For each vertex~$a\in A\cap S(\witness)$, we set $\witness'_a := \witness_a + x$, and for each vertex~$b\in B\cap S(\witness)$, we set $\witness'_b := \witness_b - x$.
     Since $-c \in R^a$ for every~$a\in A$ and $-1 \in R^b$ for every~$b\in B$, the definition of~$x$ implies that $\wof{v}' \ge \neu{-}w(v)$ for every~$v\in V$.
     \new{By \Cref{obs:witness-bounds},}
%     Note that for each unmatched vertex~$u$, 
     we have~$\wof{u} = 0\in R^u$ \new{for each unmatched vertex~$u$}.
     \new{Again by \Cref{obs:witness-bounds},}
%     Thus, 
     a \new{matched} vertex $u$ fulfills $\wof{u} \in R^u$ if and only if $\witness_{M(u)} \in R^{ M(u)}$, and therefore, $|A \cap S(\witness) | = |B \cap S(\witness)|$.
     Thus, we have $\sum_{v\in A\cup B} \wof{v}' = \sum_{v\in A \cup B} \wof{v}+ x \cdot |A \cap S(\witness)| - x \cdot |B \cap S(\witness)| = 0$.
     For any edge~$\{a, b\}$ with either $a, b\in S(\witness)$ or $a, b\notin S(\witness)$, we have $\wof{a}' + \wof{b}' = \wof{a}+\wof{b}$, and thus, $\witness'$ also satisfies $\wof{a}' + \wof{b}' \ge \weightedvote^M (\{a, b\})$.
     So assume that only one end vertex of edge~$\{a, b\}$ is in~$S(\witness)$.
     This implies~$\{a, b\} \notin M$. % by \cref{lem:tight}.
     Thus, we have $\weightedvote^M (\{a, b\}) \in \{-c-1, 1-c, c-1, 1 + c\}$.
     If \neu{$a \notin S( \witness)$, i.e.\ } $\witness_a \in R^a$, then (since $\weightedvote^M (\{a, b\}) - \witness_a \le \witness_b \le c$ \new{using \Cref{obs:witness-bounds} for the last inequality}) it follows that $\weightedvote^M (\{a, b\}) - \witness_a$ \new{is either at most \neu{$-1$} or contained in $\{0, 1, c-2, c- 1, c\}$.
     By the definition of~$x$ and because $b \in S(\witness)$, for any $r \in R^b$ with $\witness_b \ge r$, we have $\wof{b} \ge r + |x|$.
     Since $-1, 0, 1, c-2, c-1$, and $c$ are all contained in $R^b$, it follows from $\witness_b \ge \weightedvote^M (\{a, b\}) - \wof{a}$ that also $\witness_b \ge \weightedvote^M (\{a, b\}) - \wof{a} + |x|$.
     Since $|\witness_a + \witness_b - (\witness'_a + \witness'_b) | = |x|$, we have $\witness'_a + \witness'_b \ge \weightedvote^M (\{a, b\})$.}
     If \neu{$b \notin S (\witness)$, i.e.\ } $\witness_b \in R^b$, then (since $\weightedvote^M (\{a, b\}) - \witness_b \le \witness_a \le 1$ \new{using \Cref{obs:witness-bounds} for the last inequality}) it follows that \new{$\weightedvote^M (\{a, b\}) - \witness_b$ is at most $-c$ or contained in $\{1-c, 2-c, -1, 0, 1\}$.
     Since $-c, 1-c, 2-c, -1, 0$, and $1$ are all contained in $R^a$, it follows from $\witness_a \ge \weightedvote^M (\{a, b\}) - \wof{b}$ that also $\witness_a \ge \weightedvote^M (\{a, b\}) - \wof{b} + |x|$.
     Since $|\witness_a + \witness_b - (\witness'_a + \witness'_b) | = |x|$, we have $\witness'_a + \witness'_b \ge \weightedvote^M (\{a, b\})$.
     We have shown that} $\witness'$ is a witness  of~$M$.
     Clearly, $S (\witness') \subseteq S(\witness)$, and since $v\notin S(\witness')$, it follows that $|S(\witness') |< |S(\witness)|$, a contradiction.
     
     The case $v\in B$ is symmetric.
    \end{proof}
    
    %%%%%%%%%%

    We further classify nice witnesses depending on which values they attain to be \emph{\even} or \emph{\odd}.

    \begin{definition}
     A witness containing only the values $-c, 2-c, -1, 1, c -2$, and $c$ is \emph{\odd}.
     A witness containing only the values $1-c, 0$, and $c-1$ is \emph{\even}.
    \end{definition}

    \begin{lemma}\label{lem:even-odd}
     Let $M_1$ and $M_2$ be popular matchings and $C$ a connected component of~$M_1 \triangle M_2$.
     Then every nice witness~$\witness$ of $M_1$ on $C$ is either \even\ or \odd.
    \end{lemma}
    
    %%%%%%%%%%%
    
\begin{proof}
     Since the witness is nice, it can only take values 
     $$-c, 1-c, 2-c, -1 , 0, 1, c-2, c-1, c.$$
     For any edge $\{a, b\} \in E(C)$ with $\{a, b\} \notin M_{i}$, we have that $\weightedvote^{M_i} (\{a, b\}) \in \{-c-1, 1-c, c-1, 1+ c\}$.
     Fix a vertex $v\in V(C)$.
     If $\witness_v  \in \{-c, 2-c, -1, 1, c-2, c\}$, then \Cref{lem:tight} implies that $\witness_{u} \in \{-c, 2-c, -1, 1, c-2, c\}$ for all neighbors $u \in N_C (v)$.
     By iterating this argument, it follows that $\witness $ is \odd.
     Otherwise we have $\witness_v \in \{1-c, 0, c-1\}$.
     Then \Cref{lem:tight} implies that $\witness_u \in \{1-c, 0, c-1\}$ for all neighbors $u$ of $v$, implying that $\witness$ is \even.
    \end{proof}
    
    %%%%%%%%%%%%%%%%

  \subsection{The algorithm}
  Having shown structural properties of witnesses, we can now describe the algorithm.
%  Note that while the structural properties of witnesses in \Cref{sec:structural-properties} hold for every~$c \ge 1$, the results in this section only hold for~$c> 3$

  %The general approach is as follows:
  Our algorithm consist of the following three stages. Each of these stages is illustrated on an example in \ref{sec:example}.
  \begin{enumerate}
      \item \textbf{Pruning non-popular edges}\\
      We reduce the set of edges that may appear in a popular matching to two incident edges per vertex. This allows us to decompose the graph containing the edges that may appear in a popular matching into a disjoint union of paths and cycles.
      \item \textbf{Computing ``local'' witnesses}\\
      For each path or cycle component, we show that we can restrict ourselves to up to four possible ``local'' witnesses of popularity. 
      \item \textbf{Constructing a ``global'' witness}\\
      Our algorithm greedily recognizes local witnesses that cannot be part of a global witness. At termination, we have either constructed a global witness and thus found a popular matching, or dismissed all local witnesses for one component. In the latter case, we conclude that no popular matching exists.
  \end{enumerate}

  %The general approach is as follows:  We reduce the set of edges that may appear in a popular matching to two incident edges per vertex. Thus, we can compose the graph containing the edges that may appear in a popular matching into a disjoint union of paths and cycles.  For each such path or cycle, we can show that we can restrict ourselves to up to four possible witnesses of popularity.  Our algorithm then greedily recognizes ``local'' witnesses that cannot be part of a ``global witness'' until we either found a popular matching, or dismissed all ``local'' witnesses for one path or cycle.  In the latter case, we conclude that no popular matching exists. Each step of our algorithm is illustrated on an example at the end of the appendix.
  
  \subsubsection{Pruning non-popular edges}
  \label{sec:non-pop}
  
  We start by computing a subgraph of degree at most two at each vertex that contains every popular edge.
  First, we show that for every vertex~$a\in A$, there are at most two other vertices to which $a$ can be matched in a popular matching, namely $f(a)$ and~$s(a)$.
  To do so, given an instance~$\mathcal{I}$, we show that each popular matching is also popular in the instance~$\mathcal{I}'$ where the weight of each vertex in~$B$ is reduced to~0, while the rest of $\mathcal{I}'$ equals~$\mathcal{I}$. Since $w(a) = c$ for each $a \in A$, this modified instance is equivalent to a house allocation instance, for which there is a characterization of popular matchings that relies on $f(a)$ and $s(a)$ vertices, as presented in \Cref{th:abraham}.
  We remark that this key lemma does not hold anymore when vertices from~$A$ may have different weights or weight smaller or equal to three.
  
  \begin{lemma}\label{lem:popular-with-0}
   Let $\mathcal{I}$ be an instance of \textsc{Popular Matching with Weight\-ed Voters} such that $w (a) = c > 3$ for all $a\in A$ and $w (b) = 1 $ for all $b\in B$.
   Any popular matching~$M$ in the instance~$\mathcal{I}$ is also popular in the instance $\mathcal{I}'$ that arises from $\mathcal{I}$ by setting $w (b) = 0$ for all~${b\in B}$.
  \end{lemma}
  
  \begin{proof}
   Let $M$ be a popular matching in $\mathcal{I}$, and assume for a contradiction that there is a matching that is more popular than $M$ in $\mathcal{I}'$.
   Then, by \Cref{th:abraham}, there exists a vertex~$a\in A$ not matched to its $f$-post while its $f$-post is unmatched, a vertex~$a\in A$ not matched to its $f$-post~$f(a)$ while $f(a)$ is matched to a vertex~$a'$ with $f(a') \neq f(a)$, or a vertex $a\in A$ matched neither to its $f$-post nor its $s$-post.
   In the first case, matching $a$ to its $f$-post clearly results in a more popular matching.
   In the second case, matching~$a$ to its $f$-post~$f(a)$ and $a'$ to its $f$-post~$f(a')$ (possibly breaking up the matching of $f(a')$) results in a more popular matching~$M'$ since $a $ and $a'$ prefer $M'$ while at most one vertex in $A$, namely $M(f(a'))$, and three vertices in $B$, namely $M(a)$, $f(a)$, and~$f(a')$, can prefer~$M$.
   \new{Consequently, the weight of agents preferring~$M'$ to~$M$ is $2c$ which is larger than the weight of agents preferring~$M$ to~$M'$ which is at most~$c+3$ (here we use $c > 3$).}
   So assume that the third case applies and neither of the first two cases applies for any vertex.
   Note that this implies that~$a$ prefers~$s(a)$ to $M (a)$.
   We construct a matching~$M^*$ more popular than~$M$ in~$\mathcal{I}$ from~$M$ by matching $a$ to its $s$-post~$s(a)$, and if $s(a)$ is matched by~$M$ to a vertex $a'$, then also matching $a'$ to its $f$-post~$f(a')$ (possibly breaking up the matching of $f(a')$).
   Note that $f(a') \neq s(a) $ by the definition of an $s$-post.
   Matching~$M^*$ is more popular matching~$M$ since $a$ and $a'$ prefer $M$ while at most one vertex in $A$, namely $M(f(a'))$, and three vertices in $B$, namely $M(a)$, $s(a)$, and $f(a')$, can prefer~$M$.
   \new{Consequently, the weight of agents preferring~$M'$ to~$M$ is $2c$ which is larger than the weight of agents preferring~$M$ to~$M'$ which is at most~$c+3$ (here we use $c > 3$).}
  \end{proof}  

  We define a subgraph $H^{\deg (A) \le 2}$ of $G$ that contains all edges of the form $\{a, f(a)\}$ or $\{a, s(a)\}$ for some $a\in A$---no other edge can appear in a popular matching by \Cref{th:abraham} and \Cref{lem:popular-with-0}.
     Our goal is to also restrict the degree of each vertex in $B$ to at most two.
\new{We first briefly sketch the proof before giving the formal reduction and correctness proof.}
%     For every cycle~$C$ in~$H^{\deg (A) \le 2}$, \Cref{lem:popular-with-0} implies that every vertex in~$C$ has to be matched along an edge in~$C$:
     Consider a vertex~$b \in B$ which is contained in a cycle~$C$.
     Each vertex~$a \in A\cap V(C)$ has to be matched along an edge from~$C$ by \Cref{lem:popular-with-0}, and since $|B \cap V(C)| = |A \cap V(C)|$, it follows that every vertex from~$b \in B \cap V(C)$ is also matched along an edge of~$C$.
     Consequently, we can delete every edge incident to a vertex in~$C$ if this edge is not in~$E(C)$.
     This reduces~$H^{\deg (A) \le 2}$ to a graph~$\hct$ that is a disjoint union of cycles and trees.
     For every tree~$T$, one can show that for every vertex~$b\in B$, there are at most two edges incident to~$b$ that appear in a popular matching (see \Cref{lem:deg-le-2} 
     %in the appendix 
     for the proof).
    For every component of~$\hct$ that is a tree, we can compute the set of edges contained in a popular matching (in the instance restricted to solely this tree) via bottom-up induction, resulting in a graph~$\hcp$, which contains every popular edge and whose maximum degree is two by \cref{lem:deg-le-2}. Therefore, $\hcp$ is a disjoint union of paths and cycles, and furthermore, every connected component of $\hcp $ consists of two disjoint popular matchings. \new{This finishes the proof sketch.}
    
    \begin{lemma}\label{lem:b_two}
      We can compute in $O(n + m)$ time a subgraph $\hcp$ of $G$ such that $\hcp$ contains every popular edge, $\hcp $ is a disjoint union of paths and cycles, and every edge in~$\hcp$ is contained in a popular matching in~$\hcp$.
    \end{lemma}
    
    %%%%%%%%%%%%
    
 The proof of \Cref{lem:b_two} consists of several steps. We start by reducing~$H^{\deg (A) \le 2}$ to a subgraph~$\hct$ that contains every popular edge and whose connected components are trees or cycles.
  The basic idea here is that any edge incident to a cycle can be deleted as such an edge cannot be contained in a popular matching.
  
    \begin{numberedclaim}\label{lem:cycle}
   Let $H$ be the graph with edges $\{a, f(a)\}$ and $\{a, s(a)\}$ for all $a\in A$.
   For every cycle~$C$ in $H$, each popular matching~$M$ matches all vertices in $C$ along edges of $C$, i.e.\ for every $v \in V(C)$ with $\{v, v_1\}$ and $\{v, v_2\} \in E(C)$, we have $M(v ) = v_1 $ or $M(v ) = v_2$.
  \end{numberedclaim}

  \begin{proof}
   If the statement is not true, then there exists a cycle $C$ containing a vertex $v$ that is unmatched or is \neu{matched via} \new{an edge~$\{v, M(v)\} \notin E(C)$}.
   \Cref{th:abraham} and \Cref{lem:popular-with-0} imply that $M(a) \in V(C)$ for every $a\in A\cap V(C)$.
   Since $|A\cap V(C) | = |B \cap V(C)|$, it follows that every vertex $b \in B\cap V(C)$ is matched to a vertex $a\in A \cap V(C)$.
   Since no edge $\{a, b\}\in E(G) \setminus E(C)$ with $a\in A\cap V(C)$ and $b\in B \cap V(C)$ can be part of a \new{popular} matching by \Cref{th:abraham} and \Cref{lem:popular-with-0}, the claim follows.
  \end{proof}
  
  Thus, we define $\hct$ to be the graph arising from $H^{\deg(A)\le 2}$ by exhaustively deleting for every cycle $C\in H^{\deg (A) \le 2}$ all edges incident to a vertex in~$V(C)$ if they are not in~$E(C)$.
  Note that $\hct$ consists of a disjoint union of cycles and trees.
  As our goal is to reach a graph that has maximum degree two, we now take care of high-degree vertices in~$\hct$.

    \begin{numberedclaim}\label{lem:deg-le-2}
     For any vertex $v\in V (\hct)$, there are at most two edges incident to $v$ that can appear in a popular matching in~$\hct$.
    \end{numberedclaim}

    \begin{proof}
     As every vertex in~$A$ has degree at most two, the lemma holds for every vertex in~$A$.

     So consider some $b\in B$.
     Assume for a contradiction that there are three popular matchings $M_1$, $M_2$, and $M_3$ such that $b$ prefers $a_1:= M_1 (b)$ to $a_2 := M_2 (b)$\new{, and prefers $a_2$} to $a_3 := M_3 (b)$.
     Let $\witness^i$ be a nice witness of $M_i$ for~$i\in \{1,2,3\}$.
     We distinguish two cases:

     \textbf{Case 1:} $ b = f(a) $ for some $a\in A$.

     \Cref{th:abraham} and \Cref{lem:popular-with-0} imply that $f(a_1) = b$.
     Thus, for every~$i \in \{2,3 \}$, we have $\weightedvote^{M_i} (\{a_1, b\}) = c + 1$, and therefore \new{(as $\wof{a} \le 1$ and $\wof{b'} \le c$ for every~$a \in A$, $b' \in B$, and witness~$\witness$ by \Cref{obs:witness-bounds})}, $\wtwo_b = c = \witness^3_b$.
     \new{Therefore, \Cref{lem:even-odd} implies that} both $\wtwo$ and~$\witness^3$ are \odd\ on the connected component~$C$ of $M_2 \triangle M_3$ that contains~$b$.
     However, since $M_2 \triangle M_3$ is not a cycle (otherwise $\{a_1, b\}$ would have been deleted from $\hct$), it follows that there is a vertex~$u \in V(C)$ that is unmatched in~$M_i$ for some $i\in \{2, 3\}$.
     This implies that $\witness^i_u = 0$ \new{by \Cref{obs:witness-bounds}}.
     By \Cref{lem:even-odd}, it follows that $\witness^i$ is \even\ on $C$, a contradiction to $\witness^i$ being \odd\ on~$C$.

     \textbf{Case 2:} $b = s(a)$ for some $a\in A$.

     By \Cref{th:abraham} and \Cref{lem:popular-with-0}, it follows that $M_1 (a_3) = f(a_3) = M_2 (a_3)$.
     Thus, we have $\weightedvote^{M_i} (\{ a_3, b\}) = -c-1$ for every~$i \in \{\new{1,2}\}$.
     By \Cref{lem:tight}, \new{we have $\woneof{a_3} + \woneof{b} = -c -1 = \wtwoof{a_3} + \wtwoof{b}$, implying} $\wone_b  = -1 = \wtwo_b$ \new{(and $\woneof{a_3} = -c = \wtwoof{c}$)}.
     \new{Therefore, \Cref{lem:even-odd} implies that} both $\wone$ and $\wtwo$ are \odd\ on the connected component~$C$ of~\new{$M_1 \triangle M_2$ }containing~$b$.
     However, since $M_1 \triangle M_2$ is not a cycle (otherwise $\{a_3, b\}$ would have been deleted from $H$), it follows that there is a vertex~$u$ in~$C$ that is unmatched in~$M_i$ for some $i\in \{1,2\}$.
     This implies that $\witness^i_u = 0$ \new{by \Cref{obs:witness-bounds}}.
     By \Cref{lem:even-odd}, it follows that $\witness^i$ is \even\ on~$C$, a contradiction to~$\witness^i$ being \odd\ on~$C$.
     \end{proof}

    Let now $\hcp$ be the graph whose edge set is the disjoint union of popular matchings in each component of $\hct$.
    We now show that $\hcp$ can be computed in polynomial time, starting with the connected components consisting of cycles.
    
    \begin{observation}\label{obs:popular-on-cycles}
      Let $C$ be a cycle in $\hct$.
      Then we can compute the set of popular matchings in~$C$ in linear time.
    \end{observation}

    \begin{proof}
      \Cref{th:abraham} and \Cref{lem:popular-with-0} imply that every popular matching must match every vertex in~$V(C)$.
      Thus, there are only two different candidates for a popular matching, and these can be checked for popularity in \new{linear} time by \Cref{cl:zero_weight}.
      Note that on cycles, a simple dynamic program allows to find a maximum-weight matching in linear time.
      \end{proof}

    The next claim shows how to compute the set of edges in a popular matching for a tree, implying that $\hcp$ can be computed in \new{linear} time. \neu{We remark that this result holds for arbitrary trees, not only for trees in~$\hct$.}
    
    \begin{numberedclaim}\label{lem:trees}
      If $T$ is a tree, then the set of edges that appear in at least one popular matching in~$T$ can be computed in linear time.
    \end{numberedclaim}%\ronecom{ P.35,l.909: Make it clear $T$ is a tree in $H^{C+T}$}

    \begin{proof}
\newcommand{\setX}{X}
\newcommand{\inX}{x}
\newcommand{\noDPentry}{\square}
%\khcom{Could you please go over the DP again and see whether the write-up can be improved? Reviewer 3 complained about it.}\accom{I read through this part again very carefully. I corrected some typos (all marked) and added some comments. Nothing essential though. Don't know how to please R3.}
      We compute the set of edges contained in a popular matching via dynamic programming.
      Fix an arbitrary vertex~$r \in V(T) $ to be the root of~\new{$T$.} %the tree.
      For a vertex $v\in V(T)$, we denote by~$T_v$ the subtree of~$T$ rooted at~$v$\new{; so $T_r = T$}.
      Let~$\setX:= \{-c, 1-c, 2-c , -1 , 0, 1, c-2, c-1, c\}$ be the set of values a nice witness can attain. \new{In this proof, we first define edge set %s $F^{e , \inX}_v$ and 
    $\neu{F}^{e , \inX}_v$, capturing the set of edges\neu{~$f$ such that there is some popular matching~$M$ in~$T_v$ with $e, f \in M$ that fulfills the additional constraint that there is a witness~$\witness$ of~$M$ with~$\wof{v} = x$\neu{.} If $e = \bot$, then we drop the condition that $e \in M$.}
    %which fulfill the additional constraints that $e \in M$ (if $e = \bot$, then this constraint is replaced by the constraint that $v$ is unmatched) and there is a witness~$\witness$ of~$M$ with~$\wof{v} = x$\neu{.}
    Then we present our dynamic program and 
 prove that it indeed computes $\neu{F}^{e , \inX}_v$ and also outputs the set of popular edges for the root node. Finally, we elaborate on the running time.}
      
      \paragraph{\new{General approach of the dynamic program}} The dynamic program performs bottom-up induction on~$T$\new{: starting from the leaves, it decides at each vertex~$v$ whether a matching on~$T_v$ together with a witness~$\witness$ can potentially be extended to a popular matching and a witness on~$T$. To make this decision, it is enough to know the value of~$\wof{v}$ and the vertex \neu{to which} $v$ is matched to (if any)}.
      %Note that to decide whether a matching on~$T_v$ together with a witness~$\witness$ can be extended to a popular matching on~$T$, it is enough to know to which vertex (if any) $v$ is matched and the value of~$\wof{v}$.
      Thus, 
      for each vertex $v\in V(T)$, every~$e\in \delta_T (v) \cup \{\bot\}$ (representing the edge with which $v$ is matched, where $\bot $ indicates that $v$ remains unmatched), and every \new{possible witness} $\inX\in \setX$, the dynamic program \new{creates} %contains
      one \new{edge} \new{set}~$\new{\neu{F}}^{e, \inX}_v$.
      \new{This set shall consist of the edges that appear in any popular matching together with $e$ on~$T_v$ %which contain the edge~$e$ (or leave $v$ unmatched if $e = \bot$)
      and have a witness~$\witness$ with~$\wof{v} = \inX$.
      
      \paragraph{Defining $\neu{F}^{e , \inX}_v$} If $e = \{v, z\}$ is the upward edge incident to~$v$ \new{in $T$}, then any matching containing~$e$ is not contained in~$T_v$, but only in~$T_v \cup \{z\}$\new{. T}o capture this, we define~$T_v^e := T_v \cup \{z\}$ if $e = \{v, z\}$ is the upward edge at~$v$, and $T_v^{e'} := T_v$ for all~$e' \in \bigr(\delta (v) \setminus \{e\}\bigl) \cup \{\bot\}$.}
      If $e = \{v, z\} \in \delta_T (v)$, then define~\new{$\neu{F}^{e, x}_v$} to be the union \new{of all edges appearing in some} popular matching~$M$ on~$\new{T_v^e}$ with~$e \in M$ and having a witness~$\witness$ with~$\wof{v} = \inX$ assuming that at least one such popular matching exists; if no such matching exists, then we set~\new{$\neu{F}^{e, x}_v = \noDPentry$}.
%      \new{We remark that for all edges~$e= \{v, z\} \in \delta_T (v)$ except for the unique upward edge, we have $T_v \cup \{z \} = T_v$.}
      If $e = \bot$, then we define \neu{$\neu{F}^{e, x}_v$} %$G^{e, \new{x}}$
      to be the union \new{of all edges appearing in some} \neu{popular} matching~$M$ on~$T_v$ with~$v$ being unmatched and having a witness~$\witness $ with~$\wof{v} = \inX$ assuming that at least one such popular matching exists; if no such matching exists, then we set~$\neu{F}^{e, \inX}_v = \noDPentry$.
      \medskip
      
      We will now \new{present} a dynamic program computing \new{sets}~$\new{\Ghat}^{e, \inX}_v$ \neu{(we will later show that $\neu{F}^{e, \inX}_v = \Ghat^{e, \inX}_v$ for every $e, \inX$, and $v$)}.
%      for which we will later show that~$F^{e , \inX}_v = G^{e, \inX}_v$.
%      This entry~$F^{e, \inX}_v$ contains the union over all popular matchings~$M$ on~$T_v \cup V_e$, where $V_e = \{z \}$ if $e= \{v, z\}$ and $V_e= \emptyset$ otherwise (i.e.\ $e = \bot$), such that $e \in M$ if $e \in \delta_T (v)$ and $v$ is unmatched otherwise (i.e.\ $e= \bot$) and there is a witness~$\witness$ of~$M$ with~$\wof{v} = \inX$.
%      If no such popular matching exists, then we set $F^{e, \inX}_v := \bot$.
      
       \paragraph{\new{Computing $\Ghat^{e , \inX}_v$ for leaves}} The dynamic program computes the sets~$\Ghat^{e,\inX}_v$ via bottom-up induction (i.e.\ processing the vertices of the tree starting with the leaves and then ``going up'' the tree towards the root).
      For a leaf~$v$, this set can be computed in constant time \new{as follows.
      First note that~$T_v^\bot$ contains only a single vertex (namely~$v$), and thus the empty matching is the unique popular matching (with unique witness~$\wof{v} = 0$ by \Cref{obs:witness-bounds}).
      Consequently,
      we set $\Ghat^{\bot, \inX}_v = \emptyset $ if $\inX = 0$ and %$F^{\bot, \inX}_v = \noDPentry$
      \neu{$\Ghat^{\bot, \inX}_v = \noDPentry$}
      for all $\inX \in X \setminus\{0\}$.
      We also need to compute~$\Ghat^{e, \inX}_v$ for the unique edge~$e = \{v, z\}$ incident to~$v$.
      The graph~$T_v^e$ admits a unique popular matching, namely~$M := \{e\}$.
      Each feasible witness~$\witness$ for~$M$ fulfills~$\wof{z} = - \wof{v} $ (by \Cref{obs:witness-bounds}) and $-w(v) \le \wof{v} \le w(z)$ (and indeed any~$\witness$ fulfilling these two conditions is a witness for~$M$).
      Thus, we set $\Ghat^{e, \inX}_v = \{ e\}$ if $-w(v) \le \inX \le w( z)$, and $\Ghat^{e, \inX}_v = \noDPentry$ for all other values of~$\inX$.
      From \neu{these} observations it follows that $\Ghat_v^{e, x}$ is computed correctly (i.e.\ it coincides with the above defin\neu{i}tion of~$F_v^{e, x}$\neu{)} for each leaf~$v$.}
%      $F_v^{e, x} = G_v^{e, x}$) for each leaf~$v$.}
%      If~$e = \{v, z\}$, then \new{$\witness $ is a witness of the popularity of $\{e\}$ in $T[\{v, z\}]$ if and only if $\wof{v} = - \wof{z} $ (by \Cref{lem:tight}), $\wof{v} \ge - w(v)$, and $\wof{z} \ge - w (z)$.
%      Thus, we have} $F^{e, \inX}_v = \{\{v, z\}\}$ if $\inX \ge - w(v)$ and $- \inX \ge -w(z)$.
%      Otherwise, we have $F^{e, \inX}_v = \noDPentry$.

      \paragraph{\new{Deriving the set of popular edges from the dynamic program}}
      For the root~$r$, the union of%~$F^{e, \inX}_r$
      ~\neu{$\Ghat^{e, \inX}_r$}
      over every~$e\in \delta_T (r) \cup \{\bot\}$ and every~$\inX \in \setX$ is the set of popular edges.
      
      \paragraph{\new{Computing $\Ghat^{e , \inX}_v$ for parent nodes}} \new{We continue by showing} how to compute~$\new{\Ghat}^{e, \inX}_v$ for a vertex~$v$ from the sets~$\new{\Ghat}^{e_d, \inX_d}_d$ for every child $d$ of~$v$, every~$e_d \in \delta_T (d) \cup \{ \bot\}$, and $\inX_d \in \setX$.
      Denote the set of children of $v$ by~$C (v)$. %\accom{This notation is used very sporadically, although $d \in C(v)$ is written out in words many times. I don't know if we should switch to $d \in C(v)$.}\khcom{I am indifferent, do as you prefer.}
      For every child~$d \in C(v)$ \new{such that $d$ is not an endpoint of~$e$}, we define~$S_d^{e, \inX}$ to be the set containing every pair~$(e_d, \inX_d)$ with $\new{\Ghat}^{e_d, \inX_d}_d \neq \noDPentry$ such that $\{d, v\}$ does not violate the condition of a witness.
      More precisely, $S_d^{e, \inX}$ contains every pair~$(e_d, \inX_d)$ with $e_d \in \delta_T (d) \cup \{\bot\}$ such that $\new{\Ghat}^{e_d, \inX_d}_d \neq \noDPentry $ and $\weightedvote^{e_d}_{d} (v) + \weightedvote^{e}_v (d) \le \inX + \inX_d$.
      For the child~$d^*$ with $e = \{v, d^*\}$ \new{(if such a child exists; this is not the case when $e $ is the upward edge or $e = \bot$)}, we set~$S_{d^*}^{e, \inX} := \{(e, -\inX)\}$.
      \neu{An example of~$S^{e,x}_d$ and the computation of~$\Ghat^{e, \inX}_v$ can be seen in \Cref{fig:example-G-S}.}
      \new{The intuition here is that in order to extend the edge~$e$ to a popular matching with a witness~$\witness$ with $\wof{v} = \inX$, this extension needs to contain, for each child~$d$ that is not an endpoint of~$e$, a popular matching~$M_d$ on~$T_d$ together with a witness~$\witness^d$ such that $\weightedvote^{M_d}_d (v) + \weightedvote^{e}_v (d) \le \witness^d + \inX $ (and in fact, this condition is also sufficient).
      The set~$S_d^{e, \inX}$ now precisely contains the \neu{pairs~$(e_d, x_d)$} that fulfill this condition.
      Consequently, the set~$\new{\Ghat}^{e, \inX}_v$ is computed as follows:}
      If $S_d^{e, \inX} = \neu{\emptyset} $ for some child $d \in C(v)$, then \neu{there is no popular matching which contains~$e$ and fulfills $\wof{v} = \inX$ and so} we set $\new{\Ghat}^{e, \inX}_v:= \noDPentry$.
      Otherwise, \neu{taking the union over an arbitrary popular matching in~$T_d$ containing~$e_d$ and having a witness with $\wof{d} = x_d$ for some~$(e_d, x_d) \in S^{e,x}_d$ for every child~$d$ and the edge~$e$ results in a popular matching in~$T_v$, and so} we set $\new{\Ghat}^{e, \inX}_v := \bigl(\{e \}\setminus\{\noDPentry\}\bigr) \cup \bigcup_{d \in C(\new{v})} \bigcup_{(e_d, \inX_d) \in S_d^{e, \inX}} \new{\neu{F}}^{e_d, \inX_d}_d$.

\neu{
\begin{figure}
    \centering
\begin{tikzpicture}[yscale = 0.85]
          \node[terminal, label=90:$v$] (v) at (0,1) {};
          \node[vertex, label={[xshift=-.0cm]180:$c_1$}] (c1) at (-2.,-1) {};
          \node[vertex, label=0:$c_2$] (c2) at (0.,-1) {};
          \node[vertex, label=0:$c_3$] (c3) at (2.,-1) {};

          \node[terminal, label=0:$w_1$] (w1) at (-4, -3) {};
          \node[terminal, label=0:$w_2$] (w2) at (-2, -3) {};
          
          \node[terminal, label=0:$w_3$] (w3) at (-0, -3) {};

          \node[terminal, label=0:$w_4$] (w4) at (2, -3) {};
          \node[terminal, label=0:$w_5$] (w5) at (4, -3) {};
        
          \node[vertex, label={[xshift=-.0cm]180:$x_1$}] (x1) at (-4.,-5) {};
          \node[vertex, label=0:$x_2$] (x2) at (.,-5) {};
          \node[vertex, label=0:$x_3$] (x3) at (2.,-5) {};
          
        \draw (v) edge[] node[pos=0.2, fill=white, inner sep=2pt] {\scriptsize $1$}  node[pos=0.76, fill=white, inner sep=2pt] {\scriptsize $3$} (c1);
        \draw (v) edge[red, ultra thick] node[pos=0.2, fill=white, inner sep=2pt] {\scriptsize $2$}  node[pos=0.76, fill=white, inner sep=2pt] {\scriptsize $1$} (c2);
        \draw (v) edge[] node[pos=0.2, fill=white, inner sep=2pt] {\scriptsize $3$}  node[pos=0.76, fill=white, inner sep=2pt] {\scriptsize $1$} (c3);

        \draw (c1) edge[] node[pos=0.2, fill=white, inner sep=2pt] {\scriptsize $1$}  node[pos=0.76, fill=white, inner sep=2pt] {\scriptsize $1$} (w1);
        \draw (c1) edge[] node[pos=0.2, fill=white, inner sep=2pt] {\scriptsize $2$}  node[pos=0.76, fill=white, inner sep=2pt] {\scriptsize $1$} (w2);
        
        \draw (c2) edge[] node[pos=0.2, fill=white, inner sep=2pt] {\scriptsize $2$}  node[pos=0.76, fill=white, inner sep=2pt] {\scriptsize $2$} (w3);
        
        \draw (c3) edge[] node[pos=0.2, fill=white, inner sep=2pt] {\scriptsize $2$}  node[pos=0.76, fill=white, inner sep=2pt] {\scriptsize $1$} (w4);
        \draw (c3) edge[] node[pos=0.2, fill=white, inner sep=2pt] {\scriptsize $3$}  node[pos=0.76, fill=white, inner sep=2pt] {\scriptsize $1$} (w5);
        
        \draw (w1) edge[] node[pos=0.2, fill=white, inner sep=2pt] {\scriptsize $2$}  node[pos=0.76, fill=white, inner sep=2pt] {\scriptsize $1$} (x1);
        
        \draw (w3) edge[] node[pos=0.2, fill=white, inner sep=2pt] {\scriptsize $1$}  node[pos=0.76, fill=white, inner sep=2pt] {\scriptsize $1$} (x2);
        
        \draw (w4) edge[] node[pos=0.2, fill=white, inner sep=2pt] {\scriptsize $2$}  node[pos=0.76, fill=white, inner sep=2pt] {\scriptsize $1$} (x3);
\end{tikzpicture}
  \captionsetup{singlelinecheck=off}
\caption[]{\neu{An example for the computation of $\Ghat^{e, \inX}_v$ where $e = \{v, c_2\}$ and $\inX = -1$.
The edge~$e$ is depicted in red.
The sets~$S^{e,x}_{c_i}$ for $i \in [3]$ are the following:
\begin{itemize}
  \item[$c_1$:] There are two popular matchings in~$T_{c_1}$, namely $M_{c_1}^1 := \{\{w_1,  c_1\}\}$ and $M_{c_2}^2 := \{\{w_1, x_1\}, \{w_2, c_1\}\}$.
The unique witness~$\wone$ for~$M_{c_1}^1$ fulfills $\woneof{c_1} = c-1$ but we have $\woneof{c_1} + \inX  = c-2 < \weightedvote_{c_1}^{M_{c_1}^1} (v) + \weightedvote_v^{e} (c_1) = -1 + c$.
Thus $(\{w_1, c_1\}, x') \notin S^{e, x}_{c_1}$ for every~$x' \in X$.
Matching~$M_{c_1}^2$ has a unique witness~$\wone$ with~$\woneof{c_1} = c$ and we have $\woneof{c_1} + \inX  = c-1 = \weightedvote_{c_1}^{M_{c_1}} (v) + \weightedvote_v^{e} (c_1)$.
Consequently, $S^{e, x}_{c_1} = \{(\{w_2, c_1\}, c)\}$.
    \item[$c_2$:] The only popular matching in~$T^{e}_{c_2}$ is $\{e, \{w_3, x_2\}\}$, and this matching has a feasible witness~$\wtwo$ with $\wtwoof{c_2} = - \inX$, so we have $S^{e, \inX}_{c_2} = \{(e, - \inX)\}$.
\item[$c_3$:] There are two popular matchings in~$T_{c_3}$, namely $M_{c_3}^1 := \{ \{ w_4, c_3\}\}$ and $M_{c_3}^2 := \{\{w_5, c_3\}, \{w_4, x_3\}\}$.
The unique witness~$\wthree$ for $M_{c_3}^3$ fulfills $\woneof{c_3} = c-1$ and we have $\woneof{c_3} + \inX = c- 2 \ge \weightedvote_{c_3}^{M_{c_3}^1} (v) + \weightedvote_v^{e} (c_3) = 1- c$.
Thus, $(\{w_1, c_1\}, c-1) \in S^{e, x}_{c_3}$.
Matching~$M_{c_3}^2$ has a unique witness~$\wthree$ with~$\wthreeof{c_3} = c$ and we have $\woneof{c_3} + \inX  = c-1 > \weightedvote_{c_3}^{M_{c_3}} (v) + \weightedvote_v^{e} (c_3) = 1-c$.
Thus, $(\{w_5, c_3\}, c) \in S^{e, x}_{c_3}$.
Consequently, we have $S^{e, x}_{c_3} = \{(\{w_4, c_1\}, c-1), (\{w_5, c_1\}, c)\}$.
\end{itemize}
From the sets~$S^{e,x}_{c_i}$ for~$i\in [3]$, it follows that $\Ghat^{e, x}_v = \{\{w_2, c_1\}, \{w_1, x_1\} \} \cup \{\{v, c_2\}, \{w_3, x_2\}\} \cup \{\{w_4, c_3\}, \{w_5, c_3\}, \{w_4, x_3\}\}$.
}
}
    \label{fig:example-G-S}
\end{figure}
}%\accom{A little overuse of 'we have' in the fig caption.}\khcom{Removed roughly half of them.}

%      For each vertex~$v\in V(G)$, the dynamic program takes $O (n \cdot \deg (v))$ time.
%      Thus, the whole dynamic program takes quadratic time.
      \paragraph{\new{$\new{\neu{F}}^{e, \inX}_v$ is computed correctly}} We now prove by bottom-up induction that~$\new{\neu{F}}^{e, \inX}_v$ is computed correctly\neu{, i.e.\ we have $\neu{F}^{e, \inX}_v = \Ghat^{e, \inX}_v$}.
%      \new{In order to do so, we denote by~$\Ghat_v^{e, x}$ the value computed by the dynamic program, while $G^{e, \inX}_v$ (the set we aim to compute) is defined as above.}
      For the leaves, \new{we argued \neu{that $\neu{F}^{e, \inX}_v = \Ghat^{e, \inX}_v$} already when describing how to compute $\Ghat^{e, \inX}_v$}.
      So fix a non-leaf node~$v$, an element~$e \in \delta_T (v) \cup \{\new{\bot}\}$, and some~$\inX \in \setX$.
      \new{We will show $\Ghat^{e, \inX}_v = \neu{F}^{e, \inX}_v$ in two steps: first we show that $\neu{F}^{e, \inX}_v \subseteq \Ghat^{e, \inX}_v$ and then we show that $\Ghat^{e, \inX}_v \subseteq \neu{F}^{e, \inX}_v$.}
%      Let~$F^*$ be the set which~$F^{e, \inX}_v$ shall compute, i.e.\ $F^*$ contains the union over all popular matchings~$M$ on~$T_v^e$ with~$e \in M$ if~$e = \{v, z\}$ (on $T_v$ if $ e= \bot$) and having a witness~$\witness$ with~$\wof{v} = \inX$ if such a matching exists and $F^* = \bot$ otherwise.
      
      \subparagraph{\new{Showing $\neu{F}^{e, \inX}_v \subseteq \Ghat_v^{e, \inX}$.}}
      First we show that $\neu{F}^{e, \inX}_v \subseteq \Ghat^{e, \inX}_v$.
      Consider an edge~$f\in \neu{F}^{e, \inX}_v$.
      By the definition of~$\neu{F}^{e, \inX}_v$\new{,} there is a popular matching~$M$ on~\new{$T_v^e$} with $e \in M$ if $e \in \delta (v)$ and with $v$ being unmatched in~$M$ if $e = \bot$ such that there exists a witness~$\witness$ of~$M$ with $\wof{v} = \inX$.
      This implies that for every child $d \in C(v)$, we have that~$\neu{F}^{\{d, M(d)\}, \wof{d}}_{\neu{d}} \neq \noDPentry$.
      By \new{the} induction hypothesis, we have~$\Ghat^{\{d, M(d)\}, \wof{d}}_{\new{d}} = \neu{F}^{\{d, M(d)\}, \wof{d}}_{\new{d}} \neq \noDPentry$.
      Thus, we have that~$\Ghat^{e, \inX}_v \neq \noDPentry$.
      If $f = \{v, p\}  = e$ where $p$ is the parent of~$v$, then clearly $f \in \Ghat^{e, \inX}_v$.
      Otherwise $f$ is contained in~$E(T_d^{\new{e_d}})$ for some child $d \in C(v)$. 
      By induction, we have that~$f\in \Ghat^{\{d, M(d)\}, \wof{d}}_d$, implying that~\neu{$f\in \Ghat^{e, \inX}_v$}.
      
\subparagraph{\new{Showing $\Ghat^{e, \inX}_v \subseteq \neu{F}^{e, \inX}_v$.}} Next, we show that $\Ghat^{e, \inX}_v \subseteq \neu{F}^{e, \inX}_v$.
      Consider an edge~$f \in \Ghat^{e, \inX}_v$.
      We will show that there exists a popular matching~$M$ containing~$f$ on~$T_v^\new{ e}$ with $e \in M$ if $e \in \delta_T (v)$ and with $v$ being unmatched in~$M$ if $e = \bot$ such that there exists a witness~$\witness$ of~$M$ with $\wof{v} = \inX$.
      Since~$f \in \Ghat^{e, \inX}_v$ there exist $(e_d, \inX_d) \in S_d^{e, \inX}$ for every $d\in C (v)$ such that $f \in \bigl(\{e\}\setminus \{\noDPentry\}\bigr) \cup \bigcup_{d\in C(v) } \Ghat^{e_d, \inX_d}_d$ and~$\Ghat^{e_d, \inX_d}_d \neq \noDPentry$.
      By induction, we have $\Ghat^{e_d, \inX_d}_d = \neu{F}^{e_d, \inX_d}_d$ for every child~$d$.
      Consequently, for every child~$d$, there exists a popular matching~$M_d$ with witness~$\witness^d$ such that~$\wof{d}^d = \inX_d$ and~$M_d (d) \neq v$, and if $f\in E(T_v)$, then we can \neu{construct}%pick
      ~$M_d$ such that it contains~$f$.
      Then~$\bigl(\{e\} \setminus \{\bot\}\bigr) \cup \bigcup_{d \in C(v)} M_d$ is a popular matching containing~$f$ \new{(if $e= f$, then this follows by the definition of the new matching; otherwise $f$ is contained in~$M_d$ for some child~$d \in C(v)$)}, of which a witness can be constructed by setting~$\wof{ u} := \wof{u}^d$ for every~$u\in T_d$ and every child~$d \in C(v)$ and $\wof{v}:= \inX$.
      No conflicting edge can be contained in $T_d$ for some~$d$ because $\witness^d$ is a feasible witness, and $\{v, d\}$ is not conflicting for every child~$d$ because~$\weightedvote_d^{e_d} (v) + \weightedvote_v^e (d) \le \inX + \inX_d$.
      
      \paragraph{\new{Running time}} It remains to analyze the running time.
      Consider a vertex~$v \in V(G)$ and fix some~$\inX \in \setX$.
      We process the edges from~$\delta (v)$ in order according to the preferences of~$v$.
      Thus, when turning from one edge~$e = \{v, c\}$ to the next edge~$e' = \{v, c'\}$, sets~$S_d^{e, \inX}$ and $S_{d}^{e', \inX}$ can only change for~$d \in \{c, c'\}$.
      Consequently, the computation of~$\Ghat^{e, \inX}_v$ can be done in amortized constant time when \new{only the corresponding backtracking information is stored instead of the set of edges appearing in a popular matching.}
      %not storing all edges contained in a popular matching, but only the corresponding backtracking information instead).
      It follows that we need $O(\deg (v))$ time for vertex~$v$, and consequently $O (n)$ time in total.
%\end{comment}
\end{proof}

    We are now ready to prove \Cref{lem:b_two}.
    
    \begin{proof}[Proof of \Cref{lem:b_two}]
     By \Cref{lem:popular-with-0} and \Cref{th:abraham} we can compute a graph $H^{\deg (A) \le 2}$ containing every popular edge and fulfilling~$\deg (a) \le 2$ for every $a\in A$.
     Applying \Cref{lem:cycle} results in a graph~$\hct$ such that $\hct$ contains every popular edge and every connected component of $\hct$ is a cycle or a tree.
     Applying \Cref{lem:trees,lem:deg-le-2}, we compute from~$\hct$ a subgraph~$\hletwo$ such that $\hletwo$ contains every popular edge and has maximum degree two.
     For every connected component~$C$ of~$\hletwo$, we apply
     \Cref{obs:popular-on-cycles} or \Cref{lem:trees} to exhaustively delete every edge not contained in a popular matching in~$C$.
     The resulting graph~$\hcp$ fulfills the requirements of the lemma.
     
     The running time follows from the observation that $H^{\deg (A) \le 2}$ can be computed in $O( n+ m)$ time.
%     , and all subsequent graphs have $O(n)$ edges.
    \end{proof}
    %%%%%%%%%%%%%
    
    While restricting the set of edges which can appear in a popular matching to disjoint paths and cycles as in \Cref{lem:b_two} is a severe restriction, we remark that this does not immediately imply tractability:
    For example, for the instance constructed in the reduction in \Cref{thm:np-hardb}, we know that a set of disjoint cycles (one for each variable and one for each clause gadget) together with three edges from the \unnamedGadget\ such that any popular matching is a subset of these edges, but still deciding the existence of a popular matching is \NP-hard.
    
    \subsubsection{Computing ``local'' witnesses}
    \label{sec:witness}
    
    %Having restricted the set of edges to a graph~$\hcp$ of maximum degree two, w
    We first restrict for each connected component of~$\hcp$ the set of possible nice witnesses, and afterwards dismiss possible witnesses one by one, until we either found a popular matching or have no witness remaining for a connected component and conclude that no popular matching exists.
    The idea here is to order the nice witnesses (and the matching whose popularity the witness proves) by how ``good'' they are on the vertices in~$B$.
    Given such an order, we can apply the following algorithm:
    At the beginning, we assign to each connected component~$C$ the witness that is optimal for~$B\cap V(C)$.
    This induces a matching~$M$ by taking for every connected component the matching whose popularity is verified by the corresponding witness.
    Whenever a conflicting edge, i.e.\ an edge~$\{a, b\}$ with $\weightedvote^M (a) + \weightedvote^M (b) > \wof{a}+ \wof{b}$, occurs, we know that the only way to eliminate this conflict is to change the witness (and possibly also the matching) on the connected component~$C_a$ containing~$a$, or, in one special case, the witness on the connected component containing~$b$.
    Therefore, we dismiss the current witness for~$C_a$ and replace it with the next witness in the order we set up on them.
    This procedure eventually terminates with a feasible witness (and thus also a popular matching), or dismisses all witnesses for a connected component of~$\hcp$.
    In the latter case, we know that no popular matching exists.
    
    We now describe how we order the witnesses for one connected component of~$\hcp$.
    Consider a connected component~$C$ of~$\hcp$, two popular matchings~$M_1$ and $M_2$, and two witnesses~$ \wone$ and~$\wtwo$ of them.
    If we have that (i) $\woneof{b} \ge \wtwoof{b} + 2$ or (ii) $\woneof{b} \ge \wtwoof{b}$ and $M_1 \succeq_b M_2$, then $\wone$ is clearly better than~$\wtwo$ for~$b$, i.e.\ any conflicting edge~$\{a, b\}$ for~$\wone$ will also be conflicting for~$\wtwo$.
    This leads to the following definition:
    \begin{definition}
     Let~$C$ be a connected component of~$\hcp$.
     A witness~$\wone$ of a popular matching \new{$M_1$} on~$C$ \emph{(strongly) dominates} a witness $\wtwo$ of a popular matching \new{$M_2$} on~$C$ \emph{at a vertex~$b \in B\cap V(C)$} if (i) $\woneof{b} \ge \wtwoof{b} + 2$ or (ii) $\woneof{b} \ge \wtwoof{b}$ and $M_1 \succeq_b M_2$.
     \neu{Witness~$\wone$ strongly dominates~$\wtwo$ if $\wone$ strongly dominates~$\wtwo$ for each vertex~$b \in B \cap C$.
     An \odd\ witness~$\wone$ \emph{weakly dominates} a witness~$\wtwo$ if (i) $\wone$ strongly dominates $\wtwo$ or (ii) $\wtwo$ is \even\ and $\wone$ strongly dominates~$\wtwo - \bm{1}$.}
    \end{definition}
%    In this case, we also say that~$\wone$ \emph{dominates}~$\wtwo$ at~$b$.
%    If $\wone$ dominates $\wtwo$ at every~$b\in B \cap C$, then we say that $\wone $ dominates~$\wtwo$.
    From now on, we will refer ``strong dominance" just by ``dominance".
    \neu{Note that only an odd witness can weakly dominate another witness; for even witnesses, weak dominance is not defined.
    We will show} 
    %in the appendix 
    (\Cref{lem:path,lem:dom-comp}) that domination at a vertex~$b\in B$ already implies domination on the whole connected component of~$\hcp$.
    \neu{T}his dominance relationship \neu{then induces an order of the} \odd\ witnesses.
    However, there may be pairs of \odd\ and \even\ witnesses that do not dominate each other, e.g.\ it may happen that $\wone$ is \even, $\wtwo$ is \odd, $\woneof{b} = \wtwoof{b} + 1$ and $M_2 \succ_b M_1$.
    When considering only conflicting edges~$\{a, b\}$ (note that these conflicting edges may be between different components of $\hcp$) for which the witness in the component containing~$a$ is \odd,
    it is irrelevant whether~$\wof{b} = \woneof{b}$ or $\wof{b} = \wtwoof{b}$ holds because $\weightedvote_a^M (b) + \weightedvote_b^M (a) \in \{-1-c, 1-c, c-1, 1+c\}$ for every matching~$M$.
    This idea is captured by weak dominance.
%    Therefore, if $\wone$ is \even\ and $\wtwo $ dominates~$\wone - \bm{1}$, then we say that $\wtwo$ \emph{weakly dominates} $\wone$ at~$b$.
    We now formalize the implications of the (weak) dominance relation.
    
    \begin{lemma}
    \label{lem:dominated-witnesses}
     Let $C$ and $D$ be two connected components of $\hcp$, and let $\witness^C$ and $\witness^D$ be nice witnesses for two popular matchings~$M^C$ and~$M^D$ on~$C$ and~$D$.
     Assume that there is a conflicting edge~$\{a, b\}$ for matching~$M^C \cup M^D$ with $a \in V(C) \cap A$ and $b \in V(D) \cap B$.
     
     If $\witness^C$ is \odd, then $\{a, b\}$ is also conflicting when exchanging~$\witness^D$ for a witness that~$\witness^D$ weakly dominates.
    \end{lemma}
    
    \begin{proof}
     Let~$\wDhat$ be a witness of a popular matching~$M$ on~$D$ that is dominated by~$\witness^D$.
     If $ \wof{b}^D  \ge \wDhatof{b} + 2$, then $\wof{a}^C + \wDhatof{b} \le \wof{a}^{\neu{C}} + \wof{b}^D -2 < \weightedvote_a^{M^C} (b) + \weightedvote_b^{M^D} (a) -2 \le \weightedvote_a^{M^C} (b) + \weightedvote_b^{M} (a) + 2-2 = \weightedvote_a^{M^C} (b) + \weightedvote_b^{M} (a)$ \new{(the strict inequality holds because $\{a,b\} $ is conflicting and the last inequality because~$\weightedvote_{b'}^{M'} (a')\in \{-1,0,1\}$ for every matching~$M'$, vertex~$a' \in A$, $b' \in B$)} and consequently $\{a, b\}$ is also conflicting when exchanging~$\witness^D$ for $\wDhat$.
     
     If $\wof{b}^D \ge \wDhatof{b}$ and $b $ does not prefer $M(b) $ to $M^D (b)$, then $\wof{a}^C + \wDhatof{b} \le \wof{a}^C + \wof{b}^D < \weightedvote_a^{M^C} (b) + \weightedvote_b^{M^D} (a) \le \weightedvote_a^{M^C} (b) + \weightedvote_b^{M} (a) $ and consequently $\{a, b\}$ is also conflicting when exchanging~$\witness^D$ for $\wDhat$.
     
     If $\witness^D$ weakly dominates $\wDhat$ but $\witness^D$ does not dominate~$\wDhat$, then $\wDhat $ is \even\ and $\witness^D + \bm{1}$ dominates~$\wDhat$.
     The above arguments yield $\wof{a}^C + \wDhatof{b} < \weightedvote_a^{M^C} (b) + \weightedvote_b^{M} (a) + 1$.
     Since $\witness^C$ is \odd\ and $\wDhat$ is \even,  we have $\wof{a}^C \in \{-c, 2-c, -1, 1\}$ and $\wDhatof{b} \in \{0, c-1\}$.
     It follows that $\wof{a}^C+ \wDhatof{b} \in \{-c, 2-c, -1, 1, c-2, c\}$, while $\weightedvote_a^{M^C} (b) + \weightedvote_b^{M}  (a) \in \{-1-c, 1-c, c-1, 1+c\}$.
     \new{Using $c > 3$, a case distinction on the possible values which $\wof{a}^C + \wDhatof{b}$ can attain shows} that even $\wof{a}^C + \wDhatof{b} < \weightedvote_a^{M^C} (b) + \weightedvote_b^{M} (a)$ holds, proving the lemma.
    \end{proof}
    
    Next, we show that we can restrict the set of possible nice witnesses for every connected component in such a way that the witnesses are in a weak dominance relation to each other.
    Note that we use here that for every connected component~$C$ of~$\hcp$, each edge in~$C$ is contained in at least one popular matching in~$C$.

        \begin{lemma}\label{lem:fourwitnesses}
        For each connected component~$C$ of $\hcp$, we can compute in linear time a set~$Y$ of up to four nice witnesses with the properties below (depending on whether $C$ contains a single edge, a path, or a cycle), such that for any feasible witness $\witness$ of a popular matching~$M$, there exists a witness~$\wprime \in Y$ that can replace the values~$\wof{v}$ for every~$v\in V(C)$.
        \begin{itemize}
            \item If $C$ contains only a single edge~$e$, then
        \begin{itemize}
            \item $Y = \{\weone, \wetwo, \weeven\}$, where $\weone$ and $\wetwo$ are \odd\ and $\weeven$ is \even, and
            \item $\weone$ dominates $\wetwo$ and $\wetwo$ weakly dominates $\weeven$.
        \end{itemize}
        
            \item If $C$ is a path containing at least three \new{vertices}, then 
            \begin{itemize}
                \item $Y = \{\wpone, \wptwo\}$ where $\wpone $ is \odd\ and $\wptwo$ is \even, and
                \item $\wpone$ weakly dominates~$\wptwo$.
            \end{itemize}
        
        \item If $C$ is a cycle, then
        \begin{itemize}
            \item $Y = \{\wcone, \wcthree\}$ or $Y = \{\wcone, \wcthree, \wcmtwo, \wceven\}$ where $\wcmtwo$, $\wcone$, and $\wcthree$ are \odd\ and $\wceven$ is \even, and
            \item $\wcone $ dominates $\wcthree$ or $\wcone$ dominates $\wcmtwo$, witness~$\wcmtwo $ dominates $\wcthree$, and $\wcthree$ weakly dominates~$\wceven$.
        \end{itemize}
        \end{itemize}
    \end{lemma}

    The proof of \Cref{lem:fourwitnesses} is split into three parts (\Cref{obs:edge,lem:path,lem:cycle-witness}), where \Cref{obs:edge} shows the statement for components consisting of a single edge, \Cref{lem:path} shows the statement for paths with at least three vertices, and \Cref{lem:cycle-witness} shows the statement for cycles.
    
    We first consider components consisting of a single edge.
    \begin{numberedclaim}\label{obs:edge}
     Let $P$ be a connected component of $\hcp$ consisting of a single edge~$ e= \{a, b\}$, and assume that~$G$ admits at least one popular matching.
     
     Then there exist three witnesses~$\weone$, $\wetwo$, and $\weeven$ such that $\weone $ dominates $\wetwo$, witness~$\wetwo$ weakly dominates $\weeven$, and for every popular matching~$M$ in~$G$ with nice witness~$\witness$, there exists some $\witness^e \in \{\weone,\wetwo, \weeven\}$ such that $\witness'$ defined by $\wof{a}' := \wof{a}^e$, $\wof{b}' := \wof{b}^e$, and $\wof{v}' := \wof{v}$ is a witness of~$M$.
    \end{numberedclaim}

    \begin{proof}
     Let $M$ be a popular matching.
     By \Cref{lem:nice}, there exists a nice witness~$\witness$.
     Furthermore, $M$ contains~$e$:
     \new{If $M$ would not contain~$e$, then $a$ and $b$ are unmatched (since $\hcp$ contains every popular edge and $e$ is a connected component of~$\hcp$), implying that $M \cup \{e\}$ is more popular than~$M$, a contradiction.
     By \Cref{obs:witness-bounds}, it follows that}~$\wof{a} = - \wof{b}$.
     
     We first consider the case that $b = f(a)$.
     We define
    \begin{itemize}
        \item $\weone[a]:= -c$ and $\weone[b] := c$,
        \item $\wetwo[a] := 2-c$ and $\wetwo[b]:= c-2$, and
        \item $\weeven[a] := 1-c$ and $\weeven[b] := c-1$.
    \end{itemize}
    Clearly, $\weone$ dominates $\wetwo$ and $\wetwo$ weakly dominates $\weeven$.
     It remains to show that, given a popular matching~$M$ with nice witness~$\witness$, there exists some~$\witness^e \in \{\weone, \wetwo, \weeven\}$ such that $\witness'$ defined by $\wof{a}' := \wof{a}^e$, $\wof{b}' := \wof{b}^e$, and $\wof{v}' := \wof{v}$ \new{for $v \in V(G) \setminus\{a, b\}$} is a witness of~$M$.
     If $\wof{a} \le 2-c$, then \new{\Cref{lem:nice} implies that} $\wof{a} = \wof{a}^e$ and $\wof{b}= \wof{b}^e$ for some $\witness^e \in \{\weone, \wetwo, \weeven\}$, and this $\witness^e$ fulfills the claim.
     Thus, we assume $\wof{a} \ge -1$\new{, which implies $\wof{b} = - \wof{a} < \wetwo[b]$}.
     We claim that we can choose~$\witness^e := \wetwo$ in this case.
     We assume for a contradiction that there exists a conflicting edge~$\{a', b'\}$.
     Note that $\witness'_v \ge \witness_v $ for all $v\neq a$, and thus we have $a = a'$.
     Since $b = f(a)$, we have $\weightedvote^M_a (\{a, b'\}) = -c$ and consequently $\weightedvote^M (\{a, b'\}) \le 1-c$.
     Because $\witness'_{b'} \ge \neu{-}\new{w(b\neu{'}) =}-1$, it follows that $\{a', b'\}$ is not conflicting, a contradiction.
     
     We now handle the case $b = s(a)$.
     We define
    \begin{itemize}
        \item $\weone[a]:= -1$ and $\weone[b] := 1$, 
        \item $\wetwo[a] := 1$ and $\wetwo[b]:= -1$, and 
        \item $\weeven[a] := 0$ and $\weeven[b] := 0$.
    \end{itemize}
    Again, $\weone$ clearly dominates $\wetwo$ and $\wetwo$ weakly dominates $\weeven$.
     It remains to show that, given a popular matching~$M$ with nice witness~$\witness$, there exists some~$\witness^e \in \{\weone, \wetwo, \weeven\}$ such that $\witness'$ defined by $\wof{a}' := \wof{a}^e$, $\wof{b}' := \wof{b}^e$, and $\wof{v}' := \wof{v}$ \new{for every $v\in V(G) \setminus \{a, b\}$} is a witness of~$M$.
     If $\wof{a} \ge -1$, then \new{\Cref{lem:nice} implies that} $\wof{a} = \wof{a}^e$ and $\wof{b}= \wof{b}^e$ for some $\witness^e \in \{\weone, \wetwo, \weeven\}$, and this $\witness^e$ fulfills the claim.
     Thus, we assume that $\wof{a} \le 2-c$.
     Then $\weightedvote^M (\{a, f(a)\}) \ge c -1 > 2  = 2 - c + c \ge \wof{a} + \wof{f(a)}$\neu{. T}\new{%t
     he first inequality holds as $a$ prefers~$f(a)$ to $M(a) = s(a)$ by definition of~$f(a)$, and the last inequality follows from $\wof{b'} \le c$ for every $b' \in B$ by \Cref{obs:witness-bounds}}, a contradiction to $\witness$ being a witness of~$M$.
     \end{proof}
    
    We now turn to paths with at least three vertices.
    We start with a helpful claim.
      
    \begin{numberedclaim}\label{lem:even-paths}
     Let $M_1 $ and $M_2$ be two popular matchings.
     Then every path $P$ in $M_1 \triangle M_2$ has one end vertex in~$A$ and the other one in~$B$.
    \end{numberedclaim}

    \begin{proof}
     Assume for a contradiction that $P$ has both end vertices in~$A$, and let $a_1$ and $a_2$ be these end vertices.
     Without loss of generality $a_1$ is matched in $M_1$ and $a_2 $ is matched in $M_2$.
     Let $\wone $ and $\wtwo$ be nice witnesses of $M_1 $ and $M_2$.
     Since $a_1$ is unmatched in $M_2$, it follows that $\wtwoof{a_1} = 0$ and thus, $\wtwo$ is \even\ on~$P$.
     Since $a_2$ is unmatched in $M_1$, it follows that $\woneof{a_2} = 0$ and thus, $\wone $ is \even\ on~$P$.
     Let $b:= M_2 (a_2)$, $a:= M_1 (b)$, and $b' := M_2 (a)$.
     We have $\weightedvote_{a_2}^{M_1} (b) = \new{c}$.
     \new{Further, because $\woneof{a_2} = 0$ and $\wof{b} \le c$ (as $\wof{b^*} \le c$ for every~$b\in B$ by \Cref{obs:witness-bounds}), we have} $\woneof{a_2} + \woneof{b} \le c$\new{.
     I}t follows \new{from} \Cref{lem:tight} that $\weightedvote^{M_1} (\{a_2, b\}) = \woneof{b} = \new{c - 1} = -\woneof{a}$.
     Thus, $a \succ_b a_2$.
     Again \Cref{lem:tight} \new{(this time applied to $\{a, b'\}$)} implies that $\weightedvote^{M_1} (\{a, b'\}) = 1-c$.
     Consequently, $b \succ_a b'$ holds.
     Thus, we have $\weightedvote^{M_2} (\{a, b\}) = c+1$, and consequently, $\wtwoof{a} =1 $ and $\wtwoof{b} = c$.
     However, this implies that $\wtwo$ is \odd, a contradiction to $\wtwo$ being \even.
     
     The case that both end vertices of $P$ are in~$B$ leads to a contradiction by symmetric arguments.
    \end{proof}

    We now describe the witnesses for path components.

    \begin{numberedclaim}\label{lem:path}
     Let $P $ be a connected component of $\hcp$ that is a path with at least three vertices and let $M_1$ and~$M_2$ be the two popular matchings on~$P$ such that $M_2$ leaves at least one vertex in~$P$ unmatched.
     
     Then there is a unique nice witness~$\wpone$ for $M_1$ and a unique nice witness~$\wptwo$ for~$M_2$, and $\wpone $ weakly dominates $\wptwo$.
    \end{numberedclaim}

    \begin{proof}
     Let $M_1^P$ and $M_2^P$ be the two popular matchings in~$P$.
     By \Cref{lem:even-paths}, one end vertex~$a_1$ of $P$ is in~$A$ and the other end vertex is contained in~$B$.
     Let $P = \langle a_1, b_1, a_2, b_2, a_3, b_3, \dots, a_k, b_k \rangle$.
     We assume without loss of generality that $a_1$ (and $b_k$) are unmatched in \new{$M_2^P$}.
     First, we prove by induction on $i$ that for every $i\le \lceil \frac{k}{2} \rceil$,
     \begin{itemize}
         \item $f(a_{2i}) = b_{2i-1} = f(a_{2i - 1})$ and $s(a_{2i}) = b_{2i} = s(a_{2i+1})$ (in particular, $a_{2i}$ and $b_{2i}$ exist),
         \item $ a_{2i} \succ_{b_{2i - 1}}a_{2i - 1}$ and if $2i < k$, then $a_{2i} \succ_{b_{2i}} a_{2i+1}$,
         \item every witness~$\wpone$ for~$M_1^P$ fulfills $\wpone_{a_{2i -1}} = -c$, $\wpone_{a_{2i}} = 1$, $\wpone_{b_{2i-1}}  = c$, and $\wpone_{b_{2i}}  = -1$, and
         \item every witness~$\wptwo$ for $M_2^P$ fulfills $\wptwo_{a_{2i -1}} = 0$, $\wptwo_{a_{2i}} = 1-c$, $\wptwo_{b_{2i-1}}  = c - 1$, and $\wptwo_{b_{2i}}  = 0$.
     \end{itemize}
     From this, one easily observes that $\wpone$ weakly dominates $\wptwo$ and the lemma follows.

    \new{
    We first show the first bullet point.
    By \Cref{lem:popular-with-0,th:abraham}, $a_1$ being unmatched in~$M_2^P$ implies that $b_1 = M_1^P (a_1) = f(a_1)$.
    Note that $b_1 = f(a_1) $ implies $b_1 \neq s(a_2)$ by the definition of~$s (a_2)$.
    Applying \Cref{lem:popular-with-0,th:abraham} to $a_2$ implies that $f(a_2) = b_1$ and $ s (a_2) = b_2$ (note that $a_2$ needs to exists as $b_1 = f(a_1)$ cannot be unmatched in a popular matching by \Cref{lem:popular-with-0,th:abraham}).
    Iterating these arguments now implies the first bullet point:
    Because $b_{2i} = s(a_{2i})$, \Cref{lem:popular-with-0,th:abraham} imply that $b_{2i} = s (a_{2i+1})$ which then implies $f(a_{2i+1}) = b_{2i+ 1}$ (and in particular the existence of $a_{2i+1} $ as $b_{2i+1}$ cannot be unmatched in a popular matching).
    This now implies that $f(a_{2i+2} ) = b_{2i+1}$ which implies $s(a_{2i+2})  = b_{2i+2}$.

    We continue by showing the remaining bullet points.
    }
     For $i =  \neu{1}$, due to the popularity of~$M^P_2$, it follows that $a_2 \succ_{b_1} a_1$ (otherwise $M_2': = (M^P_2 \setminus \{\{a_2, b_1\}\}) \cup \{ \{a_1, b_1\}\}$ is more popular).
     As~$a_1$ is unmatched in~$M_2$, we have $\wptwo_{a_1} = 0$ \new{and $\weightedvote^{M_2}_{a_1} (b_1) = c -1$}.
     By \Cref{lem:tight}, it follows that $\wptwo_{b_1} = \new{\weightedvote^{M_2} (\{a_1, b_1\} ) = }c-1$ and thus $\wptwo_{a_2} = 1-c$.
%     Therefore, $b_1$ prefers $a_2$ to~$a_1$.
     Applying again \Cref{lem:tight} to the edge $\{a_2, b_2\}$, we conclude that $\wptwo_{b_2} = 0$, and if $a_3$ exists, then $b_2$ prefers $a_2$ to $a_3$.
     Since $\weightedvote^{M_1} (\{a_2, b_1\}) = c+ 1$, it follows that $\wpone_{b_1} = c$ and $\wpone_{a_2} = 1$, implying $\wpone_{b_2} = -1$ and $\wpone_{a_1}  = -c$.

     The argument for the induction step is basically identical:
     %Since
     \neu{As} ${\wptwo_{b_{2i}} = 0}$, we have $\wptwo_{a_{2i+1}} = 0$.
     By \new{the first bullet point, $f(a_{2i+1}) = b_{2i+1}$, which implies $\weightedvote^{M_2}_{a_{2i+1}} (b_{2i+1}) = c$.
     Because $\wptwo_{b_{2i+1}} \le c$ by \Cref{obs:witness-bounds} and $\wptwo_{a_{2i+1}} + \wptwo_{b_2i+1} \ge \weightedvote^{M_2} (\{a_{2i+1}, b_{2i+1}\} )$, we have $\weightedvote^{M_2} (\{a_{2i+1}, b_{2i+1}\} ) = c-1$.
     Applying \Cref{lem:tight} to~$\{a_{2i+1}, b_{2i+1}\}$}, it follows that $\wptwo_{b_{2i+1}} = c-1$ and thus $\wptwo_{{2i+2}} = 1-c$.
     Therefore, $b_{2i + 1}$ prefers $a_{2i+2}$ to $a_{2i+1}$.
     Applying again \Cref{lem:tight} to the edge $\{a_{2i+2}, b_{2i+2}\}$, we conclude that $\wptwo_{b_{2i+2}} = 0$, and if $a_{2i+3}$ exists, then $b_{2i+2}$ prefers $a_{2i+2}$ to~$a_{2i+3}$.
     Since $\weightedvote^{M_1} (\{a_{2i+2}, b_{2i+1}\}) = c+ 1$, it follows that $\wpone_{b_{2i+1}} = c$ and ${\wpone_{a_{2i+2}} = 1}$, implying $\wpone_{b_{2i+2}} = -1$ and $\wpone_{a_{2i+1}}  = -c$.
    \end{proof}

    We now turn to the cycle components of $\hcp$.
%    \neu{Recall} that a witness~$\wone $ of~$M_1$ dominates a witness~$\wtwo$ of~$M_2$ at a vertex~$b \in B$ if $\woneof{b} \ge \wtwoof{b} + 2$ or $\woneof{b} \ge \wtwoof{b} $ and $b$ prefers~$M_1$ to $M_2$.
    First, we show that if $\wone$ dominates~$\wtwo$ at one vertex~$b \in B$, then this already implies that $\wone $ dominates $\wtwo$.
    
    \begin{numberedclaim}\label{lem:dom-comp}
      Let~$C$ be a cycle component of $\hcp$.
      Let $M_1$ and $M_2$ be two popular matchings in $C$, and $\wone $ and $\wtwo$ be \odd\ witnesses of $M_1$ and $M_2$ such that $\wone $ dominates $\wtwo$ at some $b^* \in V(B) \cap V(C)$.
      Then $\wone $ dominates $\wtwo $ at every $b \in V(B) \cap V(C)$.
    \end{numberedclaim}

    \begin{proof}
      Let $\wone$ be a witness of~$M_1$, and $\wtwo$ be a witness of~$M_2$.
     If $M_1 = M_2$, then the claim follows from \Cref{lem:tight}.
     Thus, we assume~$M_1 \neq M_2$.
     Let $ M_1 = \{\{a_i, b_i\} : i \in [k]\}$ for some $k \in \mathbb{N}$, and $M_2 = \{\{a_{i+1}, b_i\} : i \in [k -1]\} \cup \{\{a_1, b_k\}\}$.
     
     Let $b^* = b_j$ such that $\wone$ dominates $\wtwo $ at~$b^*$.
     We will show that $\wone $ dominates $\wtwo$ also at $b_{j+1}$, implying the claim.
     To simplify notation, we assume without loss of generality that $j = 1$.
     
     \textbf{Case 1: }$b_1 = f(a_1)$.
     
     \new{By \Cref{th:abraham,lem:popular-with-0}, $b_1, b_2, \dots, b_k$ alternates between $f$-posts and $s$-posts, implying that} $b_{2i -1} = f(a_{2i -1} ) = f(a_{2i})$ for every~$i$, and $b_{2i} = s (a_{2i-1} ) = s (a_{2i})$.
     Because $a_2$ prefers $b_1$ to $M_1 (a_2) = b_2$, we have \new{the following: $\weightedvote^{M_1} (\{a_2, b_1\}) \ge c-1$. This, together with $\wof{a_2} \le 1 $, by \Cref{obs:witness-bounds} implies} $\woneof{b_1} \ge c-2$.
     Similarly, because $a_1$ prefers $b_1$ to $M_2 (a_1) = b_k$, we have $\wtwoof{b_\neu{1}} \ge c-2$.
     
     \textbf{Case 1.1: }$\woneof{b_1} = c-2$.
     
     We will show that this case leads to a contradiction and therefore cannot occur.
     Since $\wone $ dominates $\wtwo$ at $b_1$, it follows that $\witness^2_{b_1} = c-2$, and $b_1$ prefers~$a_{1}$ to $a_{2}$.
     \new{This implies} $\weightedvote^{M_2} (\{a_1, b_1\}) = c + 1$.
     Since $\witness^2_{a_{1}} \le 1$ \new{by \Cref{obs:witness-bounds}}, this contradicts $\wtwo$ being a witness of~$M_2$.
     
     \textbf{Case 1.2: }$\woneof{b_1} = c$.
     
     \textbf{Case 1.2.1: }$\wtwoof{b_1} = c-2$.
     
     Then $b_1$ prefers $a_{2} $ to $a_{1}$, since otherwise $\weightedvote^{M_2} (\{a_1, b_1\}) = c+1 > \wtwoof{a_1} + \wtwoof{b_1}$ \new{(using $\wtwoof{a_1} \le 1$ by \Cref{obs:witness-bounds} for the inequality)}, contradicting~$\witness^2$ being a witness of $M_2$.
     \new{Thus, we have $\weightedvote^{M_1} (\{a_2, b_1\}) = c+1$.
     Since $\woneof{b_1} = c$, we have} $\witness^1_{a_{2}} = 1$, which implies $\witness^1_{b_{ 2}} = -1$ \new{by \Cref{obs:witness-bounds}}.
     Because $\witness^2_{a_{2}} = 2-c$ \new{(due to $\wtwoof{b_1} = c-2$ and \Cref{obs:witness-bounds}) and $\weightedvote^{M_2}  (\{a_2, b_2\} ) \in \{-1-c, 1-c\}$}, \Cref{lem:tight} implies $\witness^2_{b_{2} } = -1$ \new{and $\weightedvote^{M_2 } (\{a_2, b_2\} ) = 1-c$.
     In particular, $b_2$ prefers~$M^1(b_2) = a_2$ to~$M^2 (b_2) = a_3$, implying that $\wone$ dominates $\wtwo$ at~$b_2$.}

     \textbf{Case 1.2.2: }$\wtwoof{b_1} = c$.
     
     Then we have $\wtwoof{a_2} = -c$ \new{by \Cref{obs:witness-bounds}}.
     Since $\wone $ dominates $\wtwo$ at $b_1$, we have $a_1 \succ_{b_{1}} a_2$.
     \new{Thus, we have $\weightedvote^{M_1} (\{b_1, a_2\}) = c-1$.
     \Cref{lem:tight} implies} that $\woneof{a_2} = -1$ and \new{consequently} $\woneof{b_2} = 1$ \new{by \Cref{obs:witness-bounds}}.
     If $\wtwoof{b_2} = -1$, then $\wone $ dominates $\wtwo$ at $b_2$.
     Otherwise we have $\wtwoof{b_2} \new{\ge} 1$.
     \new{Since $\weightedvote^{M_2} (\{a_2, b_2\}) \in \{-c-1, 1- c\}$, \Cref{lem:tight} implies $\weightedvote^{M_2} (\{a_2, b_2\}) = 1-c $ and $\wtwoof{b_2} = 1$.
     This implies} $a_2 \succ_{b_{2}} a_3$, and therefore, $\wone$ dominates $\wtwo $ at $b_2$.
    
     \textbf{Case 2: }$b_1 = s (a_1)$.
     
     \new{By \Cref{th:abraham,lem:popular-with-0}, $b_1, b_2, \dots, b_k$ alternates between $f$-posts and $s$-posts, implying that} $b_{2i -1} = s(a_{2i -1} ) = s(a_{2i})$ for every $i$, and $b_{2i} = f (a_{2i-1} ) = f (a_{2i})$.
     \new{Thus, we have $\weightedvote^{M_1} (\{a_2, b_2\}) \in \{-c-1, -c+1\}$ and $\weightedvote^{M_2} (\{a_1, b_1\}) \in \{-c-1, -c + 1\}$.}
     \Cref{lem:tight} \new{now} implies that $\woneof{b_1} \le 1$ and $\wtwoof{b_2} \le 1$.
     
     \textbf{Case 2.1: }$\woneof{b_1} = -1$.
     
     Since $\wone $ dominates $\wtwo$ at $b_1$, we have $\wtwoof{b_1} = -1$, and $a_1 \succ_{b_{1}} a_2$.
     \new{\Cref{lem:tight} applied to $\{a_2, b_1\}$ implies} $\woneof{a_2} = -c$ \new{which then implies} $\woneof{b_2} = c$.
     \new{\Cref{obs:witness-bounds} implies} $\wtwoof{a_2} = 1$.
     If $\wtwoof{b_2} = c -2$, then $\wone$ dominates $\wtwo$ at $b_2$.
     Otherwise, we have $\wtwoof{b_2} = c$.
     \new{Applying} \Cref{lem:tight} \new{to $\{a_2, b_2\}$ now implies} that $a_2 \succ_{b_{2}} a_3$, and thus $\wone$ dominates $\wtwo$ at~$b_2$.
     
     \textbf{Case 2.2: }$\woneof{b_1} = 1$.
     
     If $a_1 \succ_{b_1} a_2$, then $\weightedvote^{M_1} (\{a_2, b_1\}) = -1 -c < -c + 1 \le \woneof{a_2} + \woneof{b_1}$, a contradiction to \Cref{lem:tight}.
     Thus, we have $a_2 \succ_{b_1} a_1$.
     \new{Applying \Cref{lem:tight} to $\{a_2, b_1\}$ implies} $\woneof{a_2} = -c$ and thus $\woneof{b_2} = c$ \new{by \Cref{obs:witness-bounds}}.
     As $\wone$ dominates $\wtwo$ at $b_1$, it follows that $\wtwoof{b_1} = -1$ and thus $\wtwoof{a_2} = 1$ \new{by \Cref{obs:witness-bounds}}.
     If $\wtwoof{b_2} \new{\le} c-2$, then $\wone$ dominates $\wtwo $ at $b_2$.
     Otherwise $\wtwoof{b_2} = c$, and \new{applying} \Cref{lem:tight} \new{to $\{a_2, b_2\}$ implies} $a_2 \succ_{b_2} a_3$, and thus $\wone$ dominates $\wtwo$ at~$b_2$.
     \end{proof}
    
    We can now show that the witnesses in a cycle component (weakly) dominate each other.
    
    \begin{numberedclaim}\label{lem:cycle-witness}
     Let $C$ be a connected component of $\hcp$ that is a cycle.
     Then $C$ has two popular matchings $M_0^C$ and $M_1^C$ such that $M_0^C$ has exactly one (\odd) witness~$\wcmtwoof{}$, while $M_1^C$ has either one witness~$\wcone$ or three nice witnesses~$\wcone, \wcthree$, and~$ \wctwo$.
     
     If $M_1^C$ has one witness, then $\wcone $ dominates $\wcmtwo$.
     
     If $M_1^C$ has three nice witnesses, then $\wcone$ dominates $\wcmtwo$, witness~$\wcmtwo$ dominates $\wcthree$ and $\wcthree$ weakly dominates $\wctwo$.
    \end{numberedclaim}

    \begin{proof}
     Since every vertex $a\in A \cap V(C)$ can have at most two edges to $f(a)$ and $s(a)$, it follows that $f(a) \in V(C)$ for all $a\in V(C) \cap A$.
     Let $b\in B \cap \hcp$ be the $f$-post of at least one vertex in $ A$.
     Then $b$ prefers one of its neighbors in~$C$ to the other, say, $a_1 \succ _b a_2$.
     Assuming that $M_0^C$ contains $\{a_2, b\}$, it follows that $\weightedvote^{M_0^C} (\{a_1, b\}) =  c + 1$.
     Thus, we have $\wof{a_1} = 1$ and $\wof{b} = c$ for every witness of $M_0^C$\new{, implying that every witness of~$M_0^C$ is \odd}.
     By \Cref{lem:tight}, it follows that the witness of~$M_0^C$ is unique.
     
     Considering $M_1^C$, note that we have $\weightedvote^{M_1^C} (\{ a_2, b\}) = c- 1$. 
     \neu{Let $\witness$ be a nice witness of~$M_1^C$.}
     \new{Because $\wof{a} \le 1 $ by \Cref{obs:witness-bounds}}, we have $\wof{b} \in \{c-2, c-1, c\}$.
     If $\wof{b} = c-1$, then also \new{$\witness^+$ defined as $\witness^+_v := \begin{cases}\witness_v+1 & v \in A\\\witness_v-1 & v \in B\end{cases}$ and $\witness^-$ defined as $\witness^-_v := \begin{cases}\witness_v - 1 & v \in A\\\witness_v + 1 & v \in B\end{cases}$} are feasible witnesses.
     If there exists a feasible witness~\neu{$\wcone$ with $\wconeof{b} = c$ and a feasible witness $\wthree $} with $\wthreeof{b} = c-2$, then also \new{$\wctwo$ defined as $\wctwoof{v} := \begin{cases}\wconeof{v} +1 & v\in A\\\wconeof{v} -1 & v \in B\end{cases}$} is a feasible witness.
     Consequently, $M_1^C$ admits either one witness \new{(with $\wof{b} \in \{c, c-2\}$)} or three witnesses.
%     The claim now follows by \Cref{lem:tight}.

     \new{It remains to show that the witnesses (weakly) dominate each other.
     First a}ssume that~$M_1^C$ admits exactly one witness~$ \wcone$.
     Then we have $\wconeof{b} \in \{c, c-2\}$.
     \new{If $\wconeof{b} = c-2$, then $\wcmtwo$ dominates $\wcoone$ at~$b$.
     Otherwise $\wconeof{b} = \wcmtwoof{b}$, and because~$b$ prefers $M_1^C$ to~$M_0^C$, we have that $\wcone$ dominates $\wcmtwo$ at~$b$.}
%     In both cases}, either $\wcone$ dominates $\wcmtwo$ at~$b$ or $\wcmtwo$ dominates $\wcone$ at~$b$.
     \Cref{lem:dom-comp} implies that either $\wcone$ dominates $\wcmtwo$ or $\wcmtwo$ dominates~$\wcone$.
     The claim follows (possibly by switching the names of $M_0^C$ and $M_1^C$).
     
     Assume that~$M_1^C$ admits three witnesses~$\wcone$, $\wceven$, and $\wcthree$ with $\wconeof{b} = c$, $\wcevenof{b} = c-1$, and $\wcthreeof{b} = c-2$.
     Because $b$ prefers $M_1^C$ to $M_0^C$, it follows that $\wcone$ dominates $\wcmtwo$ at $b$.
     Since $\wcmtwoof{b} = c = \wcthreeof{b} + 2$, it follows that $\wcmtwo$ dominates $\wcthree$ at~$b$.
     \Cref{lem:dom-comp} implies that $\wcone$ dominates $\wcmtwo$ and $\wcmtwo$ dominates $\wcthree$.
     Since \new{$\wcthreeof{b} + \bm{1} = \wcevenof{b}$ for every~$b \in B$}, it follows that $\wcthree$ weakly dominates $\wceven$.
%     
%     This finishes the proof of \Cref{lem:fourwitnesses}.
     \end{proof}
    
    We remark that \Cref{lem:fourwitnesses} requires that each agent from~$A$ has the same weight and each agent from~$B$ has the same weight (see \ref{apx:no-dominance} for an example where agents from~$B$ have different weights and the witnesses are not in a dominance relation).
%    Considering instances with $w (a) = c' \le 3$ for every~$a\in A$ for some constant $c' $ and $w(b) = 1$ for every~$b \in B$, the reduction from \Cref{thm:np-hardb} shows that there is no dominance relation between witnesses when $w (a) = c \le 2$ for every~$a \in A$, although the instance constructed in the reduction admits a subgraph which contains every popular edge and is a disjoint union of paths and cycles.
    
    \subsubsection{Constructing a ``global'' witness}
    \Cref{lem:fourwitnesses,lem:b_two} lead to the following algorithm \new{(see \Cref{alg:popular} for a pseudocode description)}. We start by \new{ordering the witnesses according to the weak dominance relation and} assigning to each edge component the \odd\ witness $\weone$, to each path component~$\wpone$, and to each cycle component~$\wcone$ \new{(Lines~\ref{line:wpm-init-start} to~\ref{line:wpm-init-end})}.
    Whenever we encounter a conflicting edge~$\{a, b\}$, we distinguish two cases:
    If the witness in the component containing $a$ is \even, then this will be the only witness remaining for the component containing~$a$ and we dismiss the witness in the component~$C$ containing~$b$ \new{(see Lines~\ref{line:odd-start} to~\ref{line:odd-end})}.
    Otherwise, we dismiss the witness in the component~$C$ containing~$a$ (\new{Lines~\ref{line:even-start} to~\ref{line:even-end})}.
    In both cases, we assign a new witness (together with a matching) to $C$ as follows:
    If there is an undismissed \odd\ witness of $C$, then we assign this witness to $C$ (the only case when this witness is not unique is if $C$ is a cycle component and neither $\wcthree$ nor~$\wcmtwoof{}$ have been dismissed; in this case we assign $\wcmtwoof{}$ to $C$).
    If no undismissed \odd\ witness of $C$ exists, then we assign the \even\ witness to~$C$.
    If we eventually dismissed all witnesses for one component, then we conclude that there is no popular matching \new{(Lines~\ref{line:wpm-no-1} and~\ref{line:wpm-no-2})}.
    Otherwise, we eventually found a matching together with a witness of it, and we return this matching.
    \new{
	\begin{algorithm}
		\caption{%Algorithm 
  for \wpm}\label{alg:popular}
		\begin{algorithmic}[1]
			\Input{A \wpm\ instance $\mathcal{I}$}
			\Output{A popular matching in $\mathcal{I}$ or \no\ if no popular matching exists}
			\State Compute $\hcp$.\label{line:hcp}
			
			\ForEach{connected component~$C$ of $\hcp$}\label{line:wpm-init-start}
			
			\State Compute the set of popular matchings together with their nice witnesses.
			
			\State Create a list~$L (C) $ of the witnesses by\label{line:order}
			\begin{itemize}
			 \item $L(C) = (\vect{0})$ if $C $ contains only one vertex,
			 \item $L(C) = (\weone, \wetwo, \weeven)$ if $C$ contains exactly one edge~$e$,
			 \item $L(C)= (\wpone, \wptwo)$ if $C$ is a path with at least three edges, or
			 \item $L(C)= (\wcone, \wcmtwo)$ if $C$ is a cycle and each popular matching in~$C$ has only one nice witness, or
			 \item $L(C)= (\wcone, \wcmtwo, \wcthree, \wctwo)$ otherwise.
			\end{itemize}

			\State Set $\witnessc$ to be the first witness of~$L (C)$ to $C$.
			\EndFor
			
			\State For each vertex~$v \in V$, set $\wof{v} := \witnesscof{v}$ where $C$ is the connected component of~$\hcp$ containing~$v$.
			\label{line:wpm-init-end}
			
			\While{there is a conflicting edge~$\{a, b\}$ for $\witness$}
			\label{line:wpm-while-start}
              \State Let $C_a$ be the connected component containing~$a$, and $C_b$ be the connected component containing~$b$.
              \label{line:find-con-comp}
              \If{the witness assigned to~$C_a$ is \odd}\label{line:odd-start}
              \State Dismiss the witness assigned to $C_a$ from~$L(C_a)$.\label{line:delete-odd}
              \If{$L(C_a) $ is empty}
                \State \textbf{return} \no\label{line:wpm-no-1}
              \EndIf
              \State Set $\wof{v} := \witness^{C_a}_{v}$ for every~$v \in V(C_a)$ where $\witness^{C_a}$ is the first witness of~$L (C_a)$.
              \label{line:odd-end}
              \Else
                \State Dismiss the witness assigned to~$C_b$ from~$L (C_b)$.
                \label{line:even-start}
              \If{$L(C_b) $ is empty}
                \State \textbf{return} \no\label{line:wpm-no-2}
              \EndIf
              \State Set $\wof{v} := \witness^{C_b}_{v}$ for every~$v \in V({C_b})$ where $\witness^{C_b}$ is the first witness of~$L (C_b)$.
              \label{line:even-end}
              \EndIf
			
			\EndWhile
			\label{line:wpm-while-end}
			\State \textbf{return} the matching~$M$ constructed by taking for each connected component~$C$ the popular matchings belonging to the witnesses assigned to~$C$.
			\label{line:wpm-return}
		\end{algorithmic}
	\end{algorithm}
}

    \begin{theorem}%[\appmark]
    \label{thm:c>3}
     If all vertices in $A$ have weight~$c$ for some $c > 3$ and every vertex in $B$ has weight~1, then a maximum-cardinality popular matching (if one exists) can be computed in $O (n + m)$ time.
    \end{theorem}
%\khcom{}%    We sketch the idea of the proof here; a more formal proof can be found in the appendix. 
    \begin{proof}
        \new{We claim that \Cref{alg:popular} computes a maximum-cardinality popular matching in linear time.
        We start by analyzing the running time of \Cref{alg:popular}.
        }
        The computation of~$\hcp$ takes $O(n +m)$ time due to \Cref{lem:b_two}.
        There are $O(n)$ connected components of $\hcp$;
        thus at most $O(n)$ partial witnesses get dismissed.
%        Searching for a conflicting edge can be done in $O(m)$ time.
        As we only need to check whether an edge is conflicting if the witness on one of the two connected component\new{s} incident to it changed, we need to check for each edge only a constant number of times whether it is conflicting.
        From this, the running time follows.
%        The algorithm clearly runs in polynomial time.

        To show the correctness of the algorithm, first observe that if a component is assigned its \even\ witness, then it is the only remaining witness that has not been dismissed.
        We will now show that whenever a witness~$\witness^{\operatorname{sub}}$ is dismissed \new{(which can happen in Line~\ref{line:delete-odd} or~\ref{line:even-start})}, then there is no popular matching together with a witness~$\witness$ such that $\wof{v}^{\operatorname{sub}} = \wof{v}$ for every $v\in V(C)$.
        At the beginning, the statement is clear since no witness has been dismissed.
        Consider the first witness~$\witness^{\operatorname{sub}}$ of a connected component~$C$ deleted by the algorithm that is a subwitness of a witness~$\witness$ of some matching, and let~$\{a, b\}$ be the conflicting edge because of which~$\witness^{\operatorname{sub}}$ was dismissed.
        \new{If $\witness^{\operatorname{sub}}$ was deleted in Line~\ref{line:delete-odd}, then}
        $a \in V(C)$, $\witness^{\operatorname{sub}}$ is \odd, and the current witness in the component~$C^b$ containing~$b$ weakly dominates all other non-dismissed witnesses on this component by \cref{lem:fourwitnesses}.
        Note that $\witness$ restricted to~$C^b$ is not dismissed by the choice of~$\witness^{\operatorname{sub}}$.
        Consequently, \Cref{lem:dominated-witnesses} implies that $\{a, b\}$ is conflicting for~$\witness$, a contradiction.

        \new{If $\witness^{\operatorname{sub}}$ was deleted in Line~\ref{line:even-start}, then }
        $b \in V(C)$, the witness in the component containing~$a$ is \even, and consequently the only non-dismissed witness~$\witness^{\operatorname{sub}}$ on this component.
        Therefore, by the definition of~$\witness^{\operatorname{sub}}$, every witness~$\witness$ of some popular matching coincides with~$\witness^{\operatorname{sub}}_a$ on~$a$, and therefore cannot coincide with~$\witness^{\operatorname{sub}}$ on $b$ (as otherwise the edge~$\{a, b\}$ would be conflicting for~$\witness$).
        %\hfill\qedhere\qedsymbol
        
        Finally, we show that there is no larger popular matching than the computed one.
        Let $C$ be a connected component which is a path and let~$M_1$ and $M_2$ be the two popular matchings on~$C$.
        Assume without loss of generality that $|M_1 | \le |M_2|$.
        Then at least one vertex~$v$ of~$C$ is unmatched by~$M_2$.
        This implies that every witness~$\wtwo$ of~$M_2$ has $\wtwoof{v} = 0$, implying that $\wtwo$ is \even.
%        For every connected component that is a path, the witness corresponding to the smaller matching is \even.
    Because we initially assign the \odd\ witness (and thus corresponding to the larger matching~$M_1$) to this connected component,
    the computed matching will contain the smaller matching only if no popular matching contains~$M_1$.
    Thus, the computed matching is of maximum cardinality among all popular matchings.
    \end{proof}
    
    \section{Future directions} 
    \label{sec:open}
    We discovered an unusual pattern in the complexity of \textsc{Popular Matching with Weighted Voters} in instances where the weights on one side are fixed at~1. A solution is guaranteed to exist and easy to find if the weight~$c$ of the other side equals~1. Then, for  $1 < c \le 2$, the problem becomes $\NP$-complete as \Cref{thm:np-hardb} shows. Polynomial-time solvability then returns for $3 < c$, but no-instances occur.
    A straightforward open question is whether \textsc{Popular Matching with Weighted Voters} can be solved in polynomial time when all vertices in~$A$ have weight $2 < c \le 3$ and all vertices in $B$ have weight~1.
    
    Furthermore, it would be interesting whether allowing each vertex from one side an individual weight larger than three (or some other constant~$c$) while all vertices from the other side still have weight one is also solvable in polynomial time.

    \new{\section*{Acknowledgements} The authors thank the anonymous reviewers of earlier versions of the paper, whose suggestions helped to improve the presentation.}

\newpage
\appendix

    \section{Example execution of the algorithm}
\label{sec:example}

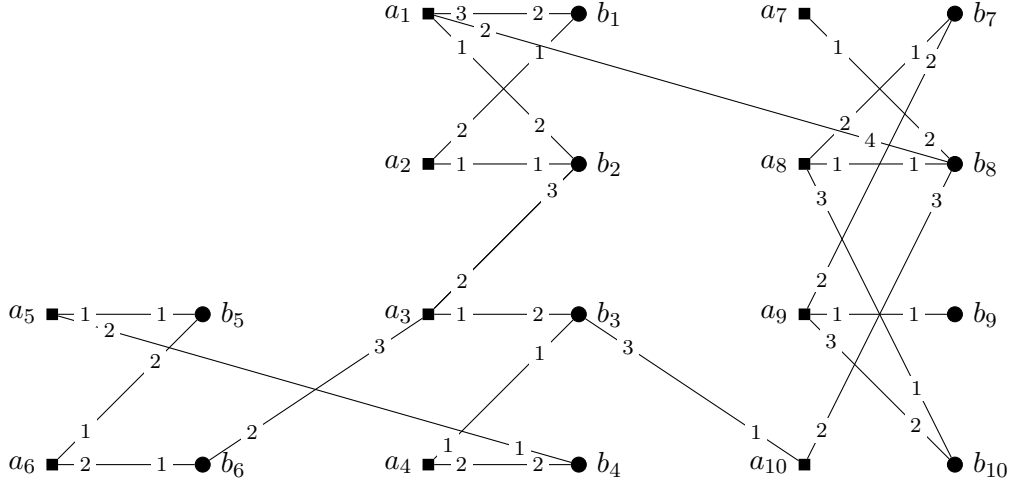
\begin{figure}
    \centering
\begin{tikzpicture}[yscale = 0.9]
          \node[terminal, label=180:$\eaf$] (a1) at (-1,1) {};
          \node[terminal, label=180:$\eafh$] (a1p) at (-1,3) {};
          \node[vertex, label={[xshift=-.0cm]0:$\ebfour$}] (b1) at (1.,1) {};
          \node[vertex, label=0:$\ebfh$] (b1p) at (1.,3) {};

        \node[terminal, label=180:$\eato$] (a2) at ($(a1) + (0, -4)$) {};
        \node[terminal, label=180:$\eatoh$] (a2p) at ($(a2) + (0, 2)$) {};
        \node[vertex, label=0:$\ebto$] (b2) at ($(a2) + (2, 0)$) {};
        \node[vertex, label={[xshift=-.0cm]0:$\ebtoh$}] (b2p) at ($(b2) + (0, 2)$) {};

        \node[terminal, label=180:$\eatt$] (a3) at ($(a2) + (-5, 0)$) {};
        \node[terminal, label=180:$\eatth$] (a3p) at ($(a3) + (0, 2)$) {};
        \node[vertex, label=0:$\ebtt$] (b3) at ($(a3) + (2, 0)$) {};
        \node[vertex, label=0:$\ebtth$] (b3p) at ($(b3) + (0, 2)$) {};

          \node[terminal, label=180:$\eap$] (ap) at (4,1) {};
          \node[terminal, label=180:$\eaph$] (app) at (4,3) {};
          \node[vertex, label={[xshift=-.0cm]0:$\ebp$}] (bp) at (6.,1) {};
          \node[vertex, label=0:$\ebph$] (bpp) at (6.,3) {};

          \node[terminal, label=180:$\eaedge$] (ae) at (4, -1) {};
          \node[vertex, label=0:$\ebe$] (be) at (6, -1) {};

          \node[terminal, label=180:$\eas$] (as) at (4, -3) {};
          \node[vertex, label=0:$\ebs$] (bs) at (6, -3) {};

        \draw (a1) edge[] node[pos=0.2, fill=white, inner sep=2pt] {\scriptsize $1$}  node[pos=0.76, fill=white, inner sep=2pt] {\scriptsize $1$} (b1);
        \draw (a1) edge[] node[pos=0.2, fill=white, inner sep=2pt] {\scriptsize $2$}  node[pos=0.76, fill=white, inner sep=2pt] {\scriptsize $1$} (b1p);

        \draw (a1p) edge[] node[pos=0.2, fill=white, inner sep=2pt] {\scriptsize $1$}  node[pos=0.76, fill=white, inner sep=2pt] {\scriptsize $2$} (b1);
        \draw (a1p) edge[] node[pos=0.2, fill=white, inner sep=2pt] {\scriptsize $3$}  node[pos=0.76, fill=white, inner sep=2pt] {\scriptsize $2$} (b1p);

        \draw (a2) edge[] node[pos=0.2, fill=white, inner sep=2pt] {\scriptsize $2$}  node[pos=0.76, fill=white, inner sep=2pt] {\scriptsize $2$} (b2);
        \draw (a2) edge[] node[pos=0.1, fill=white, inner sep=2pt] {\scriptsize $1$}  node[pos=0.76, fill=white, inner sep=2pt] {\scriptsize $1$} (b2p);

        \draw (a2p) edge[] node[pos=0.2, fill=white, inner sep=2pt] {\scriptsize $3$}  node[pos=0.8, fill=white, inner sep=2pt] {\scriptsize $2$} (b3);
        \draw (a2p) edge[] node[pos=0.2, fill=white, inner sep=2pt] {\scriptsize $1$}  node[pos=0.76, fill=white, inner sep=2pt] {\scriptsize $2$} (b2p);

        \draw (a3) edge[] node[pos=0.2, fill=white, inner sep=2pt] {\scriptsize $2$}  node[pos=0.76, fill=white, inner sep=2pt] {\scriptsize $1$} (b3);
        \draw (a3) edge[] node[pos=0.2, fill=white, inner sep=2pt] {\scriptsize $1$}  node[pos=0.7, fill=white, inner sep=2pt] {\scriptsize $2$} (b3p);

        \draw (a3p) edge[] node[pos=0.1, fill=white, inner sep=2pt] {\scriptsize $2$}  node[pos=0.9, fill=white, inner sep=2pt] {\scriptsize $1$} (b2);
        \draw (a3p) edge[] node[pos=0.2, fill=white, inner sep=2pt] {\scriptsize $1$}  node[pos=0.76, fill=white, inner sep=2pt] {\scriptsize $1$} (b3p);

        \draw (a2p) edge[] node[pos=0.2, fill=white, inner sep=2pt] {\scriptsize $2$}  node[pos=0.85, fill=white, inner sep=2pt] {\scriptsize $3$} (b1);
        
        \draw (a2p) edge[] node[pos=0.2, fill=white, inner sep=2pt] {\scriptsize $2$}  node[pos=0.85, fill=white, inner sep=2pt] {\scriptsize $3$} (b1);
        
        \draw (ap) edge[] node[pos=0.2, fill=white, inner sep=2pt] {\scriptsize $1$}  node[pos=0.76, fill=white, inner sep=2pt] {\scriptsize $1$} (bp);
        \draw (ap) edge[] node[pos=0.25, fill=white, inner sep=2pt] {\scriptsize $2$}  node[pos=0.76, fill=white, inner sep=2pt] {\scriptsize $1$} (bpp);

        \draw (app) edge[] node[pos=0.2, fill=white, inner sep=2pt] {\scriptsize $1$}  node[pos=0.86, fill=white, inner sep=2pt] {\scriptsize $2$} (bp);
        
        \draw (ae) edge[] node[pos=0.2, fill=white, inner sep=2pt] {\scriptsize $1$}  node[pos=0.76, fill=white, inner sep=2pt] {\scriptsize $1$} (be);
        
        \draw (ae) edge[] node[pos=0.15, fill=white, inner sep=2pt] {\scriptsize $3$}  node[pos=0.76, fill=white, inner sep=2pt] {\scriptsize $2$} (bs);
        
        \draw (as) edge[] node[pos=0.1, fill=white, inner sep=2pt] {\scriptsize $2$}  node[pos=0.9, fill=white, inner sep=2pt] {\scriptsize $3$} (bp);
        \draw (ap) edge[] node[pos=0.1, fill=white, inner sep=2pt] {\scriptsize $3$}  node[pos=0.76, fill=white, inner sep=2pt] {\scriptsize $1$} (bs);

        \draw (a1p) edge[] node[pos=0.1, fill=white, inner sep=2pt] {\scriptsize $2$}  node[pos=0.85, fill=white, inner sep=2pt] {\scriptsize $4$} (bp);
        
        \draw (ae) edge[] node[pos=0.1, fill=white, inner sep=2pt] {\scriptsize $2$}  node[pos=0.86, fill=white, inner sep=2pt] {\scriptsize $2$} (bpp);
        
        \draw (as) edge[] node[pos=0.2, fill=white, inner sep=2pt] {\scriptsize $1$}  node[pos=0.8, fill=white, inner sep=2pt] {\scriptsize $3$} (b2p);
\end{tikzpicture}
\caption{The input instance for our example execution of the algorithm.}
    \label{fig:example_instance}
\end{figure}
We now present an example execution for our algorithm from \Cref{thm:c>3}.
The input instance is depicted in \Cref{fig:example_instance}.
\subsection{Phase 1: Pruning edges}
In the first phase of the algorithm, see \Cref{sec:non-pop}, we prune non-popular edges in three steps, each of which results in the graph $H^{\deg (A) \le 2}$, $\hct$, and $\hcp$, respectively. \Cref{fig:degA,fig:hct,fig:hcp} depict these three graphs.
\begin{enumerate}
\item In the first step, we compute $f(v)$ and $s(v) $ for every vertex~$v\in V(G)$.
This results in graph~$H^{\deg (A) \le 2}$, depicted in \Cref{fig:degA}.
For example, for vertex $\eafh$, $f(\eafh) = \ebfh$, because $\ebfh$ is the first choice of $\eafh$. The second choice of $\eafh$ is $\ebp$, but since $\ebp = f(\eap)$, $s(\eafh) \neq \ebp$. Instead, $s(\eafh)= \ebfh$.
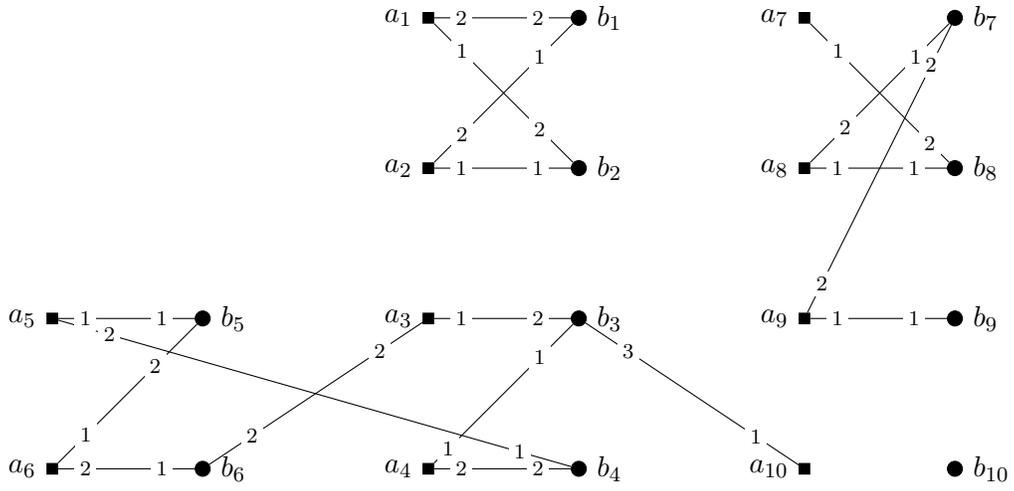
\begin{figure}
    \centering
\begin{tikzpicture}[yscale = 0.85]
          \node[terminal, label=180:$\eaf$] (a1) at (-1,1) {};
          \node[terminal, label=180:$\eafh$] (a1p) at (-1,3) {};
          \node[vertex, label={[xshift=-.0cm]0:$\ebfour$}] (b1) at (1.,1) {};
          \node[vertex, label=0:$\ebfh$] (b1p) at (1.,3) {};

        \node[terminal, label=180:$\eato$] (a2) at ($(a1) + (0, -3.8)$) {};
        \node[terminal, label=180:$\eatoh$] (a2p) at ($(a2) + (0, 2)$) {};
        \node[vertex, label=0:$\ebto$] (b2) at ($(a2) + (2, 0)$) {};
        \node[vertex, label={[xshift=-.0cm]0:$\ebtoh$}] (b2p) at ($(b2) + (0, 2)$) {};

        \node[terminal, label=180:$\eatt$] (a3) at ($(a2) + (-5, 0)$) {};
        \node[terminal, label=180:$\eatth$] (a3p) at ($(a3) + (0, 2)$) {};
        \node[vertex, label=0:$\ebtt$] (b3) at ($(a3) + (2, 0)$) {};
        \node[vertex, label=0:$\ebtth$] (b3p) at ($(b3) + (0, 2)$) {};

          \node[terminal, label=180:$\eap$] (ap) at (4,1) {};
          \node[terminal, label=180:$\eaph$] (app) at (4,3) {};
          \node[vertex, label={[xshift=-.0cm]0:$\ebp$}] (bp) at (6.,1) {};
          \node[vertex, label=0:$\ebph$] (bpp) at (6.,3) {};

          \node[terminal, label=180:$\eaedge$] (ae) at (4, -0.8) {};
          \node[vertex, label=0:$\ebe$] (be) at (6, -0.8) {};

          \node[terminal, label=180:$\eas$] (as) at (4, -2.8) {};
          \node[vertex, label=0:$\ebs$] (bs) at (6, -2.8) {};

        \draw (a1) edge[] node[pos=0.2, fill=white, inner sep=2pt] {\scriptsize $1$}  node[pos=0.76, fill=white, inner sep=2pt] {\scriptsize $1$} (b1);
        \draw (a1) edge[] node[pos=0.2, fill=white, inner sep=2pt] {\scriptsize $2$}  node[pos=0.76, fill=white, inner sep=2pt] {\scriptsize $1$} (b1p);

        \draw (a1p) edge[] node[pos=0.2, fill=white, inner sep=2pt] {\scriptsize $1$}  node[pos=0.76, fill=white, inner sep=2pt] {\scriptsize $2$} (b1);
        \draw (a1p) edge[] node[pos=0.2, fill=white, inner sep=2pt] {\scriptsize $2$}  node[pos=0.76, fill=white, inner sep=2pt] {\scriptsize $2$} (b1p);

        \draw (a2) edge[] node[pos=0.2, fill=white, inner sep=2pt] {\scriptsize $2$}  node[pos=0.76, fill=white, inner sep=2pt] {\scriptsize $2$} (b2);
        \draw (a2) edge[] node[pos=0.1, fill=white, inner sep=2pt] {\scriptsize $1$}  node[pos=0.76, fill=white, inner sep=2pt] {\scriptsize $1$} (b2p);

        \draw (a2p) edge[] node[pos=0.2, fill=white, inner sep=2pt] {\scriptsize $2$}  node[pos=0.8, fill=white, inner sep=2pt] {\scriptsize $2$} (b3);
        \draw (a2p) edge[] node[pos=0.2, fill=white, inner sep=2pt] {\scriptsize $1$}  node[pos=0.76, fill=white, inner sep=2pt] {\scriptsize $2$} (b2p);

        \draw (a3) edge[] node[pos=0.2, fill=white, inner sep=2pt] {\scriptsize $2$}  node[pos=0.76, fill=white, inner sep=2pt] {\scriptsize $1$} (b3);
        \draw (a3) edge[] node[pos=0.2, fill=white, inner sep=2pt] {\scriptsize $1$}  node[pos=0.7, fill=white, inner sep=2pt] {\scriptsize $2$} (b3p);

        \draw (a3p) edge[] node[pos=0.1, fill=white, inner sep=2pt] {\scriptsize $2$}  node[pos=0.9, fill=white, inner sep=2pt] {\scriptsize $1$} (b2);
        \draw (a3p) edge[] node[pos=0.2, fill=white, inner sep=2pt] {\scriptsize $1$}  node[pos=0.76, fill=white, inner sep=2pt] {\scriptsize $1$} (b3p);
                
        \draw (ap) edge[] node[pos=0.2, fill=white, inner sep=2pt] {\scriptsize $1$}  node[pos=0.76, fill=white, inner sep=2pt] {\scriptsize $1$} (bp);
        \draw (ap) edge[] node[pos=0.25, fill=white, inner sep=2pt] {\scriptsize $2$}  node[pos=0.76, fill=white, inner sep=2pt] {\scriptsize $1$} (bpp);

        \draw (app) edge[] node[pos=0.2, fill=white, inner sep=2pt] {\scriptsize $1$}  node[pos=0.86, fill=white, inner sep=2pt] {\scriptsize $2$} (bp);
        
        \draw (ae) edge[] node[pos=0.2, fill=white, inner sep=2pt] {\scriptsize $1$}  node[pos=0.76, fill=white, inner sep=2pt] {\scriptsize $1$} (be);
                
        \draw (ae) edge[] node[pos=0.1, fill=white, inner sep=2pt] {\scriptsize $2$}  node[pos=0.86, fill=white, inner sep=2pt] {\scriptsize $2$} (bpp);
        
        \draw (as) edge[] node[pos=0.2, fill=white, inner sep=2pt] {\scriptsize $1$}  node[pos=0.8, fill=white, inner sep=2pt] {\scriptsize $3$} (b2p);
\end{tikzpicture}
\caption{The graph~$H^{\deg (A)\le 2}$.}
    \label{fig:degA}
\end{figure}

\item In the second step, we delete every edge incident to a cycle, resulting in graph~$\hct$ (see \Cref{fig:hct}). There are two cycles in the graph~$H^{\deg (A) \le 2}$ in \Cref{fig:degA}, but only one of them is incident to an edge. This edge is $\{\eas, \ebtoh\}$.
\begin{figure}
    \centering
\begin{tikzpicture}[yscale = 0.85]
          \node[terminal, label=180:$\eaf$] (a1) at (-1,1) {};
          \node[terminal, label=180:$\eafh$] (a1p) at (-1,3) {};
          \node[vertex, label={[xshift=-.0cm]0:$\ebfour$}] (b1) at (1.,1) {};
          \node[vertex, label=0:$\ebfh$] (b1p) at (1.,3) {};

        \node[terminal, label=180:$\eato$] (a2) at ($(a1) + (0, -3.8)$) {};
        \node[terminal, label=180:$\eatoh$] (a2p) at ($(a2) + (0, 2)$) {};
        \node[vertex, label=0:$\ebto$] (b2) at ($(a2) + (2, 0)$) {};
        \node[vertex, label={[xshift=-.0cm]0:$\ebtoh$}] (b2p) at ($(b2) + (0, 2)$) {};

        \node[terminal, label=180:$\eatt$] (a3) at ($(a2) + (-5, 0)$) {};
        \node[terminal, label=180:$\eatth$] (a3p) at ($(a3) + (0, 2)$) {};
        \node[vertex, label=0:$\ebtt$] (b3) at ($(a3) + (2, 0)$) {};
        \node[vertex, label=0:$\ebtth$] (b3p) at ($(b3) + (0, 2)$) {};

          \node[terminal, label=180:$\eap$] (ap) at (4,1) {};
          \node[terminal, label=180:$\eaph$] (app) at (4,3) {};
          \node[vertex, label={[xshift=-.0cm]0:$\ebp$}] (bp) at (6.,1) {};
          \node[vertex, label=0:$\ebph$] (bpp) at (6.,3) {};

          \node[terminal, label=180:$\eaedge$] (ae) at (4, -0.8) {};
          \node[vertex, label=0:$\ebe$] (be) at (6, -0.8) {};

          \node[terminal, label=180:$\eas$] (as) at (4, -2.8) {};
          \node[vertex, label=0:$\ebs$] (bs) at (6, -2.8) {};

        \draw (a1) edge[] node[pos=0.2, fill=white, inner sep=2pt] {\scriptsize $1$}  node[pos=0.76, fill=white, inner sep=2pt] {\scriptsize $1$} (b1);
        \draw (a1) edge[] node[pos=0.2, fill=white, inner sep=2pt] {\scriptsize $2$}  node[pos=0.76, fill=white, inner sep=2pt] {\scriptsize $1$} (b1p);

        \draw (a1p) edge[] node[pos=0.2, fill=white, inner sep=2pt] {\scriptsize $1$}  node[pos=0.76, fill=white, inner sep=2pt] {\scriptsize $2$} (b1);
        \draw (a1p) edge[] node[pos=0.2, fill=white, inner sep=2pt] {\scriptsize $2$}  node[pos=0.76, fill=white, inner sep=2pt] {\scriptsize $2$} (b1p);

        \draw (a2) edge[] node[pos=0.2, fill=white, inner sep=2pt] {\scriptsize $2$}  node[pos=0.76, fill=white, inner sep=2pt] {\scriptsize $2$} (b2);
        \draw (a2) edge[] node[pos=0.1, fill=white, inner sep=2pt] {\scriptsize $1$}  node[pos=0.76, fill=white, inner sep=2pt] {\scriptsize $1$} (b2p);

        \draw (a2p) edge[] node[pos=0.2, fill=white, inner sep=2pt] {\scriptsize $2$}  node[pos=0.8, fill=white, inner sep=2pt] {\scriptsize $2$} (b3);
        \draw (a2p) edge[] node[pos=0.2, fill=white, inner sep=2pt] {\scriptsize $1$}  node[pos=0.76, fill=white, inner sep=2pt] {\scriptsize $2$} (b2p);

        \draw (a3) edge[] node[pos=0.2, fill=white, inner sep=2pt] {\scriptsize $2$}  node[pos=0.76, fill=white, inner sep=2pt] {\scriptsize $1$} (b3);
        \draw (a3) edge[] node[pos=0.2, fill=white, inner sep=2pt] {\scriptsize $1$}  node[pos=0.7, fill=white, inner sep=2pt] {\scriptsize $2$} (b3p);

        \draw (a3p) edge[] node[pos=0.1, fill=white, inner sep=2pt] {\scriptsize $2$}  node[pos=0.9, fill=white, inner sep=2pt] {\scriptsize $1$} (b2);
        \draw (a3p) edge[] node[pos=0.2, fill=white, inner sep=2pt] {\scriptsize $1$}  node[pos=0.76, fill=white, inner sep=2pt] {\scriptsize $1$} (b3p);
                
        \draw (ap) edge[] node[pos=0.2, fill=white, inner sep=2pt] {\scriptsize $1$}  node[pos=0.76, fill=white, inner sep=2pt] {\scriptsize $1$} (bp);
        \draw (ap) edge[] node[pos=0.25, fill=white, inner sep=2pt] {\scriptsize $2$}  node[pos=0.76, fill=white, inner sep=2pt] {\scriptsize $1$} (bpp);

        \draw (app) edge[] node[pos=0.2, fill=white, inner sep=2pt] {\scriptsize $1$}  node[pos=0.86, fill=white, inner sep=2pt] {\scriptsize $2$} (bp);
        
        \draw (ae) edge[] node[pos=0.2, fill=white, inner sep=2pt] {\scriptsize $1$}  node[pos=0.76, fill=white, inner sep=2pt] {\scriptsize $1$} (be);
        
        \draw (ae) edge[] node[pos=0.1, fill=white, inner sep=2pt] {\scriptsize $2$}  node[pos=0.86, fill=white, inner sep=2pt] {\scriptsize $2$} (bpp);
\end{tikzpicture}
\caption{The graph~$\hct$.}
    \label{fig:hct}
\end{figure}

\item In the third step, we consider each connected component, which is either a tree or a cycle, and exhaustively delete every edge not contained in a popular matching in this component. This results in the graph~$\hcp$, depicted in \Cref{fig:hcp}. In our example, only the component on vertices $\{\eaph,\eap,\eaedge,\ebph,\ebp,\ebe\}$ has such an edge, all other components have popular matchings that cover all edges.
\end{enumerate}

\begin{figure}
    \centering
\begin{tikzpicture}[yscale = 0.85]
          \node[terminal, label=180:$\eaf$] (a1) at (-1,1) {};
          \node[terminal, label=180:$\eafh$] (a1p) at (-1,3) {};
          \node[vertex, label={[xshift=-.0cm]0:$\ebfour$}] (b1) at (1.,1) {};
          \node[vertex, label=0:$\ebfh$] (b1p) at (1.,3) {};

        \node[terminal, label=180:$\eato$] (a2) at ($(a1) + (0, -3.5)$) {};
        \node[terminal, label=180:$\eatoh$] (a2p) at ($(a2) + (0, 2)$) {};
        \node[vertex, label=0:$\ebto$] (b2) at ($(a2) + (2, 0)$) {};
        \node[vertex, label={[xshift=-.0cm]0:$\ebtoh$}] (b2p) at ($(b2) + (0, 2)$) {};

        \node[terminal, label=180:$\eatt$] (a3) at ($(a2) + (-5, 0)$) {};
        \node[terminal, label=180:$\eatth$] (a3p) at ($(a3) + (0, 2)$) {};
        \node[vertex, label=0:$\ebtt$] (b3) at ($(a3) + (2, 0)$) {};
        \node[vertex, label=0:$\ebtth$] (b3p) at ($(b3) + (0, 2)$) {};

          \node[terminal, label=180:$\eap$] (ap) at (4,1) {};
          \node[terminal, label=180:$\eaph$] (app) at (4,3) {};
          \node[vertex, label={[xshift=-.0cm]0:$\ebp$}] (bp) at (6.,1) {};
          \node[vertex, label=0:$\ebph$] (bpp) at (6.,3) {};

          \node[terminal, label=180:$\eaedge$] (ae) at (4, -0.5) {};
          \node[vertex, label=0:$\ebe$] (be) at (6, -0.5) {};

          \node[terminal, label=180:$\eas$] (as) at (4, -2.5) {};
          \node[vertex, label=0:$\ebs$] (bs) at (6, -2.5) {};

        \draw (a1) edge[] node[pos=0.2, fill=white, inner sep=2pt] {\scriptsize $1$}  node[pos=0.76, fill=white, inner sep=2pt] {\scriptsize $1$} (b1);
        \draw (a1) edge[] node[pos=0.2, fill=white, inner sep=2pt] {\scriptsize $2$}  node[pos=0.76, fill=white, inner sep=2pt] {\scriptsize $1$} (b1p);

        \draw (a1p) edge[] node[pos=0.2, fill=white, inner sep=2pt] {\scriptsize $1$}  node[pos=0.76, fill=white, inner sep=2pt] {\scriptsize $2$} (b1);
        \draw (a1p) edge[] node[pos=0.2, fill=white, inner sep=2pt] {\scriptsize $2$}  node[pos=0.76, fill=white, inner sep=2pt] {\scriptsize $2$} (b1p);

        \draw (a2) edge[] node[pos=0.2, fill=white, inner sep=2pt] {\scriptsize $2$}  node[pos=0.76, fill=white, inner sep=2pt] {\scriptsize $2$} (b2);
        \draw (a2) edge[] node[pos=0.1, fill=white, inner sep=2pt] {\scriptsize $1$}  node[pos=0.76, fill=white, inner sep=2pt] {\scriptsize $1$} (b2p);

        \draw (a2p) edge[] node[pos=0.2, fill=white, inner sep=2pt] {\scriptsize $2$}  node[pos=0.8, fill=white, inner sep=2pt] {\scriptsize $2$} (b3);
        \draw (a2p) edge[] node[pos=0.2, fill=white, inner sep=2pt] {\scriptsize $1$}  node[pos=0.76, fill=white, inner sep=2pt] {\scriptsize $2$} (b2p);

        \draw (a3) edge[] node[pos=0.2, fill=white, inner sep=2pt] {\scriptsize $2$}  node[pos=0.76, fill=white, inner sep=2pt] {\scriptsize $1$} (b3);
        \draw (a3) edge[] node[pos=0.2, fill=white, inner sep=2pt] {\scriptsize $1$}  node[pos=0.7, fill=white, inner sep=2pt] {\scriptsize $2$} (b3p);

        \draw (a3p) edge[] node[pos=0.1, fill=white, inner sep=2pt] {\scriptsize $2$}  node[pos=0.9, fill=white, inner sep=2pt] {\scriptsize $1$} (b2);
        \draw (a3p) edge[] node[pos=0.2, fill=white, inner sep=2pt] {\scriptsize $1$}  node[pos=0.76, fill=white, inner sep=2pt] {\scriptsize $1$} (b3p);
                
        \draw (ap) edge[] node[pos=0.2, fill=white, inner sep=2pt] {\scriptsize $1$}  node[pos=0.76, fill=white, inner sep=2pt] {\scriptsize $1$} (bp);
        \draw (ap) edge[] node[pos=0.25, fill=white, inner sep=2pt] {\scriptsize $2$}  node[pos=0.76, fill=white, inner sep=2pt] {\scriptsize $1$} (bpp);

        \draw (app) edge[] node[pos=0.2, fill=white, inner sep=2pt] {\scriptsize $1$}  node[pos=0.86, fill=white, inner sep=2pt] {\scriptsize $2$} (bp);
        
        \draw (ae) edge[] node[pos=0.2, fill=white, inner sep=2pt] {\scriptsize $1$}  node[pos=0.76, fill=white, inner sep=2pt] {\scriptsize $1$} (be);
\end{tikzpicture}
\caption{The graph~$\hcp$.}
    \label{fig:hcp}
\end{figure}
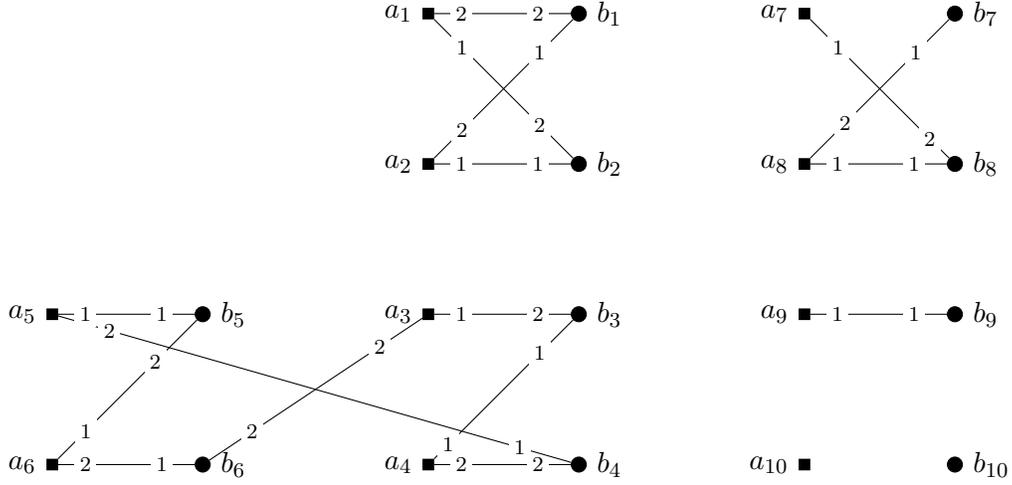
\subsection{Phase 2: Computing ``local'' witnesses}
In the second phase, we compute witnesses for each connected component of~$\hcp$, as described in \Cref{sec:witness}. These ``local'' witnesses are depicted in \Cref{fig:initial-witnesses}.

  \begin{table}
      \centering
      \begin{tabular}{c | c | c || c | c | c}
           Witness & Matching & Vertex & Value of $\witness$ & Vertex & Value of $\witness$\\
           \hline
           \multirow{2}*{$\wcoone$} & \multirow{2}*{$\{\{\eaf, \ebfour\}, \{\eafh, \ebfh\}\}$} &$\eaf$ & $-c$ & $\ebfour$ & $c$\\
           & & $\eafh$ & $-1$ & $\ebfh$ & $1$\\
           \hline
          \multirow{2}*{$\wcomtwo$} & \multirow{2}*{$\{\{\eaf, \ebfh\}, \{\eafh, \ebfour\}\}$} &$\eaf$ & $1$ & $\ebfour$ & $c$\\
           & & $\eafh$ & $-c$ & $\ebfh$ & $-1$\\
           \hline
           \multirow{2}*{$\wcothree$} & \multirow{2}*{$\{\{\eaf, \ebfour\}, \{\eafh, \ebfh\}\}$} &$\eaf$ & $2-c$ & $\ebfour$ & $c-2$\\
           & & $\eafh$ & $1$ & $\ebfh$ & $-1$\\
           \hline
           \multirow{2}*{$\wcoeven$} & \multirow{2}*{$\{\{\eaf, \ebfour\}, \{\eafh, \ebfh\}\}$} &$\eaf$ & $1-c$ & $\ebfour$ & $c -1$\\
           & & $\eafh$ & $0$ & $\ebfh$ & $0$\\
           \hline
           \hline
           \multirow{4}*{$\wctone$} & \multirow{4}*{$\{\{\eato, \ebtoh\}, \{\eatoh, \ebtt\}, \{\eatt, \ebtth\}, \{\eatth, \ebto\}\}$} &$\eato$ & $-c$ & $\ebto$ & $-1$\\
           & & $\eatoh$ & $-1$ & $\ebtoh$ & $c$\\
           & & $\eatt$ & $-c$ & $\ebtt$ & $1$\\
           & & $\eatth$ & $1$ & $\ebtth$ & $c$\\
           \hline
           \multirow{4}*{$\wctthree$} & \multirow{4}*{$\{\{\eato, \ebto\}, \{\eatoh, \ebtoh\}, \{\eatt, \ebtt\}, \{\eatth, \ebtth\}\}$} &$\eato$ & $1$ & $\ebto$ & $-1$\\
           & & $\eatoh$ & $-c$ & $\ebtoh$ & $c$\\
           & & $\eatt$ & $1$ & $\ebtt$ & $-1$\\
           & & $\eatth$ & $2-c$ & $\ebtth$ & $c-2$\\
           \hline
           \hline
           \multirow{2}*{$\wpone$} & \multirow{2}*{$\{\{\eap, \ebph\}, \{\eaph, \ebp\}\}$} &$\eap$ & $1$ & $\ebp$ & $c$\\
           & & $\eaph$ & $-c$ & $\ebph$ & $-1$\\
           \hline
           \multirow{2}*{$\wptwo$} & \multirow{2}*{$\{\{\eap, \ebp\}\}$} &$\eap$ & $1-c$ & $\ebp$ & $c - 1$\\
           & & $\eaph$ & $0$ & $\ebph$ & $0$\\
           \hline
           \hline
           $\weone$ & $\{\{\eaedge, \ebe\}\}$ & $\eaedge$ & $-c$ & $\ebe$ & $c$ \\
           \hline
           $\wetwo$ & $\{\{\eaedge, \ebe\}\}$ & $\eaedge$ & $2-c$ & $\ebe$ & $c -2 $ \\
           \hline
           $\weeven$ & $\{\{\eaedge, \ebe\}\}$ & $\eaedge$ & $1 -c$ & $\ebe$ & $c -1$ \\
           \hline
           \hline
           & $\emptyset $ & $\eas$ & 0 & & \\
           \hline
           \hline
           & $\emptyset $ & & & $\ebs $& 0\\
      \end{tabular}
      %\vspace{0.2cm}
      \caption{The set of witnesses for the connected components of~$\hcp$. The components are separated by a double horizontal line in the table.}
      \label{fig:initial-witnesses}
  \end{table}
  
\subsection{Phase 3: Constructing a ``global'' witness}
%Then the pruning procedure starts.
  We now turn to the final step of the algorithm, which prunes ``local'' witnesses step-by-step until we found a ``global'' witness.
  Throughout this phase, each connected component of~$\hcp$ is assigned a ``local'' witness.
  Initially, witnesses~$\wcoone$, $\wctone$, $\wpone$, and $\weone$ are assigned to the four non-trivial components of~$\hcp$.
  We now choose an arbitrary order of conflicting edges.
  We start with edge~$\{\eaedge, \ebs\}$.
  Because $\weone$ is \odd, we dismiss~$\weone$ and assign~$\wetwo$.
  Next, we choose edge~$\{\eas, \ebp\}$.
  Because the witness~$\bm{0}$ for $\eas$ is \even, we dismiss~$\wpone$ and assign~$\wptwo$.
  Now edge~$\{\eafh, \ebp\}$ becomes conflicting and we choose it.
  Because $\wptwo$ is \even, we dismiss~$\wcoone$ and assign $\wcomtwo$.
  Now edge~$\{\eatoh, \ebfour\}$ becomes conflicting and we choose it.
  Because~$\wctone$ is \odd, we dismiss witness~$\wctone$ and assign~$\wctthree$.
  Afterwards, there are no conflicting edges, and thus, we found the popular matching~$M=\{\{\eaf, \ebfh\}, \{\eafh, \ebfour\}, \{\eato, \ebtoh\}, \{\eatoh, \ebtt\}, \allowbreak\{\eatt, \ebtth\} , \allowbreak\{\eatth, \ebtoh\}, \{\eap, \ebp\}, \{\eaedge, \ebe\}\}$.

  \section{Example for non-dominating witnesses for non-unit weights for $B$}
  \label{apx:no-dominance}
  
  \begin{figure}
      \centering
        \begin{tikzpicture}
          \node[terminal, label=180:$a_{2}$, label={[xshift=-20pt]180:\color{green}$-c$, \color{red}$\frac{1}{3}$}] (a1) at (-1,1) {};
          \node[terminal, label=180:${a}_{1}$, label={[xshift=-20pt]180:\color{green}$-\frac{1}{3}$, \color{red}$-c$}] (a1p) at (-1,3) {};
          \node[white-vertex, label={[xshift=0pt]0:$b_{ 2}$}, label={[xshift=20pt]0:\color{green}$c$, \color{red}$c$}] (b1) at (1.,1) {};
          \node[vertex, label=0:${b}^*$, label={[xshift=20pt]0:\color{green}$\frac{1}{3}$, \color{red}$-1$}] (b1p) at (1.,3) {};

        \node[terminal, label=180:$a_{4}$, label={[xshift=-20pt]180:\color{green}$\frac{2}{3} -c$, \color{red}$1$}] (a2) at ($(a1) + (0, -4)$) {};
        \node[terminal, label=180:${a}_{3}$, label={[xshift=-20pt]180:\color{green}$\frac{1}{3}$, \color{red}$\frac{2}{3}- c$}] (a2p) at ($(a2) + (0, 2)$) {};
        \node[white-vertex, label=0:$b_{4}$, label={[xshift=20pt]0:\color{green}$c-\frac{2}{3}$, \color{red}$c - \frac{2}{3}$}] (b2) at ($(a2) + (2, 0)$) {};
        \node[white-vertex, label={[xshift=-.0cm]90:${b}_{3}$}, label={[xshift=20pt]0:\color{green}$-\frac{1}{3}$, \color{red}$-\frac{1}{3}$}] (b2p) at ($(b2) + (0, 2)$) {};

        \draw (a1) edge[ultra thick, green] node[pos=0.2, fill=white, inner sep=2pt] {\scriptsize $1$}  node[pos=0.76, fill=white, inner sep=2pt] {\scriptsize $1$} (b1);
        \draw (a1) edge[ultra thick, densely dotted, red] node[pos=0.2, fill=white, inner sep=2pt] {\scriptsize $2$}  node[pos=0.76, fill=white, inner sep=2pt] {\scriptsize $2$} (b2p);

        \draw (a1p) edge[ultra thick, green] node[pos=0.2, fill=white, inner sep=2pt] {\scriptsize $2$}  node[pos=0.76, fill=white, inner sep=2pt] {\scriptsize $2$} (b1p);
        \draw (a1p) edge[ultra thick, densely dotted, red] node[pos=0.2, fill=white, inner sep=2pt] {\scriptsize $1$}  node[pos=0.76, fill=white, inner sep=2pt] {\scriptsize $2$} (b1);

        \draw (a2) edge[ultra thick, green] node[pos=0.2, fill=white, inner sep=2pt] {\scriptsize $1$}  node[pos=0.76, fill=white, inner sep=2pt] {\scriptsize $1$} (b2);
        \draw (a2) edge[ultra thick, densely dotted, bend right = 52, red] node[pos=0.1, fill=white, inner sep=2pt] {\scriptsize $2$}  node[pos=0.9, fill=white, inner sep=2pt] {\scriptsize $1$} (b1p);

        \draw (a2p) edge[densely dotted, ultra thick, red] node[pos=0.2, fill=white, inner sep=2pt] {\scriptsize $2$}  node[pos=0.85, fill=white, inner sep=2pt] {\scriptsize $1$} (b2);
        \draw (a2p) edge[ultra thick, green] node[pos=0.2, fill=white, inner sep=2pt] {\scriptsize $1$}  node[pos=0.76, fill=white, inner sep=2pt] {\scriptsize $2$} (b2p);

  \end{tikzpicture}
      \caption{An example for an instance where each agent from~$A$ has weight $c > 3$ and each agent from~$B$ has weight at most 1 such that the witnesses are not in a dominance relation.
      Each squared vertex has weight $c > 3$, each black circular vertex has weight~1, and each white (circular) vertex has weight~$\frac{1}{3}$.
      }
      \label{fig:no-dominance}
  \end{figure}
  
  If we drop the assumption that vertices from~$B$ have weight~$1$ and allow them instead to have arbitrary weights not larger than 1, then \Cref{lem:fourwitnesses} does not hold any more:
  As an example, consider \Cref{fig:no-dominance}.
  There are two popular matchings (one containing the green edges, the other containing the red dotted edges), each having a unique witness (indicated by the green and red numbers).
  At vertex~$b^*$, neither witness dominates the other.




\bibliographystyle{elsarticle-num}
\bibliography{popular}
\end{document}

%% file: fancyProblem.tex
\usepackage{framed}
\usepackage{tabularx}

\newlength{\RoundedBoxWidth}
\newsavebox{\GrayRoundedBox}
\newenvironment{GrayBox}[1]%
   {\setlength{\RoundedBoxWidth}{.93\columnwidth}
    \def\boxheading{#1}
    \begin{lrbox}{\GrayRoundedBox}
       \begin{minipage}{\RoundedBoxWidth}}%
   {   \end{minipage}
    \end{lrbox}
    \begin{center}
    \begin{tikzpicture}%
       \node(Text)[draw=black!20,fill=white,rounded corners,inner sep=2ex,text width=\RoundedBoxWidth]
             {\usebox{\GrayRoundedBox}};
        \coordinate(x) at (current bounding box.north west);
        \node [draw=white,rectangle,inner sep=3pt,anchor=north west,fill=white]
        at ($(x)+(6pt,.75em)$) {\boxheading};
    \end{tikzpicture}
    \end{center}}

\newenvironment{defproblemx}[1]{\noindent\ignorespaces%
                                \FrameSep=6pt%
                                \parindent=0pt%
                \begin{GrayBox}{#1}%
                \begin{tabular*}{\columnwidth}{!{\extracolsep{\fill}}@{\hspace{.1em}} >{\itshape} p{1.5cm} p{0.86\columnwidth} @{}}%
            }{
                \end{tabular*}%
                \end{GrayBox}%
                \ignorespacesafterend
            }

\newcommand{\defProblemQuestion}[3]{%
  \begin{defproblemx}{#1}
    Input: & #2 \\
    Question: & #3
  \end{defproblemx}
}